\documentclass[pra,twocolumn,superscriptaddress,nofootinbib]{revtex4-1}
\usepackage[letterpaper, left=0.8in, right=0.8in, top=0.9in, bottom=0.9in]{geometry}
\usepackage{titlesec}

\usepackage{cmap} 
\usepackage[utf8]{inputenc}
\usepackage[english]{babel}
\usepackage[T1]{fontenc}
\usepackage{amsmath}
\usepackage{amsfonts}
\usepackage{enumerate}
\usepackage[mathscr]{euscript} 
\usepackage{bm}
\usepackage[table]{xcolor}
\definecolor{blueviolet}{rgb}{0.2, 0.2, 0.6}
\definecolor{webgreen}{rgb}{0,.5,0}
\definecolor{webbrown}{rgb}{.6,0,0}
\usepackage[pdftex,
	bookmarks=false,
	colorlinks=true, 
	urlcolor=webbrown, 
	linkcolor=blueviolet, 
	citecolor=webgreen,
	pdfstartpage=1,
	pdfstartview={FitH},  
	bookmarksopen=false
	]{hyperref}
\usepackage{tikz}
\usepackage{natbib}
\usepackage{mathtools}
\usepackage{graphicx}
\usepackage{csquotes}
\usepackage{braket}
\usepackage{dsfont}
\usepackage{listings}
\allowdisplaybreaks
\DeclareMathOperator{\Expect}{\mathbb{E}}
\usepackage{dsfont}
\usepackage{changepage} 

\usepackage{multirow}
\usepackage{amssymb}

\usepackage{amsthm}

\usepackage{color}
\usepackage{nicefrac}

\DeclareFixedFont{\ttb}{T1}{txtt}{bx}{n}{9} 
\DeclareFixedFont{\ttm}{T1}{txtt}{m}{n}{9}  

\usepackage{color}
\definecolor{deepblue}{rgb}{0,0,0.5}
\definecolor{deepred}{rgb}{0.6,0,0}
\definecolor{deepgreen}{rgb}{0,0.5,0}

\newcommand\pythonstyle{\lstset{
language=Python,
basicstyle=\ttm,
morekeywords={self},              
keywordstyle=\ttb\color{deepblue},
emph={MyClass,__init__},          
emphstyle=\ttb\color{deepred},    
stringstyle=\color{deepgreen},
frame=tb,                         
showstringspaces=false
}}

\lstnewenvironment{python}[1][]
{
\pythonstyle
\lstset{#1}
}
{}

\newcommand\pythoninline[1]{{\pythonstyle\lstinline!#1!}}

\definecolor{orange}{RGB}{255,127,0}

\def\bra#1{\ensuremath{\mathinner{\langle{#1}|}}}
\def\ket#1{\ensuremath{\mathinner{|{#1}\rangle}}}

\newcommand{\ketbra}[2]{\lvert #1 \rangle \! \langle #2 \rvert}

\newcommand{\expval}[1]{\langle #1\rangle}
\newcommand{\tr}{\text{tr}}
\newcommand{\Tr}{\text{tr}}

\newcommand{\V}{\boldsymbol{\mathscr{V}}}

\newtheorem{proposition}{Proposition}
\newtheorem{observation}{Observation}
\newtheorem{lemma}{Lemma}

\makeatletter
\newtheorem*{rep@proposition}{\rep@title}
\newcommand{\newrepproposition}[2]{%
\newenvironment{rep#1}[1]{%
 \def\rep@title{#2 \ref{##1}}%
 \begin{rep@proposition}}%
 {\end{rep@proposition}}}
\makeatother
\newrepproposition{proposition}{Proposition}

\makeatletter
\newtheorem*{rep@theorem}{\rep@title}
\newcommand{\newreptheorem}[2]{%
\newenvironment{rep#1}[1]{%
 \def\rep@title{#2 \ref{##1}}%
 \begin{rep@theorem}}%
 {\end{rep@theorem}}}
\makeatother
\newreptheorem{theorem}{Theorem}

\makeatletter
\newtheorem*{rep@definition}{\rep@title}
\newcommand{\newrepdefinition}[2]{%
\newenvironment{rep#1}[1]{%
 \def\rep@title{#2 \ref{##1}}%
 \begin{rep@definition}}%
 {\end{rep@definition}}}
\makeatother
\newrepproposition{definition}{Definition}

\newtheorem{theorem}{Theorem}
\newtheorem{corollary}{Corollary}
\newtheorem{definition}{Definition}

\usepackage[noend]{algpseudocode}
\usepackage{algorithm,algorithmicx}

\algrenewcommand\alglinenumber[1]{\sf\scriptsize\color{black}{#1}}
\algrenewcommand\algorithmicrequire{\textbf{Input:}}
\algrenewcommand\algorithmicensure{\textbf{Output:}}

\usepackage[toc,page,header]{appendix}
\usepackage{minitoc}

\newif\ifptitle
\newif\ifpnumber
\newcounter{para}

\ptitletrue  
\pnumbertrue  

\begin{document}
\doparttoc 
\faketableofcontents 


\title{Hardware-efficient learning of quantum many-body states}
\date{\today}
\author{Katherine Van Kirk}
\affiliation{Department of Physics, Harvard University, Cambridge, MA 02138, USA}
\author{Jordan Cotler}
\affiliation{Department of Physics, Harvard University, Cambridge, MA 02138, USA}
\author{Hsin-Yuan Huang}
\affiliation{Institute for Quantum Information and Matter and
Department of Computing and Mathematical Sciences, Caltech, Pasadena, CA, USA}
\author{Mikhail D. Lukin}
\email{lukin@physics.harvard.edu}
\affiliation{Department of Physics, Harvard University, Cambridge, MA 02138, USA}

\begin{abstract}
Efficient characterization of highly entangled multi-particle systems is an outstanding challenge in quantum science.
Recent developments have shown that a modest number of randomized measurements suffices to learn many properties of a quantum many-body system. However, implementing such measurements requires complete control over individual particles, which is unavailable in many experimental platforms. 
In this work, we present rigorous and efficient algorithms for learning quantum many-body states in systems with any degree of control over individual particles, including when every particle is subject to the same global field and no additional ancilla particles are available.
We numerically demonstrate the effectiveness of our algorithms for estimating energy densities in a $U(1)$ lattice gauge theory and classifying topological order using very limited measurement capabilities.
\end{abstract}

\maketitle


\textit{Introduction.} Experimental measurement capabilities and system controls determine what we can learn about nature.
Such constraints pose a particular challenge for studying quantum systems with many degrees of freedom, where the information accessible via any measurement is restricted by the uncertainty principle. 
Recently proposed learning protocols for quantum systems assume universal experimental capabilities that enable informationally-complete measurements~\cite{elben2022randomized, huang2020predicting, jiang2020optimal, stricker2022experimental, mcginley2022shadow, tran2022measuring}.
These ``randomized measurement'' protocols repeatedly prepare a quantum state $\rho$ and perform a measurement randomly sampled from a fixed ensemble. These randomized measurements can be achieved using random unitary evolutions \cite{elben2022randomized, huang2020predicting} or entangling with ancilla qubits and then measuring those ancillas \cite{jiang2020optimal, stricker2022experimental, mcginley2022shadow, tran2022measuring}.
By reusing and classically post-processing the measurement outcomes, one can learn many properties of the quantum state $\rho$.
Recent work demonstrated that judicious post-processing of the measurement data can reconstruct expectation values of an exponentially large set of non-commuting observables \cite{cotler2020quantum, huang2020predicting}. For certain observables, this can be done with exponentially fewer measurements than full quantum state tomography \cite{flammia2012quantum,haah2016sample,o2016efficient}.
For example, one such protocol, classical shadow tomography,
constructs a classical ``shadow'' of the full density matrix \cite{huang2020predicting} without any prior assumptions about the state. 
This protocol and related ideas were inspired by Aaronson's insight that partial descriptions of a quantum state suffice to predict a large number of its properties~\cite{aaronson2018shadow,aaronson2019gentle}.

\begin{figure}[t]
\centering
\includegraphics[scale = 0.32]{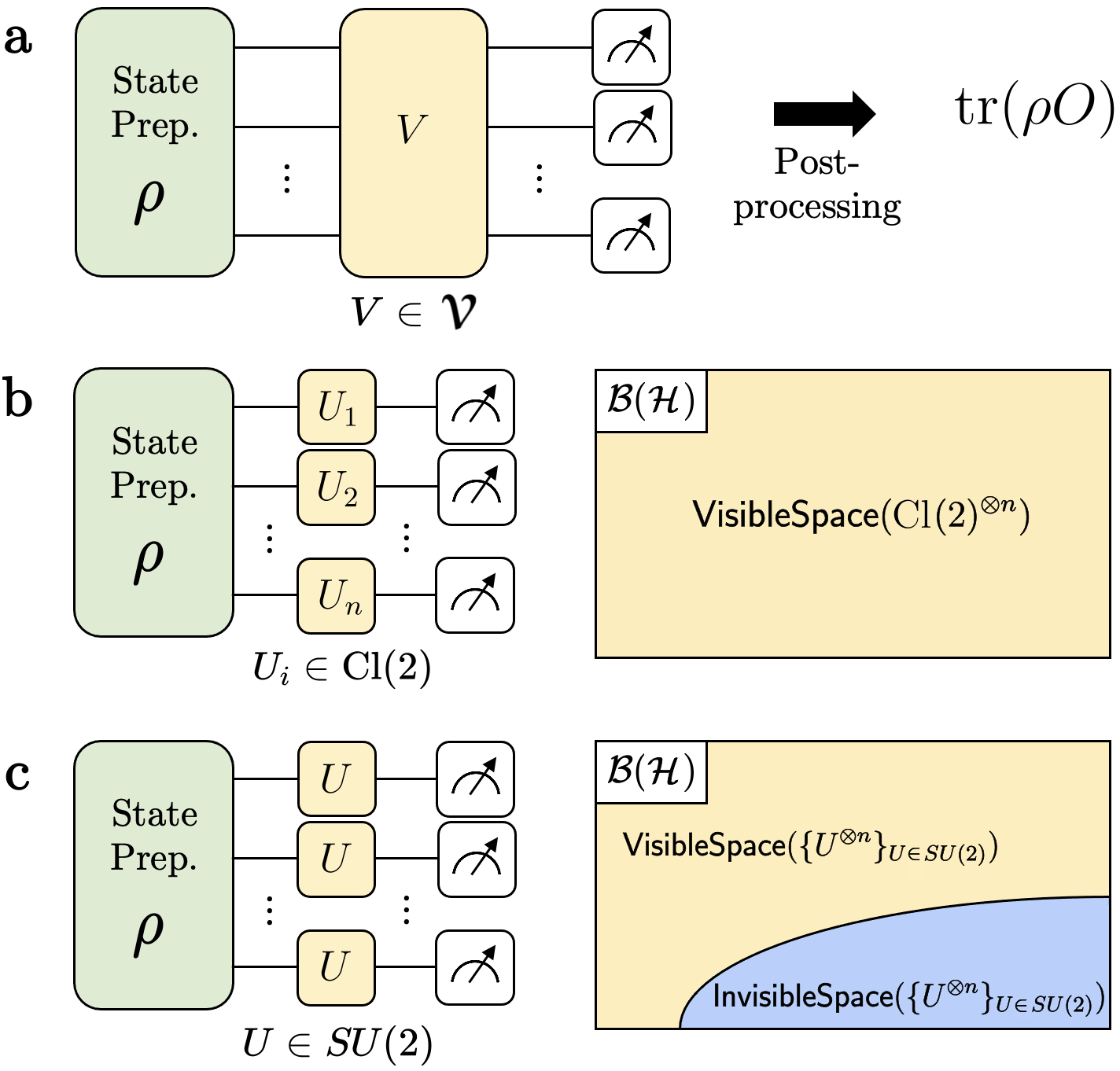}
\caption{\emph{Estimating observables in the visible space.} The visible space $\textnormal{\textsf{VisibleSpace}}$($\V$) is the subspace of the operator space $\mathcal{B}(\mathcal{H})$ containing all observables which can be estimated using the set of implementable unitaries $\V = \{V\}$. \,\,(a) The data obtained from randomized measurements, namely by evolving under $V \in \V$ and measuring in the computational basis, can be used to estimate $\tr(\rho O)$ for any $O \in \textnormal{\textsf{VisibleSpace}}(\V)$. When sampling from (b) $\V = \text{Cl}(2)^{\otimes n}$, the visible space contains all observables, but when sampling from (c) $\V = \{U^{\otimes n}\}_{U\in SU(2)}$, the visible space does not contain all observables -- though it still contains exponentially many.}
\label{fig:Figure1}
\end{figure}

In this Letter we generalize the randomized measurement paradigm. We develop an expanded framework that formalizes which properties of a quantum many-body system can be learned with \textit{specified} measurement capabilities.  
Leveraging available experimental controls, our framework allows one to learn exponentially many properties of a quantum state using only polynomially many measurements in the system size.  This massive multiplexing encompasses and generalizes the theory of classical shadow tomography~\cite{cotler2020quantum, huang2020predicting,huang2021efficient, huang2022learning,chen2021robust,koh2022classical,garcia2021quantum,levy2021classical,hadfield2022measurements,hu2021classical, hu2022logical, hu2022hamiltonian, akhtar2022scalable, bertoni2022shallow,gandhari2022continuous,becker2022classical,mcginley2022shadow,tran2022measuring, rath2021importance,stricker2022experimental}, rendering practical  protocols for contemporary quantum many-body simulators, near term quantum computers, and other controlled many-body systems \cite{semeghini2021probing,satzinger2021realizing,bluvstein2021controlling,martin2022controlling, stas2022robust, choi2019probing}. 
Moreover, by recasting the randomized measurement paradigm in the language of learning theory~\cite{mohri2012foundations}, we improve the efficiency of measuring a set of potentially non-commuting observables $\{O_1, O_2,..., O_M\}$ to within a fixed error $\epsilon$. 

Our key results can be understood as follows: 
consider an experimental platform capable of implementing any unitary in a set $\V = \{V\}$ and then subsequently measuring in the computational basis $\{|b\rangle \langle b|\}_{b \in \{0,1\}^n}$\ (see Figure~\ref{fig:Figure1}).  
With these restrictions the only accessible expectation values are $\langle b |V \rho V^\dagger |b\rangle$, assuming that the state $\rho$ can be prepared on the device. 
Since any observable of the form $V^\dagger |b\rangle \langle b| V$ can be measured, we define the space of \emph{visible observables} as all linear combinations of these $V^\dagger |b\rangle \langle b| V$. 
We can estimate any observable $O_i$ as long as it lives in the visible space. 

Focusing on the visible observables of interest $O_1, ... , O_M$,  we tailor our scheme's randomized measurements and classical post-processing to most efficiently estimate these observables.
For our randomized measurements, we determine which unitaries in $\V$ help estimate our observables of interest. Specifically, since  each $O_i$ can be expressed as a linear combination of $V^\dagger \ketbra{b}{b} V$, the most helpful unitaries $V$ contribute more to the linear combination.  In our measurement protocol, we choose to sample these helpful unitaries with higher probability, thereby improving the measurement efficiency. 
Subsequently, in the post-processing, we can bias our estimator for $\Tr(\rho O_i)$ to reduce its error for a fixed number of measurements. 
This process utilizes a fundamental concept in learning theory, the bias-variance tradeoff, to find the optimal estimator. The bias-variance tradeoff is analogous to the tradeoff in statistics of overfitting versus underfitting a line to a set of data points.
By tailoring our measurements and post-processing to our observables of interest, we maximize the utility of each measurement, 
hence greatly reducing the sample complexity as compared with existing protocols.
In Theorem~\ref{thm:shadowtomtheorem} we provide rigorous guarantees on the number of measurements required to estimate all $\tr(\rho O_i)$ within precision $\epsilon$ with high probability. 
Our scheme’s multiplexed estimation of $M$ observables recapitulates and generalizes the favorable scaling of classical 
shadows, requiring only $\mathcal{O}(\log M)$ measurements. 

 
Our techniques make the observables of interest easy to estimate within the constraints of a given quantum device. These results challenge the ideology of \emph{measure first, ask questions later} \cite{elben2022randomized}. Indeed, in practice the observables of interest are known prior to the experimental runs.  Therefore, it is natural to measure only
what is necessary to estimate these observables and to find the lowest-error estimator for each. We show that this approach significantly improves upon the sample complexity of current state-of-the-art randomized measurements protocols like classical shadows.
Furthermore, we provide several specific examples demonstrating the utility of our framework and methods.  
For instance, we show that \emph{global $SU(2)$ control}, the ability to make each qubit undergo the same Bloch sphere rotation, suffices to efficiently measure the energy density of a $U(1)$ lattice gauge theory.  Additionally, we show that when global $SU(2)$ control is interfaced with rigorous machine learning algorithms \cite{huang2022provably}, it is possible to distinguish certain trivial and topological phases of a quantum many-body system.

\vspace{1mm}
\textit{Learnable properties with specified controls.}
In quantum experiments with limited control, only certain observables can be estimated. 
We can estimate the expectation value $\text{tr}(\rho O)$ of an observable $O$  if and only if $O$ is in the span of $\{V^\dagger \ketbra{b}{b}V\}_{V,b}$. 
We call this span the ``visible space'' because it determines what properties of the quantum system can be learned.  
More formally:
\begin{definition}\label{def:VisibleSpace}
     $\textnormal{\textsf{VisibleSpace}}(\V)$ is the span of $V^{\dagger} \ketbra{b}{b} V$ over every $V \in \V$ and every $b \in \{0,1\}^n$.
\end{definition}

The control provided by an experimental platform, as parameterized by the set of implementable unitaries $\V$, is ingrained in the size and structure of the corresponding $\textnormal{\textsf{VisibleSpace}}(\V)$.
Experimental platforms with complete control, e.g.~universal quantum computers, have a visible space that contains all observables. For instance, it suffices for $\V = \text{Cl}(2)^{\otimes n}$, where $\text{Cl}(2)$ is the single-qubit Clifford group, because with these unitaries one can measure all Pauli strings (see Figure~\ref{fig:Figure1}(b)).  
By contrast, a common form of limited control in contemporary experiments is $\V = \{U^{\otimes n}\}_{U \in SU(2)}$, which we call ``global $SU(2)$ control'' (see Figure~\ref{fig:Figure1}(c)).  This will be a guiding example in our applications below. 
In Appendix~\ref{app:globalSU2} we provide an explicit orthonormal basis of operators for the global $SU(2)$ visible space, and establish that the dimension of this visible space is $\sim 2^n n^2$. Conveniently, this grows exponentially with $n$, meaning that global $SU(2)$ control still enables us to learn exponentially many properties of a quantum system.

\vspace{1mm}
\textit{Efficient protocol for learning many properties.} 
Suppose our quantum simulator can implement the unitaries $\V$. 
Given this type of experimental control, we present a randomized measurement framework to learn properties of a state $\rho$. 
We will first present the randomized measurement protocol and then discuss how to modify it to maximize the effectiveness of each measurement.
Each randomized measurement takes the form:
\begin{enumerate}
    \item Prepare $\rho$.
    \item Sample a unitary $V$ according to some probability density function $p(V)$ on $\V$.
    \item Apply $V$ to $\rho$, obtaining $V \rho V^\dagger$.
    \item Measure $V \rho V^\dagger$ in the computational basis to obtain $|b\rangle$ with probability $\langle b | V \rho V^\dagger | b\rangle$, and record the outcome $b \in \{0, 1\}^n$.
\end{enumerate}
When performing $N$ of these measurements, we package the experimental data as $\{(V_s, b_s)\}_{s=1}^N$, where $V_s$ is the $V$ sampled in the $s$th round, and $b_s$ is the outcome of the measurement in that same round.

Given an observable $O$, we can utilize this experimental data to predict $\Tr(\rho O)$. 
First, notice that the operators we measure in our randomized measurements, $\{V^\dagger \ketbra{b}{b}V\}_{V,b}$, form an overcomplete basis for the visible space, and so we can express any $\widetilde{O} \in \textsf{VisibleSpace}(\V)$ as a linear combination of $V^\dagger \ketbra{b}{b}V\,$s. 
For the coefficients of this linear combination, we want to find a function $\mathcal{K}^{\widetilde{O}}$, mapping a unitary $V$ and an $n$-bit string $b$ to a real number, where
\begin{equation}\label{eqn:OvisKVb}
    \widetilde{O} = \int_{V \in \V} dV \, p(V) \sum_{b \in \{0, 1\}^n} \mathcal{K}^{\widetilde{O}}(V,b) \hspace{1mm} V^{\dagger} \ketbra{b}{b} V
\end{equation}
and $\widetilde{O} \approx O$.
The function $\mathcal{K}^{\widetilde{O}}$ is analogous to a Wigner function, which  for  a harmonic oscillator is  defined on the overcomplete set of coherent states and represents how a state is distributed in phase space. Similarly, the $\mathcal{K}^{\widetilde{O}}$ function allows us to express $\widetilde{O}$ in terms of the overcomplete $V^\dagger \ketbra{b}{b}V\,$ operators and captures how $\widetilde{O}$ is distributed in the visible space.


Using the function $\mathcal{K}^{\widetilde{O}}$, we can predict $\tr(\rho O)$.  Since $\widetilde{O} \approx O$, we can estimate $\tr(\rho O)$ by estimating $\tr(\rho \widetilde{O})$.
For an infinite number of randomized measurements, Eq.~\eqref{eqn:OvisKVb} implies
\begin{equation}\label{eqn:fullavg}
    \tr(\rho \widetilde{O}) = \lim_{N \rightarrow \infty} \frac{1}{N} \sum_{s=1}^N \mathcal{K}^{\widetilde{O}}(V_s, b_s)\,
\end{equation} 
by probabilistic analysis.
For a finite number of randomized measurements $N$, we can estimate $\tr(\rho \widetilde{O})$ with
\begin{equation}\label{eqn:empiricalavg}
    \tr(\rho \widetilde{O}) \approx
    \frac{1}{N} \sum_{s=1}^N \mathcal{K}^{\widetilde{O}}(V_s, b_s)\,.
\end{equation} 
As long as the function $\mathcal{K}^{\widetilde{O}}(V_s,b_s) $ does not fluctuate greatly over different data points $\{ (V_s, b_s) \}_{s=1}^N$, our protocol provides a good estimate for $\Tr(\rho \widetilde{O})$.
The statistical fluctuation of $\mathcal{K}^{\widetilde{O}}$ is characterized formally by the variance $\textnormal{Var}_{\rho}[\mathcal{K}^{\widetilde{O}}]$, which is a function of the unknown state $\rho$ (see Appendix \ref{app:VisibleSpace}).
As we do not have prior information about the state $\rho$, we define the quantity
\begin{align}
\textnormal{Var}_{\textnormal{max}}[\mathcal{K}^{\widetilde{O}}] := \max_{\rho} \textnormal{Var}_{\rho}[\mathcal{K}^{\widetilde{O}}]
\end{align}
which upper bounds 
$\textnormal{Var}_{\rho}[\mathcal{K}^{\widetilde{O}}]$ over all states.
This will be used in the next section to quantify the number of measurements required. 

Since we define $\widetilde{O}$ with an overcomplete basis, many functions $\mathcal{K}^{\widetilde{O}}$ may satisfy Eq.~\eqref{eqn:OvisKVb}. 
While one can choose any such $\mathcal{K}^{\widetilde{O}}$, certain choices reduce the number of measurements required for the protocol. 
For example, we show in Appendix~\ref{app:VisibleSpace} that when our state of interest is maximally mixed, the classical shadow version of $\mathcal{K}^{\widetilde{O}}$ requires the fewest measurements.
We recover classical shadow tomography when choosing
\begin{equation}\label{eqn:shadowVersionKVb}
\mathcal{K}_{\text{CS}}^{\widetilde{O}}(V, b) = \Tr(\widetilde{O} \hat{\rho}_{V,b})\,,
\end{equation}
where $\hat{\rho}_{V,b}$ is the single-shot classical shadow defined in \cite{huang2022learning}.

\vspace{1mm}
\textit{Number of measurements.} 
The following Theorem \ref{thm:shadowtomtheorem} provides a rigorous guarantee of precision in terms of the number of measurements. In this section we develop intuition for this guarantee and then discuss how to modify our protocol to reduce the required number of measurements. 
\begin{theorem}
\label{thm:shadowtomtheorem}
Suppose we are given a quantum device with system controls $\V$ and $M$ observables $O_1, \ldots, O_M$. Consider the functions $\mathcal{K}^{\widetilde{O}_1}, \ldots, \mathcal{K}^{\widetilde{O}_M}$ where all $\widetilde{O}_i \in \textsf{VisibileSpace}(\V)$.
We can predict all $\text{\rm tr}(\rho O_i)$ up to an error $\|O_i-\widetilde{O}_i\|_{\infty} + \epsilon$ using
\begin{equation}\label{eqn:thm1nummmets}
    N = \mathcal{O}\left(\max_{i}\!\left[ \textnormal{Var}_{\textnormal{max}}[\mathcal{K}^{\widetilde{O}_i}] \!\right] \log (M) / \epsilon^2\right)
\end{equation}
randomized measurements.
\end{theorem}

\noindent 
Theorem \ref{thm:shadowtomtheorem} generalizes one of the main results of classical shadow tomography~\cite{huang2020predicting}, suggesting that we can efficiently estimate the expectation values of exponentially many observables using only a modest number of measurements. 
For example, suppose we want to estimate the $2$-point correlation function $X_i X_j+Y_i Y_j+Z_i Z_j$ across all pairs $i,j$ of qubits. Since we need to estimate this observable across all $M = n^2$ pairs, by Theorem~\ref{thm:shadowtomtheorem} the number of total measurements scales as $\mathcal{O}(\log(n))$. 
While this result is known for quantum systems with complete control~\cite{huang2020predicting}, e.g.~the ability to measure any Pauli string, we find it also holds for simulators with global $SU(2)$ control (see Appendix \ref{app:globalSU2}). In fact, our work extends this result for \textit{any} form of control $\V$ as long as $X_i X_j+Y_i Y_j+Z_i Z_j \in \textsf{VisibleSpace}(\V)$. 

\begin{figure*}[t]
\centering
\includegraphics[scale = 0.35]{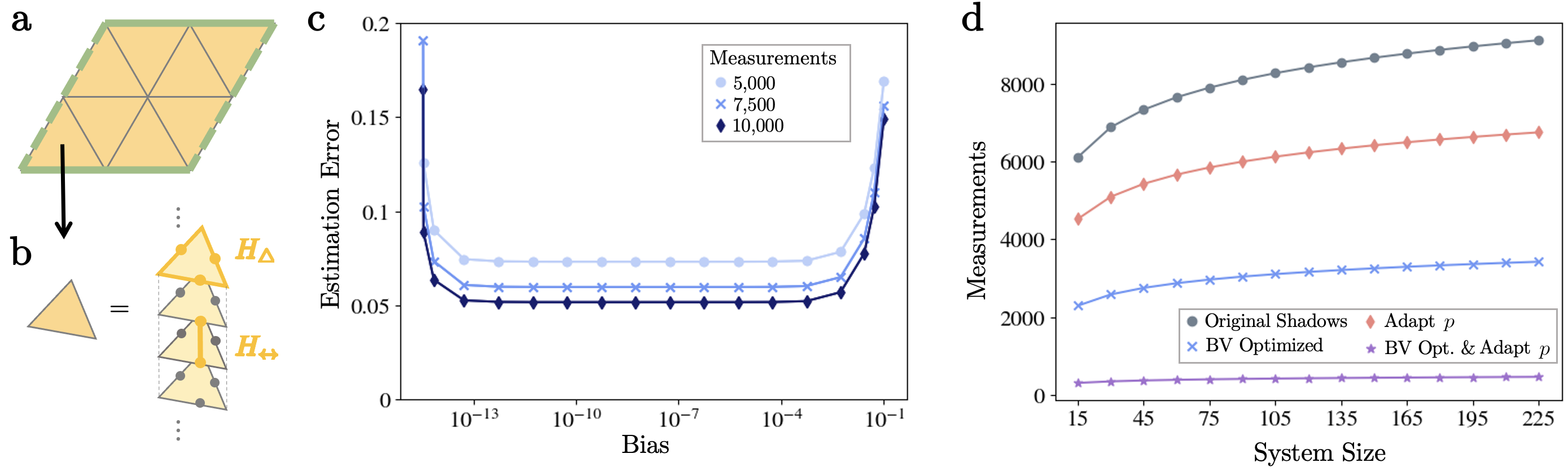}
\caption{\emph{Efficient energy density estimation of a $2+1$ $U(1)$ lattice gauge theory.} (a) The lattice lives in two spatial dimensions, tiled by triangular plaquettes with periodic boundary conditions (connecting the respective solid and dashed green lines). (b) Each spatial triangle extends in an extra dimension to enumerate over bosonic modes. The total Hamiltonian across all spatial triangles can be split into two types of terms: individual triangular plaquette terms $H_{\triangle}$ and link terms $H_{\leftrightarrow}$ that couples qubits in adjacent stacked plaquettes. (c) The bias-variance tradeoff when estimating $H_{\leftrightarrow}$. Biasing an operator can reduce its variance and, therefore, the overall estimation error (left side of bowl). Eventually, the bias will outweigh the gains made on the variance, and the estimation error will begin increasing (right side of bowl).  (d) Compared to classical shadows, our strategies reduce the number of measurements required to estimate the energy density to precision $\epsilon = 0.1$. Here we scale the system size by scaling the number of spatial triangles in our lattice.}
\label{fig:Figure2}
\end{figure*}

Theorem  1 also naturally gives rise to techniques for reducing the required number of measurements. We can make judicious choices to reduce the variances, which also reduces the number of measurements (Eq.~\eqref{eqn:thm1nummmets}). 
Suppose 
we wish to estimate the expectation value of the $k$-local operator $X^{\otimes k}$. 
Ideally one would measure all qubits in the $X$-basis, and $N$ measurements would bring the estimate to within precision $\mathcal{O}(1/\sqrt{N})$. 
However, when performing randomized measurements, not every experimental data point corresponds to measuring $X$ on all sites. The variance in Theorem~\ref{thm:shadowtomtheorem} quantifies how many randomized measurements are needed to have measured \textit{once} in the all-$X$ basis.
For instance, imagine sampling our randomized measurements from $\V=\text{Cl}(2)^{\otimes n}$. Under this set of unitaries, there are $3^k$ possible Pauli strings we could measure on the $k$ sites we are interested in, and only one of them is the all-$X$ Pauli string. This is reflected in the variance, which scales as $\mathcal{O}(3^k)$ when $\mathcal{K} = \mathcal{K}_\textnormal{CS}$ \cite{huang2020predicting}.
By contrast, imagine uniformly sampling randomized measurements from $\V= \{U^{\otimes n}\}_{U \in \text{Cl}(2)}$. Similar to global $SU(2)$ control, we call this ``global $\text{Cl}(2)$ control''. Each data point measures with the all-$X$, all-$Y$, or all-$Z$ Pauli string, and as a result, the variance is $\mathcal{O}(1)$ (see Appendix \ref{app:globalSU2}). In this latter case, since we exponentially more often measure with the all-$X$ Pauli string, we require exponentially fewer measurements.

By increasing the probability that we sample unitaries which are ``useful'' for estimating our observables, we can reduce the total number of required measurements. In the example above, useful unitaries allowed us to estimate the all-$X$ Pauli string. 
Therefore, a key feature of our protocol is the ability to adapt the probability distribution $p$ from which we sample unitaries in our randomized measurements. 
This distribution can be adapted using $\mathcal{K}^{\widetilde{O}}$, which naturally quantifies a unitary's utility: more useful unitaries $V$ have larger $|\mathcal{K}^{\widetilde{O}}(V,b)|$. By constrast, $V$'s with smaller $|\mathcal{K}^{\widetilde{O}}(V,b)|$ contribute less to the empirical average (Eq.~\eqref{eqn:empiricalavg}). In Appendix \ref{app:Adaptivity}, we derive an expression for the optimal probability distribution and find it more often samples unitaries $V$ corresponding to large $|\mathcal{K}^{\widetilde{O}}(V,b)|$. 
We can further reduce the required number of measurements by adapting the estimator $\widetilde{O}$ of our observable $O$.
In the visible space, some observables have small variances and so require fewer measurements. These observables are ``easy'' to predict. Other observables in the visible space have larger variances and so are ``harder'' to predict.
We can truncate the observable $O$'s harder-to-predict parts $O \rightarrow \widetilde{O}$ as long as $\|\widetilde{O} - O\|_\infty$ is within the desired precision. The truncation should not appreciably change the observable.
Since the \textit{biased} observable $\widetilde{O}$ focuses on the parts of $O$ that are easier to predict, we need fewer experimental data points for our estimates. 
This is a type of bias-variance tradeoff \cite{mohri2012foundations}. 
Moreover, in cases where $O$ is outside of but sufficiently close to the visible space, we can still estimate $\tr(\rho O)$ by projecting $O$ onto the visible space. 
In Appendix \ref{app:OptimizingBiasVar} we discuss these ideas in detail and provide explicit methodologies for choosing the optimal $\widetilde{O}$.

\vspace{1mm}
\textit{Applications.}
In this section, we illustrate our randomized measurement protocol with two applications.
The first application is simulating a $U(1)$ lattice gauge theory using a quantum device with limited control.
We show that only having access to global $SU(2)$ measurements suffices to efficiently measure the energy density of any state. By contrast, previous works on estimating energies in lattice gauge theories require a fine degree of control~\cite{huang2020predicting,kokail2019self,tagliacozzo2013simulation}.

We consider a 2+1 dimensional $U(1)$ lattice gauge theory, with the bosonic degrees of freedom truncated to finite dimensions, and specialize to the setting of a triangular lattice with periodic boundary conditions (Fig.~\ref{fig:Figure2}(a)). As shown in~\cite{brower2019lattice}, we can write the Hamiltonian for this theory on a quantum simulator with qubits as $H = \sum_{s}\!\left(\sum_\triangle H^{(s)}_{\triangle} + \sum_\leftrightarrow H^{(s)}_{\leftrightarrow}\right)$, where $\sum_s$ is a sum over bosonic modes, $\sum_\Delta$ is a sum over plaquettes, and $\sum_{\leftrightarrow}$ is a sum over ``links'' which couple qubits in adjacent plaquette layers. 
Fig.~\ref{fig:Figure2}(b) depicts the stack of plaquettes representing the bosonic modes of a single triangle in our 2d lattice.
In Appendix~\ref{app:NumericsLGT}, we provide a detailed description of the Hamiltonian $H$ and prove that it lives entirely in the global $SU(2)$ visible space. Therefore, we can estimate the expectation value of each term in the Hamiltonian -- and hence the local energy density -- by solely utilizing global $SU(2)$ control.

In Fig.~\ref{fig:Figure2}(c), we showcase the utility of bias-variance tradeoffs by plotting the estimation error of a link term $H_\leftrightarrow$ as a function of bias. The plot's bowl shape demonstrates how slightly biasing an operator can lead to a substantially lower estimation error. Moreover, the wide plateau at the minimum estimation error indicates that a large window of biased operators $\widetilde{O}$ gives a nearly-optimal estimation error. In Fig.~\ref{fig:Figure2}(d), we compare our learning protocol to the original classical shadows formulation (Fig.~\ref{fig:Figure2}(d)). We find that our proposed protocol utilizing bias and adapting the sampling distribution requires substantially fewer measurements (see Appendix~\ref{app:NumericsLGT} for details).

\vspace{1mm}
As our second application, we establish that certain topological states can be classified with unsupervised learning algorithms using only global $SU(2)$ measurements.  This greatly expands the practicality of the recent work~\cite{huang2022provably}, which requires tomographically complete measurements. 
We consider a $200$ qubit system where each qubit lives on the edge of a square lattice of length $L=10$ with periodic boundary conditions.
Recall that any two states in a topologically-ordered phase are connected by a small-depth geometrically-local quantum circuit~\cite{chen2010local, wen2017colloquium, zeng2019quantum}.  We generate random states within a topological phase by starting with a representative state -- a product state for the trivial phase and the toric code ground state for a non-trivial phase~\cite{kitaev2003fault} -- and evolving 
under a small-depth geometrically-local random circuit (see Figure~\ref{fig:Figure3}(a)).
Then we perform $1000$ global $SU(2)$ randomized measurements and estimate all local $3$-qubit observables in the visible space. 
These local observables form a ``feature vector'' for each state, which we then analyze using an unsupervised, classical machine learning (ML) model (see Appendix~\ref{app:NumericsPhases}). 

Figure~\ref{fig:Figure3}(b) displays a one-dimensional representation found by the ML model for different random circuit depths. At each depth we consider $50$ states from the toric code phase and $50$ from the trivial phase.
Since these are randomly-generated states in the respective phases, it is \emph{a priori} difficult to distinguish between the topological phases corresponding to these two collections.
For example, Ref.~\cite{huang2022provably} proved it is impossible to distinguish between the two collections using any linear order parameter.
However, we see that the ML model can find a 1D representation that distinguishes the phases, indicating that the global $SU(2)$ visible space is expressive enough to capture defining features of these two different phases.
In other words, there exists a nonlinear function on the global $SU(2)$ visible space, which can act as a nonlinear order parameter and therefore classify trivial and topologically-ordered phases.
Ref.~\cite{huang2022provably} demonstrated that tomographically-complete measurements allow machines to learn and classify quantum phases. Our results are particularly novel since global $SU(2)$ is far from being tomographically complete and yet is similarly capable of distinguishing the phases. 

\begin{figure}[t]
\centering
\includegraphics[scale = 0.33]{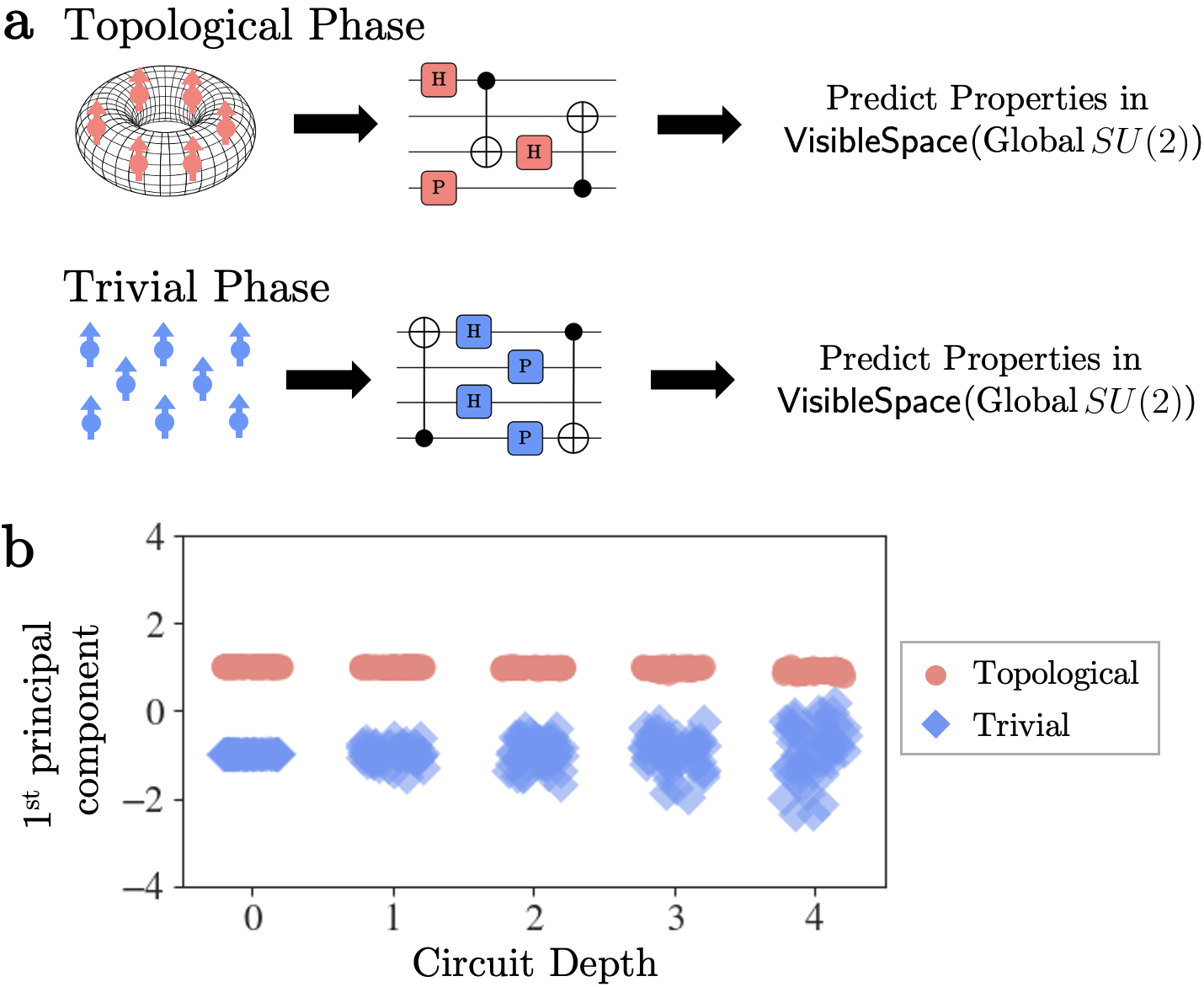}
\caption{\emph{Unsupervised phase prediction with global $SU(2)$.} (a) In a 200-qubit system, we sample a random state in either the topologically-ordered or trivial phase by applying a short-depth random circuit on a prototypical state (the toric code ground state or a product state). For each state, we create a classical representation generated by global $SU(2)$. (b) Using these classical representations, kernel PCA~\cite{huang2022provably,mohri2012foundations} can separate the states in the two phases, meaning that the global $SU(2)$ visible space is expressive enough to capture distinguishing features of these two phases. \label{fig:Figure3}}
\end{figure}

\vspace{1mm}
\textit{Discussion.}
We have developed a framework for efficiently learning many properties of physical systems given specified controls, thus expanding the scope and practicality of classical shadow tomography. The flexibility of our formalism lends itself to adoption in a large range of experimental platforms ~\cite{semeghini2021probing,satzinger2021realizing,bluvstein2021controlling,martin2022controlling, stas2022robust, choi2019probing,preskill2018quantum, altman2021quantum, gross2017quantum,ebadi2021quantum,schafer2020tools, monroe2021programmable, kjaergaard2020superconducting}, and our initial numerical investigations demonstrate the promise of our methods in experimental practice.

Our work can be extended along a number of directions. In particular, in many experimental   
platforms information is encoded in 
 bosonic ~\cite{pelucchi2022potential, slussarenko2019photonic} or fermionic~\cite{kloeffel2013prospects, burkard2021semiconductor} degrees of freedom, as opposed to qubits.  Our methods can be combined with recent generalizations of classical shadow tomography~\cite{chen2021robust, koh2022classical,levy2021classical,huang2021efficient,hu2022logical} to allow for hardware efficient learning in such bosonic and fermionic settings.
Other directions include incorporating adaptivity~\cite{lange2022adaptive} into our learning protocols so that the measurement ensemble is judiciously altered as more data is obtained. In addition, one could incorporate prior knowledge about the state, such as it having an efficient tensor network representation~\cite{akhtar2022scalable}, to further increase the efficiency of our learning protocols. 

Even as the degree of control over large quantum systems improves, certain types of operations and measurements will always be preferable due to higher fidelity or length of implementation time, for example. 
Our approach provides a robust, efficient route to hardware-efficient quantum state learning under these general circumstances, making it invaluable for analyzing experiments involving large-scale quantum systems.



\vspace{1mm}
\textit{Acknowledgements.} 
We would like to thank Richard Kueng, Nathan Leitao, Nishad Maskara, Hannes Pichler, and Leo Zhou for helpful discussions. 
We acknowledge financial support from and the U.S. Department of Energy [DE-SC0021013 and DOE Quantum Systems Accelerator Center (Contract No.: DE-AC02-05CH11231)], the National Science Foundation and CUA. 
KVK acknowledges support from the Fannie and John Hertz Foundation and the National Defense Science and Engineering Graduate (NDSEG) fellowship. 
JC is supported by a Junior Fellowship from the Harvard Society of Fellows, as well as in part by the Department of Energy under grant {DE}-{SC0007870}.
HH is supported by a Google Ph.D. fellowship and a MediaTek Research Young Scholarship.

\newpage
\bibliography{references}
\bibliographystyle{ieeetr} 


\onecolumngrid  
\appendix

\begin{appendix}
\clearpage 

\vspace{2.0em}
\begin{center}
\textbf{\Large Appendix}
\end{center}

\renewcommand{\appendixname}{APPENDIX}
\renewcommand{\thesubsection}{\MakeUppercase{\alph{section}}.\arabic{subsection}}
\renewcommand{\thesubsubsection}{\MakeUppercase{\alph{section}}.\arabic{subsection}.\alph{subsubsection}}
\makeatletter
\renewcommand{\p@subsection}{}
\renewcommand{\p@subsubsection}{}
\makeatother

\renewcommand{\figurename}{Supplementary Figure}
\setcounter{figure}{0}
\setcounter{secnumdepth}{3}

\bigskip

\noindent \textbf{\ref{app:VisibleSpace}.~~\hyperref[app:VisibleSpace]{Tomography protocol for predicting observables in the visible space}} \dotfill\textbf{\pageref{app:VisibleSpace}}
\medskip

\noindent \qquad \begin{minipage}{\dimexpr\textwidth-1.3cm}
 \hyperref[sec:VisibleSpace]{The visible space and invisible space}
 $\bullet$ 
 \hyperref[sec:learning-protocol]{Learning protocol} 
 $\bullet$ 
 \hyperref[sec:KVbAndClassicalShadows]{Recovering classical shadow tomography}
\end{minipage}
\medskip

\noindent \textbf{\ref{app:globalSU2}.~~\hyperref[app:globalSU2]{A form of limited control: global $SU(2)$}} \dotfill\textbf{\pageref{app:globalSU2}}
\medskip

\noindent \qquad \begin{minipage}{\dimexpr\textwidth-1.3cm}
 \hyperref[sec:su2Visible]{Global $SU(2)$ visible space}
 $\bullet$ 
 \hyperref[sec:su2MChannel]{Global $SU(2)$ measurement channel}
 $\bullet$ 
 \hyperref[sec:ProofOfVisInvisBasis]{ Proof of visible and invisible space basis}
\end{minipage}
\medskip

\noindent \textbf{\ref{app:globalCL2}.~~\hyperref[app:globalCL2]{Another form of limited control: global $\text{Cl}(2)$}} \dotfill\textbf{\pageref{app:globalCL2}}
\medskip

\noindent \qquad \begin{minipage}{\dimexpr\textwidth-1.3cm}
 \hyperref[subsec:cl2mmtchannel]{Global  $\text{Cl}(2)$ measurement channel}
 $\bullet$ 
 \hyperref[subsec:cl2mmtsamplecomplexity]{Global $\text{Cl}(2)$ sample complexity}
\end{minipage}
\medskip

\noindent \textbf{\ref{app:OptimizingBiasVar}.~~\hyperref[app:OptimizingBiasVar]{Minimizing on the bias-variance tradeoff}} \dotfill\textbf{\pageref{app:OptimizingBiasVar}}
\medskip

\noindent \qquad \begin{minipage}{\dimexpr\textwidth-1.3cm}
 \hyperref[subsec:appBVSmallSystems]{Optimization for small system sizes} $\bullet$ 
 \hyperref[subsec:appBVLargeSystems]{Optimization for local operators and low-depth implementable unitaries}
\end{minipage}
\medskip

\noindent \textbf{\ref{app:Adaptivity}.~~\hyperref[app:Adaptivity]{Adapting the probability density function}} \dotfill\textbf{\pageref{app:Adaptivity}}
\smallskip

\noindent \qquad \begin{minipage}{\dimexpr\textwidth-1.3cm}
 \hyperref[sec:UpdateKforPDF]{How to update $\mathcal{K}^O$ for a new probability density function} 
 $\bullet$ 
 \hyperref[sec:PDFmaxvar]{Optimal probability density function}
 $\bullet$
 \hyperref[sec:PDFmanyop]{Optimal probability density function for many operators}
 $\bullet$
 \hyperref[sec:PDFmaxmixed]{Probability density function for the infinite temperature state}
\end{minipage}
\medskip

\noindent \textbf{\ref{app:NumericsLGT}.~~\hyperref[app:NumericsLGT]{Estimating the energy density of a $U(1)$ lattice gauge theory}} \dotfill\textbf{\pageref{app:NumericsLGT}}
\medskip

\noindent \qquad \begin{minipage}{\dimexpr\textwidth-1.3cm}
 \hyperref[appx:LGT_fig3b]{Showcasing the bias-variance tradeoff}
 $\bullet$ 
 \hyperref[appx:LGT_fig3c]{Reducing the required number of measurements}
\end{minipage}
\medskip

\noindent \textbf{\ref{app:NumericsPhases}.~~\hyperref[app:NumericsPhases]{Classifying topological phases with global $SU(2)$ control}} \dotfill\textbf{\pageref{app:NumericsPhases}}
\medskip

\noindent \qquad \begin{minipage}{\dimexpr\textwidth-1.3cm}
 \hyperref[appx:phases_significance]{Results and their significance} $\bullet$ 
 \hyperref[appx:phases_globalsu2data]{Extracting information about states using global $SU(2)$} $\bullet$ 
 \hyperref[appx:phases_kernelconstructionPCA]{Kernel principal component analysis}
\end{minipage}
\medskip

\vspace{6mm}

Let us briefly summarize the contents of the appendices.  Appendix \ref{app:VisibleSpace} provides a pedagogical introduction to our learning protocol for predicting visible space observables. 
Next we discuss the details behind the global $SU(2)$ ensemble in Appendix \ref{app:globalSU2} and the global Cl$(2)$ ensemble in Appendix \ref{app:globalCL2}. Throughout our work, the global $SU(2)$ and global Cl$(2)$ ensembles are our guiding examples of limited control, and these appendices derive the relevant machinery. 

Appendix \ref{app:OptimizingBiasVar} and Appendix \ref{app:Adaptivity} provide two algorithms, which can be used in isolation or in combination to reduce the number of measurements required for our learning framework. In Appendix \ref{app:OptimizingBiasVar}, we explain how to alter the operators whose expectation values we wish to estimate so as to reduce the overall estimation error.  In Appendix \ref{app:Adaptivity}, in a similar spirit, we explain how to alter the unitary ensemble in our randomized measurement protocol. 

The final two appendices provide details on the applications discussed in the main text. In Appendix \ref{app:NumericsLGT} we estimate the energy density of a $U(1)$ lattice gauge theory, and in Appendix \ref{app:NumericsPhases} we demonstrate that we can classify certain topological phases with limited controls. In each of these appendices, we provide relevant background background and then discuss the implementation of numerics. All of our code is available at \url{https://github.com/katherinevankirk/hardware-efficient-learning}.

\vspace{3mm}
\section{\label{app:VisibleSpace} Tomography protocol for predicting observables in the visible space}

\subsection{\label{sec:VisibleSpace}The visible space and invisible space}

We present a protocol for predicting many properties of some fixed, unknown state $\rho$. Our protocol is both more general and often more efficient than previous techniques such as classical shadow tomography \cite{huang2020predicting}. Consider the following problem statement: having prepared the $n$-qubit state $\rho$, for which operators can you estimate $\tr(\rho O)$ on your experimental platform? In order to address this question, we parameterize the degree of control of an experimental platform via the set of \textit{implementable unitaries}. 
Following the preparation of some state $\rho$, you can evolve under any unitary $V$ in the set of implementable unitaries $\V = \{V\}$ and subsequently perform a projective measurement $\{\Pi_b\}_b$, where the $\Pi_b$'s are projectors satisfying $\sum_b \Pi_b = \mathds{1}$. For simplicity, let us specialize to $\Pi_b = \ketbra{b}{b}$ with $b \in {0,1}^n$.
Suppose we consider systems with $n$ qubits; we have the following definition:
\vspace{3mm}

\begin{definition}[\textbf{Implementable Unitaries}]
    The set of implementable unitaries, for some experimental platform, is the collection $\V = \{V\}$ of all $2^n \times 2^n$ unitary matrices that can evolve an $n$-qubit system prior to readout by projection onto a computational basis state.
\end{definition}
\noindent As such, the set $\V$ \textit{specifies} all possible measurements that can be made by the experimental platform. Here, a ``randomized measurement'' on the state $\rho$ consists of rotating $\rho
\rightarrow V \rho V^{\dagger}$ with some $V$ sampled from $\V$ and then projecting onto one of the computational basis states $\{\ket{b}\}_{b \in \{0,1\}^n}$.  Therefore, the implementable unitaries $\V$ determine for which operators $O$ we can estimate $\tr(\rho O)$.

\vspace{1mm}
We emphasize that there are certain $O$ whose expectation values cannot be estimated with the measurements defined by $\V$.  For example, if it so happens that $V O V^\dagger$ is traceless for all $V \in \V$, we cannot estimate $\tr(\rho O)$ because $\bra{b}VOV^\dagger\ket{b}$
will zero for all $V$ and all computational basis states $|b\rangle$.
We can define the space of operators that are invisible to us under the implementable unitaries $\V$, and we call this space the $\textnormal{\textsf{InvisibleSpace}}$($\V$).
\begin{definition}[\textbf{Invisible Space}]
    Given the implementable unitaries $\V$, an operator $O$ lives in the invisible space $\textnormal{\textsf{InvisibleSpace}}(\V)$ if for all $V \in \V$ and all computational basis states $\ket{b}$, we have $\bra{b} V O V^{\dagger} \ket{b} = 0$.
\end{definition}
\noindent The invisible space is, in fact, a subspace of $\mathcal{B}(\mathcal{H})$.  A related notion arises when we consider the orthogonal complement of the invisible space, providing us with the space of all operators $O$ for which we \textit{can} estimate the expectation value $\tr(\rho O)$. We call this orthogonal subspace the $\textnormal{\textsf{VisibleSpace}}(\V)$, and we have the decomposition
\begin{equation}
\mathcal{B}(\mathcal{H}) \simeq \textnormal{\textsf{VisibleSpace}}(\V) \oplus \textnormal{\textsf{InvisibleSpace}}(\V)\,.
\end{equation}
Sets of implementable unitaries for which the invisible space is empty are known as ``tomographically complete''. The corresponding visible space is $\mathcal{B}(\mathcal{H})$.  As a general rule, a more limited the set of implementable unitaries $\V$ corresponds to a smaller visible space $\textnormal{\textsf{VisibleSpace}}(\V)$.
\vspace{3mm}

\begin{adjustwidth}{1cm}{}
\noindent\textit{Example (Random Paulis). }Suppose that an experimentalist has enough local control to measure any Pauli string.  Thus the visible space is spanned by the set of all $n$-qubit Pauli strings, and accordingly is equal to all of $\mathcal{B}(\mathcal{H})$.  This means that the invisible space is trivial, and so the Pauli strings enable measurements which are tomographically complete.
\end{adjustwidth}
\vspace{4mm}

So far, we have defined the visible space as the orthogonal complement of the invisible space.  Using the lemma and observation below, we will derive another, equivalent definition of $\textnormal{\textsf{VisibleSpace}}(\V)$. This forthcoming definition was quoted in the main text (Definition~\ref{def:VisibleSpace}).
\vspace{3mm}

\begin{lemma} \label{lem:visibleUbbU1}
    Consider some operator $O \in \mathcal{B}(\mathcal{H})$, which is normalized with respect to the Hilbert-Schmidt inner product. If there exists a $V \in \V$ and a computational basis state $b$ such that $| \bra{b}VOV^{\dagger}\ket{b} | = 1$, then $O$ lives entirely in the visible space $\textnormal{\textsf{VisibleSpace}}(\V)$.
\end{lemma}

\begin{proof}
    We will prove the contrapositive: \textit{If a normalized operator $O$ is not entirely in $\textnormal{\textsf{VisibleSpace}}(\V)$, then for all $V \in \V$ and computational basis states $b$, $| \bra{b}VOV^{\dagger}\ket{b} | \neq 1$}. Consider a normalized operator $O$ that has some nontrivial component in $\textnormal{\textsf{VisibleSpace}}(\V)$. It can be written as a linear combination of its normalized projection onto the invisible space ($O_I$) and its normalized projection onto the visible space ($O_V$):
    \begin{equation} \nonumber
        O = c_I O_I + c_V O_V \textnormal{ where $c_I \neq 0$}\,.
    \end{equation}
    Since $O$, $O_I$, and $O_V$ are all normalized with respect to the Hilbert-Schmidt inner product, we must have $|c_I|^2 + |c_V|^2 = 1$, and since $c_I \neq 0$, we find that $|c_V| < 1$. Moreover, for any $V \in \V$ and any computational basis state $|b\rangle$, 
    \begin{eqnarray} \nonumber
        |\bra{b}VOV^{\dagger}\ket{b} |
        &=& |\bra{b}V ( c_I O_I + c_V O_V ) V^{\dagger}\ket{b}| \\\nonumber
        &=& |c_V| \,|\bra{b}V O_V V^{\dagger}\ket{b}|\\\nonumber
        &<&1
    \end{eqnarray}
    In the expression above, we move from line 1 to line 2 because, by our definition of $O_I \in \textsf{InvisibleSpace}(\V)$, we must have $\bra{b}VO_IV^{\dagger}\ket{b} = 0$ for all $V$ and all $b$. Then, we move from line 2 to line 3 by noticing that $|\bra{b}V O_V V^{\dagger}\ket{b}|$ is at most 1 and use the fact that $|c_V| < 1$. 
\end{proof}
\vspace{3mm}

This lemma provides one way to show that some operator $O$ lives entirely in the visible space. We will now use the lemma to provide an alternate definition for $\textnormal{\textsf{VisibleSpace}}(\V)$.
\vspace{3mm}

\begin{observation}{\label{obs:visibledef}}
     $\textnormal{\textsf{VisibleSpace}}(\V)$ is the span of $V^{\dagger}\ketbra{b}{b}V$ over every $V \in \V$ and every $b \in \{0,1\}^n$.
\end{observation}

\begin{proof} 
    First we establish that if $A \in \textnormal{\textsf{VisibleSpace}}(\V)$, then $A$ is in the span of $V^{\dagger}\ketbra{b}{b}V$ over every $V \in \V$ and every $b \in \{0,1\}^n$. We do so with a proof by contradiction. For some $A \in \textnormal{\textsf{VisibleSpace}}(\V)$, assume that some nontrivial component of $A$, call it $A^{\perp}$, lives outside the span of $\{V^{\dagger}\ketbra{b}{b}V\}_{V,b}$. Therefore, $A^{\perp}$ will have no overlap with $V^{\dagger}\ketbra{b}{b}V$, meaning for all $V \in \V$ and all computational basis states $\ket{b}$, we have $\tr(A^{\perp} V^{\dagger}\ketbra{b}{b}V)=0$. However, by definition this implies that $A^{\perp} \in \textnormal{\textsf{InvisibleSpace}}(\V)$, and so we have reached a contraction.

    Next we show that if $A$ is in the span of $V^{\dagger}\ketbra{b}{b}V$ over every $V \in \V$ and every $b \in \{0,1\}^n$, then $A \in \textnormal{\textsf{VisibleSpace}}(\V)$.
    Consider some operator $A = \sum_i c_i \,V_i^{\dagger}\ketbra{b_i}{b_i}V_i$ where $c_i$ is a complex number, $V_i \in \V$, and $b_i \in \{0,1\}^n$.  Writing $O_i = V_i^\dagger |b\rangle \langle b_i| V_i$, we see by Lemma~\ref{lem:visibleUbbU1} that each $O_i$ lives entirely in the visible space since $\langle b_i| V_i O_i V_i |b_i\rangle = 1$.  Thus $A = \sum_i c_i\, O_i$ is a linear combination of operators which each lie entirely in the visible space, and so by linearity $A$ itself lives entirely in the visible space.
\end{proof}

\begin{repdefinition}{def:VisibleSpace}\textnormal{\textbf{(Visible Space)}} Following Observation \ref{obs:visibledef},  $\textnormal{\textsf{VisibleSpace}}(\V)$ is the span of $V^{\dagger} \ketbra{b}{b} V$ over every $V \in \V$ and every $b \in \{0,1\}^n$.
\end{repdefinition}

Finally, let us come back to the intuition we gave at the beginning of this section: the visible space contains the operators we can estimate using $\V$. This intuition matches our new definition for the visible space. Since measurements using the implementable unitaries $\V$ are gathering statistics on expectation values of the form $\tr(\rho V^{\dagger} \ketbra{b}{b} V)$, one can estimate any operator in the span of $V^{\dagger} \ketbra{b}{b} V$ over $V \in \V$ and $b \in \{0,1\}^n$.

\subsection{Learning protocol}\label{sec:learning-protocol}

Equipped with a set of implementable unitaries $\V$ in an experimental setup, an experimentalist can use their measurement outcomes to estimate $\tr(\rho O)$ for any $O \in \textnormal{\textsf{VisibleSpace}}(\V)$. 
In this section, we will discuss \textit{how} one actually estimates these visible observables.  Consider a probability density function $p(V)$ supported on the implementable unitaries $\V$.
Following the definition of the visible space (Definition \ref{def:VisibleSpace}), for any observable $O \in \textnormal{\textsf{VisibleSpace}}(\V)$, there exists a probability density function $p$ and a function $\mathcal{K}^O: \V \times \{0,1\}^n \rightarrow \mathbb{C}$ \cite{d20042} such that
\begin{equation}\label{eqn:OvisKVb2}
    O = \int_{V \in \V} dV \, p(V) \sum_b \mathcal{K}^O(V,b) \hspace{1mm} V^{\dagger} \ketbra{b}{b} V\,.
\end{equation}
Using the representation of a visible space operator in~\eqref{eqn:OvisKVb2}, $\tr(\rho O)$ can be expressed via the following expectation over $V$ and $b$:
\begin{eqnarray}
\tr(\rho O)
&=& \int_{V \in \V} dV \, p(V) \sum_b \bra{b}V\rho V^{\dagger}\ket{b} \mathcal{K}^O(V,b) \\
&=&   \Expect_{V,b\,\sim\,P_\rho} \mathcal{K}^O(V,b)\,. \label{eqn:KVBreconstructTrRhoO}
\end{eqnarray}
In this expectation $V$ is sampled from $p(V)$ and $b$ is sampled with probability $\bra{b} V \rho V^{\dagger} \ket{b}$, meaning the probability of obtaining the outcome $(V,b)$ is $P_\rho(V,b) = p(V)\bra{b}V\rho V^\dagger\ket{b}$. 
Therefore, to estimate $\tr(\rho O)$, one can perform experiments of the following kind: each experiment prepares a new copy of $\rho$, evolves under a random unitary $V$ sampled from $p(V)$, and measures to obtain a bitstring $b \in \{0, 1\}^n$. 
Each of these experiments is a ``randomized measurement.''
Indeed, the probability $P_\rho(V,b)$ in~\eqref{eqn:KVBreconstructTrRhoO} is precisely the probability of experimentally obtaining $(V, b)$ in a randomized measurement. As a result, the set of outcomes $\{ (V_s, b_s) \}_{s=1}^N$ of $N$ randomized measurements can be used to construct the empirical average 
\begin{equation}\label{eqn:empiricalavgAPPX}
    \tr(\rho O) \approx \frac{1}{N} \sum_{s=1}^N \mathcal{K}^{O}(V_s, b_s).
\end{equation}
This average will be a good estimate for $\Tr(\rho O)$ provided the variance of $\mathcal{K}^{O}(V, b)$ is small. This variance depends on the state $\rho$, but since we have no information about the state, we define $\textnormal{Var}_{\textnormal{max}}[\mathcal{K}^{O}]$ to be the variance maximized over all possible states $\rho$:
\begin{equation}
    \textnormal{Var}_{\textnormal{max}}[\mathcal{K}^{O}(V,b)] := \max_{\rho} \textnormal{Var}_{V,b\,\sim\,P_\rho}[\mathcal{K}^{O}(V, b)]\,.
\end{equation}

\subsubsection{Introducing bias into our estimates.}

With an eye towards being resource efficient, we can introduce bias \cite{hadfield2022measurements} into our estimate in order to reduce our estimate's variance.  In particular, given an observable $O$, we can utilize the experimental data $\{ (V_s, b_s) \}_{s=1}^N$ to estimate $\Tr(\rho \widetilde{O})$ for some  $\widetilde{O} \in \textnormal{\textsf{VisibleSpace}}(\V)$ that is close to $O$. 
If we have a good estimate for $\Tr(\rho \widetilde{O})$, then it will also be a good estimate for $\Tr(\rho O)$.  Crucially, a judicious choice of $\widetilde{O}$ can dramatically reduce the variance of the estimate.
Appendix~\ref{app:OptimizingBiasVar} will discuss how to use the $\mathcal{K}(V,b)$ formalism to find an $\widetilde{O}$ close to $O$ that minimizes the estimation error for a fixed number of measurements, but we will briefly build some intuition here. 

Choosing the optimal $\widetilde{O}$ is a type of bias-variance tradeoff: by replacing $O$ with a ``biased'' operator $\widetilde{O}$, we can reduce the ``variance'' of our estimator \cite{hadfield2022measurements,mohri2012foundations}. But we do not want $\widetilde{O}$ to be so biased that $\tr(\rho \widetilde{O})$ no longer approximates $\tr(\rho O)$. Let us parameterize $\widetilde{O}$ as
\begin{equation}
\widetilde{O}(f) := \Expect_{V \sim p}\,\sum_{b \in \{0,1\}^n} f(V,b) \, V^\dagger |b\rangle \langle b| V\,.
\end{equation}
We would like to find the function $f(V,b)$ such that the estimation error $\text{Error}(f) = |\text{tr}(\rho O) - \frac{1}{N} \sum_{s=1}^N f(V_s, b_s)|$ is as small as possible.  Consider the cost function
\begin{equation}\label{eqn:costfunctionBVTradeOff}
    \textnormal{Cost}(f) := \|O-\widetilde{O}(f)\|_{\infty} + \sqrt{\frac{2}{N}\, \text{Var}[f] \, \log\!\left(\frac{1}{2\delta}\right)}\,,
\end{equation} 
where $\text{Var}[f]$ samples $V$ according to $p(V)$ and samples $b$ uniformly. 
Moreover, $\|\cdot\|_\infty$ is the operator norm. 
This cost function approximates the estimation error when $\rho$ is the maximally mixed state, but is empirically a good proxy for the estimation error for generic $\rho$ (see Appendix \ref{app:OptimizingBiasVar}).
Since the cost function is convex, it has a unique minimizer,  and in practice one can approximate $\V$ by subsampling and then use gradient descent to find the minimizing $f$ and thus the optimal biased operator $\widetilde{O}(f)$.  See Appendix~\ref{app:OptimizingBiasVar} for details.

\subsubsection{Predicting expectation values of many observables.}
We can reuse the same data set $\{ (V_s, b_s) \}_{s=1}^N$ to predict the expectation values $\Tr(\rho O_1), \ldots , \Tr(\rho O_M)$ of many observables. We can construct estimates using an empirical average (e.g.~\eqref{eqn:empiricalavgAPPX}), and a rigorous guarantee on precision is given below in Theorem \ref{thm:shadowtomtheorem}.

\begin{reptheorem}{thm:shadowtomtheorem}
Suppose we are given $M$ operators $O_1, \ldots, O_M$ and probability density function $p(V)$ with support over the implementable unitaries $\V$. Consider ``biased'' operators $\widetilde{O}_1, \ldots, \widetilde{O}_M \in \textnormal{\textsf{VisibleSpace}}(\V)$ and fix accuracy parameters $\epsilon$, $\delta$ $\in [0,1]$. 
Further set
\begin{equation} \label{eqn:thm1NumMmts}
    N = 2  \log \!\left(\frac{M}{2\delta}\right) \max_i \frac{\textnormal{Var}_{\textnormal{max}}[\mathcal{K}^{\widetilde{O}_i}] + \frac{1}{3}\epsilon \,Q_i}{\epsilon^2},
\end{equation}
where $Q_i := \max_{V_i,b_i} |\mathcal{K}^{\widetilde{O}_i}(V_i,b_i)| + \|\widetilde{O}_i\|_\infty $. Upon performing $N$ randomized measurements outputting $\{(V_s, b_s)\}_{s=1}^N$, with probability at least $1 - \delta$ we have
\begin{equation}
    \left|\frac{1}{N} \sum_{s=1}^N \mathcal{K}^{\widetilde{O}_i}(V_s, b_s) - \text{\rm tr}(O_i \rho) \right| \leq \|O_i-\widetilde{O}_i\|_{\infty} + \epsilon  \hspace{3mm} \textit{ for all }1 \leq i \leq M\,.
\end{equation}
\end{reptheorem}
\begin{proof}
    The claim follows from Bernstein's inequality \cite{foucart2013mathematical}. Consider the zero-mean random variable $X_i(V_s,b_s) = \mathcal{K}^{\widetilde{O}_i}(V_s, b_s) - \tr(\widetilde{O}_i \rho)$ which satisfies $|X_i| \leq Q_i$ by virtue of the triangle inequality and by maximizing over states and all possible $\{V_i, b_i\}$. Then $N$ randomized measurements $\{ (V_s, b_s) \}_{s=1}^N$ suffice to construct an empirical average that obeys
    \begin{equation} \label{eqn:bernsteinsperformanceguarantee}
        \Pr\left[\left|\frac{1}{N}\sum_{s=1}^N \mathcal{K}^{\widetilde{O}_i}(V_s, b_s) \hspace{1mm} - \tr(\widetilde{O}_i \rho)\right| \geq  \epsilon\right] \leq 2\exp \left(-\frac{\frac{1}{2} \epsilon^2 N}{\sigma_i^2 + \frac{1}{3}Q_i \epsilon}\right),
    \end{equation}
    where $\sigma_i^2$ is the variance $\Expect [X^{2}_i] = \textnormal{Var}_{V,b\,\sim\,P_\rho}[\mathcal{K}^{\widetilde{O}_i}(V, b)]$. 
    However, we assume no knowledge of the state of interest $\rho$, so we will upper bound the variance with the maximum variance over states $\textnormal{Var}_{\textnormal{max}}[\mathcal{K}^{\widetilde{O}_i}(V,b)]$.
    
    In order to bound the estimation error between our empirical average and the \textit{true} expectation value $\tr(\rho O_i)$, we use the triangle inequality (first to second line below) and upper bound $|\tr(\rho (O_i - \widetilde{O}_i))|$ with the operator norm (second to third line):
    \begin{eqnarray}
        \left|\frac{1}{N}\sum_{s=1}^N \mathcal{K}^{\widetilde{O}_i}(V_s, b_s) \hspace{1mm} - \tr(\widetilde{O}_i \rho)\right| 
        &=& 
        \left|\frac{1}{N}\sum_{s=1}^N \mathcal{K}^{\widetilde{O}_i}(V_s, b_s) \hspace{1mm} - \tr(O_i \rho) + \tr(O_i \rho) - \tr(\widetilde{O}_i \rho)\right|
        \\
        &\geq&
        \left|\frac{1}{N}\sum_{s=1}^N \mathcal{K}^{\widetilde{O}_i}(V_s, b_s) \hspace{1mm} - \tr(O_i \rho) \right| \hspace{1mm}-\hspace{1mm} \left|  \tr(O_i \rho) - \tr(\widetilde{O}_i \rho)\right| 
        \\
        &\geq&
        \left|\frac{1}{N}\sum_{s=1}^N \mathcal{K}^{\widetilde{O}_i}(V_s, b_s) \hspace{1mm} - \tr(O_i \rho) \right| \hspace{1mm}-\hspace{1mm} \|O_i-\widetilde{O}_i\|_{\infty}\,.
    \end{eqnarray}
    Using this lower bound, we can bound the performance guarantee~ \eqref{eqn:bernsteinsperformanceguarantee} as 
    \begin{equation}\label{eqn:bernsteinsperformance}
        \Pr\left[\left|\frac{1}{N}\sum_{s=1}^N \mathcal{K}^{\widetilde{O}_i}(V_s, b_s) \hspace{1mm} - \tr(O_i \rho)\right| \geq  \epsilon + \|O_i-\widetilde{O}_i\|_{\infty}\right] \leq 2\exp\left(-\frac{\frac{1}{2} \epsilon^2 N}{\textnormal{Var}_{\textnormal{max}}[\mathcal{K}^{\widetilde{O}_i}] + \frac{1}{3}Q_i \epsilon}\right).
    \end{equation}
    
   The number of randomized measurements $N$ is chosen such that the failure probability (the right hand side of~\eqref{eqn:bernsteinsperformance}) is less than or equal to $\delta/M$ for all $1 \leq i \leq M$. Therefore, one simply has to apply a union bound over all $M$ failure probabilities to deduce the claim.
\end{proof}

One could instead utilize a median-of-means estimator, rather than an empirical average as per~\eqref{eqn:empiricalavgAPPX}, to estimate the $M$ expectation values. In the median-of-means setting, one also needs only $N = \mathcal{O}(\frac{\log (M)}{\epsilon^2} \max_{i} \textnormal{Var}_{\textnormal{max}}[\mathcal{K}^{\widetilde{O}_i}])$ randomized measurements to estimate the expectation values to error $\|O_i-\widetilde{O}_i\|_{\infty} + \epsilon$. This is what is quoted in the main text, and the proof of this bound follows the same steps as the proof for Theorem \ref{thm:shadowtomtheorem} above. However, instead of Bernstein's inequality, one uses the rigorous performance
guarantee for median of means estimation~\cite{jerrum1986random,nemirovski1983problem}.

\subsection{\label{sec:KVbAndClassicalShadows}Recovering classical shadow tomography}

Here we will define the `classical shadow kernel' $\mathcal{K}^O_{\textnormal{CS}}$ and show how our learning procedure, equipped with this kernel, recovers classical shadow tomography. We will also comment on the utility of this choice of kernel when $\rho$ is the maximally mixed state and give a new interpretation of the visible and invisible spaces in the language of classical shadow tomography. 
Throughout this entire subsection, assume that $p(V)$ is a uniform probability density function over $\V$.

\subsubsection{$\mathcal{K}^O_{\textnormal{CS}}$ recovers the original classical shadow protocol.}

Below we define a `classical shadow kernel' $\mathcal{K}^O_{\textnormal{CS}}$, and using the language of classical shadows we verify that the visible operator $O$ can be faithfully represented by~\eqref{eqn:OvisKVb2}.
\begin{definition} \label{def:shadowversionofK}
    Assume that $p(V)$ is uniform over $\V$. For any $O \in \text{\rm\textsf{VisibleSpace}}(\V)$, the classical shadow kernel is defined as
    \begin{equation}
    \mathcal{K}^O_{\textnormal{CS}}(V,b) = \bra{b} V \mathcal{M}^{-1}(O)V^\dagger \ket{b}\,,
\end{equation}
where $\mathcal{M}^{-1}$ is the inverse measurement channel~\cite{huang2020predicting} defined for $\V$ and $p(V)$.
\end{definition}

\vspace{3mm}
\begin{adjustwidth}{1cm}{}
\noindent
\textit{Checking that $\mathcal{K}^O_{\textnormal{CS}}(V,b)$ faithfully represents $O$ via~\eqref{eqn:OvisKVb2}}.
In~\cite{huang2020predicting} the measurement channel is defined as 
$\mathcal{M}(A) = \Expect_{V \sim p} \sum_b V^\dagger\ketbra{b}{b}V \bra{b} VAV^\dagger\ket{b}$.
By substituting our expression for $\mathcal{K}^O_{\textnormal{CS}}(V,b)$ into the right side of~\eqref{eqn:OvisKVb2} and using the definition of the measurement channel, we find that we indeed recover our desired operator $O$:
\begin{eqnarray}
    \Expect_{V \sim p} \sum_b \mathcal{K}^O_{\textnormal{CS}}(V,b) \hspace{1mm} V^{\dagger} \ketbra{b}{b} V 
    &=& \Expect_{V \sim p} \sum_b \bra{b} V \mathcal{M}^{-1}(O)V^\dagger \ket{b} \hspace{1mm} V^{\dagger} \ketbra{b}{b} V \\
    &=& \mathcal{M}(\mathcal{M}^{-1}(O)) \\
    &=& O\,.
\end{eqnarray}
\end{adjustwidth}
\vspace{4mm}

It turns out that this particular choice of $\mathcal{K}^O = \mathcal{K}^O_{\textnormal{CS}}$ recovers the original classical shadow protocol~\cite{huang2020predicting}. Recall that the classical shadow corresponding to evolution under $V$ and measurement of computational basis state $|b\rangle$ is constructed as $\hat{\rho}_{V,b} = \mathcal{M}^{-1}(V^\dagger \ketbra{b}{b} V)$. 
Therefore, using the fact that the measurement channel $\mathcal{M}$ is self-adjoint, one can rewrite $\mathcal{K}^O_{\textnormal{CS}}(V,b)$ as a single shadow estimate of $\Tr(\rho O)$, 
\begin{eqnarray}
    \mathcal{K}^O_{\textnormal{CS}}(V,b) 
    &=& \Tr (O \mathcal{M}^{-1}(V^\dagger \ketbra{b}{b} V) ) \\
    &=&  \Tr (O \hat{\rho}_{V,b})\,.
\end{eqnarray}
In other words, the shadows are ingrained in the definition of $\mathcal{K}^O_{\textnormal{CS}}$.  Therefore, by estimating $\Tr(\rho O)$ with the empirical average $\frac{1}{N} \sum_{i=0}^N \mathcal{K}^O_{\textnormal{CS}}(V_i,b_i)$, one is actually performing $N$-shot classical shadow tomography. 
\vspace{3mm}

\subsubsection{Efficiency of classical shadows for the maximally mixed state.}
The choice of $\mathcal{K}^O = \mathcal{K}^O_{\textnormal{CS}}$ allows for efficient reconstruction of $\tr(\rho O)$ for certain states $\rho$. For example, we will show below that $\mathcal{K}^O_{\textnormal{CS}}$ minimizes the number of measurements required when the state $\rho$ is maximally-mixed.
We formally state and prove this result in the observation below.

\begin{observation}
    Assume that $p(V)$ is uniform over $\V$. 
    For the maximally-mixed state $\rho = \mathds{1}/2^n$, the variance $\textnormal{Var}_{V,b\,\sim\,P_\rho}[\mathcal{K}^{O}(V, b)]$ 
    is minimized for the choice of kernel $\mathcal{K}^O = \mathcal{K}^O_{\textnormal{CS}}$.
\end{observation}
\begin{proof} In order to find the kernel extremizing the variance subject to the constraint that $O = \Expect_{V \sim p}  \sum_b \mathcal{K}(V,b) \hspace{1mm} V^{\dagger} \ketbra{b}{b} V$, we will use the method of Lagrange multipliers. We write
\begin{eqnarray}
    \textnormal{cost}(\mathcal{K}^O) 
    &=& 
    \textnormal{Var}_{V,b\,\sim\,P_\rho}[\mathcal{K}^{O}(V, b)] + \tr\!\left(\Lambda \hspace{1mm} \big(O - \Expect_{V\sim p}  \sum_b \mathcal{K}(V,b) \hspace{1mm} V^{\dagger} \ketbra{b}{b} V\big)\right) 
    \\
    &=& 
    \Expect_{V,b\,\sim\,P_\rho}[\mathcal{K}^{O}(V, b)^2] - \tr(\rho O)^2 + \tr\!\left(\Lambda \hspace{1mm} \big(O - \Expect_{V\sim p}  \sum_b \mathcal{K}(V,b) \hspace{1mm} V^{\dagger} \ketbra{b}{b} V\big)\right)\,, 
\end{eqnarray}
where $\Lambda$ is a $2^n \times 2^n$ matrix of Lagrange multipliers.
The $\tr(\rho O)$ in the second line comes from the fact that $\Expect_{V,b\,\sim\,P_\rho}[\mathcal{K}^{O}(V, b)] = \tr(\rho O)$ as discussed in the previous section (see~\eqref{eqn:KVBreconstructTrRhoO}). Setting the derivative of this cost function to zero, we find that 
\begin{equation}
    \frac{\partial \hspace{1mm} \textnormal{cost} (\mathcal{K}^O)}{\partial \mathcal{K}^O (V,b)} = 2\, P_\rho(V,b)\, \mathcal{K}^O(V,b) - p(V) \bra{b} V \Lambda V^\dagger \ket{b} = 0\,.
\end{equation}
By substituting in the probability of obtaining the outcome $(V,b)$, namely $P_\rho(V,b) = p(V) \bra{b}V\rho V^\dagger\ket{b}$, and using $\rho = \mathds{1}/2^n$, we obtain an expression for $\mathcal{K}^O (V,b)$ in terms of our Lagrange multipliers:
\begin{equation} \label{eqn:appdKintermsofLambda}
     \mathcal{K}^O(V,b) = 2^{n-1} \bra{b} V \Lambda V^\dagger \ket{b}\,.
\end{equation}
Plugging this into our constraint, we can solve for $\Lambda$:
\begin{eqnarray}
    O &=& \Expect_{V \sim p}  \sum_b \mathcal{K}(V,b) \hspace{1mm} V^{\dagger} \ketbra{b}{b} V \\
    &=& \Expect_{V \sim p}  \sum_b 2^{n-1} \bra{b} V \Lambda V^\dagger \ket{b} \hspace{1mm} V^{\dagger} \ketbra{b}{b} V \\
    &=& 2^{n-1} \mathcal{M}(\Lambda)\,.
\end{eqnarray}
We thus obtain our Lagrange multiplier matrix $\Lambda = \frac{1}{2^{n-1}}\mathcal{M}^{-1}(O)$, and when we plug this into~\eqref{eqn:appdKintermsofLambda} we recover the classical shadow version of $\mathcal{K}^O$ from Definition~\ref{def:shadowversionofK}.
\end{proof}

\subsubsection{Visible and invisible space in the language of classical shadows.}

One can use the language of classical shadow tomography to develop another, equivalent understanding of the visible and invisible space. It turns out that these spaces are simply the image and nullspace of the classical shadow measurement channel $\mathcal{M}$ (see Observations \ref{obs:mChEquivalenceINVIS} and \ref{obs:mChEquivalenceVIS} below). We will first prove our statement and then give some intuition for the result.

\vspace{3mm}
\begin{observation}{\label{obs:mChEquivalenceINVIS}}
    If $p$ is uniform over $\V$, then the nullspace of the associated measurement channel $\mathcal{M}$ is equivalent to $\text{\rm \textsf{InvisibleSpace}}(\V)$. 
    Thus for some operator $A$,
    \begin{equation}
    \label{eqn:mchannelEquiv}
        A \in \textnormal{\textsf{InvisibleSpace}}(\V)
        \iff
        \mathcal{M}(A) =  \int_{V \in \V} dV \, p(V) \sum_{b} V^{\dagger}\ketbra{b}{b} V \bra{b}V A V^{\dagger}\ket{b} = 0\,.
    \end{equation}
\end{observation}

\begin{proof}
     (\textit{Forward Direction}) 
    If $A \in $ \textsf{InvisibleSpace}($\V$), then $\bra{b}V A V^{\dagger}\ket{b} = 0$ for all $V$ and $b$, and so every term inside the summand of the $\mathcal{M}$ channel will evaluate to $0$. Thus, $A$ must be in the null space of the $\mathcal{M}$ channel. \\ \\
    (\textit{Backward Direction}) 
    For $A$ in the null space of the $\mathcal{M}$ channel, i.e.~$\mathcal{M}(A) = 0$, the following must hold: 
        \begin{equation}
        \nonumber
            \tr(\mathcal{M}(A)A) =  \int_{V \in \V} dV \, p(V) \sum_{b} | \bra{b}V A V^{\dagger}\ket{b} |^2 = 0\,.
        \end{equation}
    Since each term in sum on the right-hand side is non-negative, $\bra{b}V A V^{\dagger}\ket{b} = 0$ for all $V$ and $b$. Thus $A \in $ \textsf{InvisibleSpace}($\V$).
\end{proof}
\vspace{3mm}

\begin{observation}\label{obs:mChEquivalenceVIS}
If $p$ is uniform over $\V$, then $\text{\rm \textsf{VisibleSpace}}(\V)$ is equal to the image of the associated measurement channel $\mathcal{M}$.
\end{observation}
\begin{proof}
The above follows from observation~\ref{obs:mChEquivalenceINVIS}: $\textnormal{\textsf{VisibleSpace}}(\V)$ is equivalent to the image of the measurement channel $\text{Image}(\mathcal{M})$ because the $\textnormal{\textsf{VisibleSpace}}(\V)$ is the orthogonal complement of $\textnormal{\textsf{InvisibleSpace}}(\V)$. 
\end{proof}

\vspace{4mm}
With the above observations, we can deepen our intuition that the visible space is the set of observables we can estimate using the implementable unitaries $\V$.
Recall that each classical shadow takes the form $\hat{\rho}_{V,b} = \mathcal{M}^{-1}(V^{\dagger}\ketbra{b}{b}V)$. Definition~\ref{def:VisibleSpace} tells us that $V^{\dagger}\ketbra{b}{b}V \in \textnormal{\textsf{VisibleSpace}}(\V)$, and via the observations above, the $\mathcal{M}$ channel is effectively block diagonal with respect to the visible and invisible spaces. Therefore, $\hat{\rho}_{V,b} \in \textnormal{\textsf{VisibleSpace}}(\V)$, and in expectation $\mathbb{E}_{V,b \sim P_\rho(V,b)} \hspace{1mm} \hat{\rho} = \Pi_{\textsf{Visible}}[\rho]$
where $\Pi_{\textsf{Visible}} : \mathcal{B}(\mathcal{H}) \to \V$ is a projection map into the visible space of operators.
As a result, when we use classical shadows generated by implementable unitaries $\V$ to estimate $\braket{O}$, we find that we can only estimate operators $O$ in the visible space.  In particular:
\begin{eqnarray}
     \mathbb{E}_{V,b \sim P_\rho(V,b)} \hspace{1mm} \Tr(\hat{\rho}_{V,b} O) &=&  \Tr( \Pi_{\textsf{Visible}}[\rho] \,O) 
     \\ 
     &=&  \Tr( \rho\, \Pi_{\textsf{Visible}}[O])
     \\ 
     &=& 
     \begin{cases} 
      \Tr(\rho O) & O \in \textnormal{\textsf{VisibleSpace}}(\V) \\
      0 & O \in \textnormal{\textsf{InvisibleSpace}}(\V) \\
      \Tr( \rho \,\Pi_{\textsf{Visible}}[O]) & O  \textnormal{ has support on both}
     \end{cases}\,.
\end{eqnarray} 
Since $\mathbb{E}_{V,b \sim P_\rho(V,b)} \hspace{1mm} \Tr(\hat{\rho} O)$ evaluates to $0$ for any operator $O$ living solely in the invisible space, the expectation values of such operators cannot be estimated using $\V$. Indeed, we can only estimate the expectation values of operators $O$ that live in the visible space.

\section{\label{app:globalSU2} A form of limited control: global $SU(2)$}
\vspace{3mm}
A common form of limited control in contemporary experiments is $\V = \{U^{\otimes n}\}_{U \in SU(2)}$. This form of control, which we call ``global $SU(2)$ control,'' applies some single qubit rotation $U \sim SU(2)$ to every qubit in the system. In this appendix, we provide an explicit orthonormal basis of operators for the global $SU(2)$ visible space as well as invisible space, and
establish that the dimension of the visible space is $\sim 2^n n^2$. Conveniently, this grows exponentially with the system size $n$, meaning that global $SU(2)$ control still enables us to learn exponentially many properties of a quantum system. 
We also derive the measurement channel for global $SU(2)$ assuming that the single qubit unitaries are sampled uniformly. Since $\mathcal{K}^O_{\textnormal{CS}}$ is defined in terms of an $\mathcal{M}$ channel (see Definition \ref{def:shadowversionofK} in Appendix~\ref{app:VisibleSpace}), we can use the global $SU(2)$ $\mathcal{M}$ channel to construct $\mathcal{K}^O_{\textnormal{CS}}$ for this form of control.

\subsection{\label{sec:su2Visible}Global $SU(2)$ visible space}

In this subsection, we discuss the visible and invisible spaces for global $SU(2)$. We construct an explicit basis for the visible and invisible spaces for an arbitrary number of qubits, where the notation for the visible and invisible basis elements will be $\{B\}$ and $\{B^{\perp}\}$, respectively.
Finally, we build intuition by considering the visible and invisible spaces of a $n = 2$ qubit system. 

\begin{definition}[\textbf{Fixed-$\mathds{1}$ Permutation-Invariant Sets}] Consider the set $\text{\rm Pauli}(n)$ of $n$-qubit Pauli strings.  We define special subsets of $\text{\rm Pauli}(n)$ as follows.  Let $R$ be a subset of $\{1,...,n\}$, and $n_X, n_Y, n_Z \in \mathbb{Z}_{\geq 0}$ such that $|R| + n_X + n_Y + n_Z = n$. We define $S_{R,n_X,n_Y,n_Z} \subseteq \text{\rm Pauli}(n)$ as follows: $P \in S_{R,n_X,n_Y,n_Z}$ if and only if
    \begin{center}
        \renewcommand{\theenumi}{\roman{enumi}}%
        \begin{enumerate}
            \item  $P$ contains an $\mathds{1}$ at site $i$ for each $i \in R$, and
            \item $P$ has exactly $n_X$ of the $X$ operators, exactly $n_Y$ of the $Y$ operators, and exactly $n_Z$ of the $Z$ operators.  
        \end{enumerate}
    \end{center}
    We refer to such sets $S_{R,n_X,n_Y,n_Z}$ as \text{\rm fixed-}$\mathds{1}$ \text{\rm permutation-invariant sets} because these sets are invariant under permutation of all of the non-identity elements.
\end{definition}

\vspace{2mm}

\noindent\indent\indent\textit{Example.}  The following are examples of fixed-$\mathds{1}$ permutation-invariant sets for $n = 4$:
\newline
\newline
\indent\indent
$S_{\{2,3\},0,1,1} = \{Y \otimes \mathds{1} \otimes \mathds{1} \otimes Z,\,\, Z\otimes \mathds{1} \otimes \mathds{1} \otimes Y\}$
\newline
\newline
\indent\indent
$S_{\{\}, 3,0,1} = \{X \otimes X \otimes X \otimes Z,\,\,X \otimes X \otimes Z \otimes X,\,\,X \otimes Z \otimes X \otimes X,\,\,Z \otimes X \otimes X \otimes X\}$
\newline
\newline
\indent\indent
$S_{\{1\},1,0,2} = \{\mathds{1}\otimes X \otimes Z \otimes Z,\,\, \mathds{1} \otimes Z \otimes X \otimes Z,\,\, \mathds{1} \otimes Z \otimes Z \otimes X\}$

\vspace{5mm}

We will now define an orthonormal basis for the operator space $\mathcal{B}(\mathcal{H})$ in terms of these fixed-$\mathds{1}$ permutation-invariant sets, and we find (Proposition \ref{prop:basisGlobalSu2VisInvis} below) that this basis can be bi-partitioned such that the $\{B\}$ form an orthonormal basis for the visible space and that the $\{B^{\perp}\}$ form an orthonormal basis for the invisible space. 

\begin{observation} [\textbf{A Basis in terms of Fixed-$\mathds{1}$ Permutation-Invariant Sets}] \label{obs:BiAndBiPERPBasis}
    Consider a fixed-$\mathds{1}$ permutation-invariant set $ S = \{P_1, P_2, ... , P_{|S|}\}$ where we have temporarily dropped the subscript of $S$. We can use this set to define an operator $B_S$: 
    \begin{equation}
        \label{eqn:defineBi}
        B_S = \frac{1}{\sqrt{2^n |S|}} \sum_{j=1}^{|S|} P_j.
    \end{equation}
    In addition to the operator $B_S$, we can further define the following $|S| - 1$ operators:
    \begin{equation}
    \label{eqn:defineBiperp}
        B_{S,k}^{\perp} = \frac{1}{\sqrt{2^n k(k+1)}} \left( \sum_{j = 1}^k P_j - k P_{k+1}\right)
        \hspace{2mm}\text{ for } k = 1, ... , |S| - 1.
    \end{equation}
    The set
    \begin{eqnarray}
    \bigcup_{S} \,\{B_S, B_{S,1}^\perp, B_{S,2}^\perp,...,B_{S,|S|-1}^\perp\}
    \end{eqnarray}
    forms an orthonormal basis for the operator space $\mathcal{B}(\mathcal{H})$.
\end{observation}
\begin{proof}
The collection of all $B$ and $B^{\perp}$ elements defined above form an orthonormal set under the Hilbert-Schmidt inner product. Moreover, since the set of Pauli strings can be partitioned into the fixed-$\mathds{1}$ permutation-invariant sets and since each set $S$ contributes $|S|$ operators, we must have $4^n$ orthonormal operators across the combined set of $B$'s and $B^{\perp}$'s. 
\end{proof}

\begin{proposition} [\textbf{A Basis for the Global $SU(2)$ Visible Space}] \label{prop:basisGlobalSu2VisInvis}
    For $\V = \{U^{\otimes n}\}_{U\in SU(2)}$, an orthonormal basis for $\text{\rm\textsf{VisibleSpace}}(\V)$ is
\begin{align}
\bigcup_S \{B_S\}
\end{align}
and an orthonormal basis for $\text{\rm\textsf{InvisibleSpace}}(\V)$ is
\begin{align}
    \bigcup_{S} \,\{B_{S,1}^\perp, B_{S,2}^\perp,...,B_{S,|S|-1}^\perp\}\,.
\end{align}
\end{proposition}

\noindent We will prove this proposition in Subsection~\ref{sec:ProofOfVisInvisBasis}.

\vspace{4mm}
\begin{adjustwidth}{1cm}{}
\noindent\textit{Example: $n = 2$.}
\newline 
\newline
Consider the global $SU(2)$ visible space for a 2-qubit system. In this setting, there are $13$ orthonormal basis elements for this space corresponding to $13$ distinct $S$'s.  We will just label the $B_S$'s by $B_i$ for $i=1,...,13$ to avoid cluttered notation.  We adopt a similar convention for the $B_{S,j}^\perp$ operators, labelling them as $B_{i,j}^\perp$.  Following Observation~\ref{obs:BiAndBiPERPBasis}, the $B_i$'s are as follows:
\begin{align*}
B_1&=\frac{1}{2}\mathds{1}\otimes\mathds{1} &  
B_2&=\frac{1}{2}\mathds{1}\otimes X &  
B_3&=\frac{1}{2}X\otimes\mathds{1} &  
B_{4}&=\frac{1}{2\sqrt{2}}(X\otimes Y+Y\otimes X)\\
B_5&= \frac{1}{2} X\otimes X &  
B_6&=\frac{1}{2}\mathds{1}\otimes Y  &  
B_7&=\frac{1}{2}Y\otimes\mathds{1}  &  
B_{8}&=\frac{1}{2\sqrt{2}}(Y\otimes Z+Z\otimes Y)\\
B_9&= \frac{1}{2} Y\otimes Y &  
B_{10}&=\frac{1}{2}\mathds{1}\otimes Z &  
B_{11}&=\frac{1}{2}Z\otimes\mathds{1}  &  
B_{12}&=\frac{1}{2\sqrt{2}}(Z\otimes X+X\otimes Z) \\
B_{13}&= \frac{1}{2} Z\otimes Z
\end{align*}

\noindent 
Since the visible and invisible spaces bipartition the operator space, the invisible space is the complement of the invisible space. Using the invisible space definition above, those remaining strings are as follows:
\begin{align*}
\indent\indent B_{11,1}^{\perp}&= \frac{1}{2\sqrt{2}}(X\otimes Y-Y\otimes X)           &  
B_{12,1}^{\perp}&= \frac{1}{2\sqrt{2}}(Y\otimes Z-Z\otimes Y)              &  
B_{13,1}^{\perp}&= \frac{1}{2\sqrt{2}}(Z\otimes X-X\otimes Z)          &  
\end{align*}
\end{adjustwidth}

\begin{proposition}[\textbf{\textit{Size of Global $SU(2)$ Visible Space}}]
    For an $n$-qubit system, we have
    \begin{align}
    \dim \textnormal{\textsf{VisibleSpace}}(\text{\rm Global }SU(2)) = \frac{1}{8} 2^n( n^2 + 7n + 8)\,.
    \end{align}
\end{proposition}
\begin{proof}
Since each visible basis element $B$ is defined in terms of a fixed-$\mathds{1}$ permutation-invariant set $S$, the size of the visible space is the number of unique fixed-$\mathds{1}$ permutation-invariant sets.  Thus we have
\begin{eqnarray}
    \nonumber
    \dim \textnormal{\textsf{VisibleSpace}}(\text{\rm Global }SU(2)) & = & \text{number of unique sets $S$} \\
    \nonumber
    & = & \sum_{n_I = 0}^n (\#\text{ of ways to place $n_I$ $\mathds{1}$s)} \times (\#\text{ of ways to place $n - n_I$ non-$\mathds{1}$ Paulis)} \\
    \nonumber
    & = & \sum_{n_I = 0}^n \binom{n}{n_I}  \binom{n - n_I + 2}{n - n_I} \\
    \nonumber
    & = & \frac{1}{8} 2^n( n^2 + 7n + 8)\,.
\end{eqnarray}
\end{proof}

\subsection{\label{sec:su2MChannel}Global $SU(2)$ measurement channel}

Now we will derive the global $SU(2)$ measurement channel. Not only does it provide a convenient choice of $\mathcal{K}^O$ (as per Definition \ref{def:shadowversionofK}), but it will also help us prove Proposition \ref{prop:basisGlobalSu2VisInvis}, which establishes orthonormal bases for the for the global $SU(2)$ visible and invisible spaces. 
Before considering the specific case of global $SU(2)$, we will start with the original definition of the measurement channel $\mathcal{M}$ in the context of classical shadow tomography, and put it into a form that will help us derive $\mathcal{M}$ for global $SU(2)$.

\begin{lemma}\label{lem:NewFormMchannel}
    Consider a probability density function $p(V)$ over $\V$. Then the associated measurement channel \cite{huang2020predicting}
    can be written as follows:
    \begin{equation}\label{eqn:Mchannel3}
        \mathcal{M}(\rho)  = \text{\rm tr}_{B}[O_{AB}^{\V}(\mathds{1} \otimes \rho)]\,
    \end{equation}
    where $O_{AB}^{\V} \in \mathcal{B}(\mathcal{H}_A \otimes \mathcal{H}_B)$ for $\mathcal{H}_A \simeq \mathcal{H}_B \simeq \mathbb{C}^{2^n}$ is
    \begin{equation}\label{eqn:OAB}
        O_{AB}^{\V} := \int_{V \in \V} dV \, p(V)  \hspace{2mm} \frac{1}{2^n} \sum_{k=0}^{n} \frac{1}{k! (n-k)!} \sum_{\pi \in S_n} V^{\dagger}\pi(\mathds{1}^{\otimes n-k}Z^{\otimes k}) \pi^\dagger V\otimes V^{\dagger} \pi(\mathds{1}^{\otimes n-k}Z^{\otimes k}) \pi^\dagger V
    \end{equation}
    and the $\pi : (\mathbb{C}^2)^{\otimes n} \to (\mathbb{C}^2)^{\otimes n}$ are representations of the permutation group $S_n$.
\end{lemma}
\begin{proof} By basic tensorial manipulations, we can write
\begin{eqnarray}
\label{eqn:Mchannel1}
\mathcal{M}(\rho) & = & \int_{V \in \V} dV \, p(V) \sum_{b} V^{\dagger}\ketbra{b}{b} V \bra{b}V\rho V^{\dagger}\ket{b} \\
& = & \int_{V \in \V} dV \, p(V) \sum_{b} \Tr_{B}[(V^{\dagger}\ketbra{b}{b} V \otimes V^{\dagger}\ketbra{b}{b} V)(\mathds{1} \otimes \rho)] \\
& = & \Tr_{B}[Q(\mathds{1} \otimes \rho)]\,,\label{eqn:B11}
\end{eqnarray}
where Q is defined as
\begin{equation}\label{eqn:intermediateOabcsha}
    Q
    = 
    \int_{V \in \V} dV \, p(V)  \hspace{2mm} V^{\dagger \otimes 2} \hspace{1mm}  \sum_{b} \ketbra{bb}{bb} \hspace{1mm} V^{\otimes 2},
\end{equation}

It suffices to show that $Q$  takes the form of $O_{AB}^{\V}$ in~\eqref{eqn:OAB}. 
In addition to using the computational basis states $\ket{b}$, this expression for Q can also be written in terms of the $n$-qubit permutation group $S_n$. Notice that for a single qubit we have $\sum_{b \in \{0,1\}}\ketbra{bb}{bb} = \frac{1}{2} (\mathds{1}\mathds{1} + ZZ)$. This can be extended to $n$ qubits, where the expression becomes a sum of all length-$n$ $\mathds{1}$,$Z$ strings tensored twice:
\begin{equation}
    \sum_{b} \ketbra{bb}{bb} = \frac{1}{2^n} \sum_{k=0}^{n} \frac{1}{k! (n-k)!} \sum_{\pi \in S_n} \pi(\mathds{1}^{\otimes n-k}Z^{\otimes k}) \pi^\dagger \otimes\pi(\mathds{1}^{\otimes n-k}Z^{\otimes k}) \pi^\dagger\,.
\end{equation}
Substituting this into~\eqref{eqn:intermediateOabcsha}, we find the desired expression 
\begin{equation}
Q  =  \int_{V \in \V} dV \, p(V)  \hspace{2mm} \frac{1}{2^n} \sum_{k=0}^{n} \frac{1}{k! (n-k)!} \sum_{\pi \in S_n} V^{\dagger}\pi(\mathds{1}^{\otimes n-k}Z^{\otimes k}) \pi^\dagger V\otimes V^{\dagger} \pi(\mathds{1}^{\otimes n-k}Z^{\otimes k}) \pi^\dagger V,
\end{equation}
which matches~\eqref{eqn:OAB}.
\end{proof}

The above lemma shows us that the structure of our measurement channel $\mathcal{M}$ can be encapsulated by an operator $O_{AB}^{\V}$ acting on a $2n$-qubit system, where the first $n$-qubit subsystem is notated as $A$ and the second subsystem is notated as $B$.
Thus, the problem of finding an explicit expression for the $\mathcal{M}$ channel corresponds to finding an explicit expression for $O_{AB}^{\V}$.
Indeed, when we find an expression for the global $SU(2)$ measurement channel below, we will (i) explicitly evaluate $O_{AB}^{\V}$ and (ii) substitute the newly-derived $O_{AB}^{\V}$ expression into $\mathcal{M}(\rho)  = \Tr_{B}[O_{AB}^{\V}(\mathds{1} \otimes \rho)]$. 


\begin{proposition}[\textbf{\textit{Global $SU(2)$ Measurement Channel}}]\label{prop:globalsu2mmtchannel}
    Assume that $p(V)$ is uniform over $\V$. The global $SU(2)$ measurement channel can be written as
    \begin{equation}
        \label{eqn:su2Mch}
        \mathcal{M}_{SU(2)}(\rho) =
        \sum_{S,S'} c_{S,S'} \,\text{\rm tr}(B_{S'} \rho)\, B_{S}
    \end{equation}
    where $B_S$ is the operator defined in Observation~\ref{obs:BiAndBiPERPBasis} (see Appendix~\ref{sec:su2Visible}), and the coefficients $c_{S,S'}$ are defined below.  Let $S = S_{R, n_X, n_Y, n_Z}$ and $S' = S_{R', n_X', n_Y', n_Z'}$.  Then
    \begin{equation}\label{eqn:ctensor}
        c_{S,S'} = \begin{cases} 
      2 \frac{1}{\sqrt{n_X!\hspace{1mm} n_Y!\hspace{1mm} n_Z!\hspace{1mm} n_X'!\hspace{1mm} n_Y'!\hspace{1mm} n_Z'!}} 
        \hspace{1mm} \frac{(n_X + n_X')!(n_Y + n_Y')!(n_Z + n_Z')!}{(\frac{n_X + n_X'}{2})!(\frac{n_Y + n_Y'}{2})!(\frac{n_Z + n_Z'}{2})!}
        \hspace{1mm} \frac{K!(K+1)!}{(2K+2)!}  
        & \textnormal{conditions (a), (b)}
         \\
      0 & \textnormal{otherwise}
    \end{cases}
    \end{equation}
    where $K := \frac{1}{2} (n_X + n_Y + n_Z + n_X' + n_Y' + n_Z')$, and the two conditions are defined as follows: (a) $n_{\alpha} + n_{\alpha}' \text{ is even for all }\alpha \in \{X,Y,Z\}$, and (b) $R = R'$.
\end{proposition}
\begin{proof}
Evaluating the measurement channel for the global $SU(2)$ ensemble has two parts: (i) evaluating $O_{AB}^{\V} = O_{AB}^{\textnormal{Global } SU(2)}$ and (ii) substituting the newly-derived expression for $O_{AB}^{\textnormal{Global }SU(2)}$ into our expression for the $\mathcal{M}$ channel. First, we will use our knowledge of the global $SU(2)$ ensemble to evaluate $O_{AB}^{\textnormal{Global }SU(2)}$. We will start from~\eqref{eqn:OAB} in Lemma \ref{lem:NewFormMchannel}:
\begin{equation} \label{eqn:intermediateOAB}
    O_{AB}^{\textnormal{Global }SU(2)}  =
    \mathbb{E}_{V \sim SU(2)}  \hspace{2mm} \frac{1}{2^n} \sum_{k=0}^{n} \frac{1}{k! (n-k)!} \sum_{\pi \in S_n} \pi(\mathds{1}^{\otimes n-k}\otimes[V^{\dagger}ZV]^{\otimes k}) \pi^\dagger \otimes \pi(\mathds{1}^{\otimes n-k}\otimes[V^{\dagger}ZV]^{\otimes k}) \pi^\dagger\,.
\end{equation}
We can simplify this $2k$-qubit $O_{AB}^{\textnormal{Global }SU(2)}$ expression by evaluating $\mathbb{E}_{V \sim SU(2)} [V^{\dagger}ZV]^{\otimes 2k}$. Using the Euler parameterization for $SU(2)$, we write $V = V(\theta, \phi, \psi)$, where $V(\theta, \phi, \psi) = e^{i\sigma_3 \frac{\phi}{2}} e^{i\sigma_2 \frac{\theta}{2}} e^{i\sigma_3 \frac{\psi}{2}}$. Under an $SU(2)$ rotation the $Z$ Pauli transforms as $Z \rightarrow VZV^\dagger = \cos\theta \hspace{1mm} Z - \sin\theta (\cos\phi \hspace{1mm} X - \sin\phi \hspace{1mm} Y)$.  We write
\begin{eqnarray}
    \label{eqn:preexpansion2k}
    \mathbb{E}_{V \sim SU(2)} [V^{\dagger}ZV]^{\otimes 2k} 
    & = &  \mathbb{E}_{V(\theta, \phi, \psi) \sim SU(2)} [ (-\sin\theta \cos\phi) \hspace{1mm} X + (\sin\theta \sin\phi) \hspace{1mm} Y + (\cos\theta) \hspace{1mm} Z]^{\otimes 2k} 
    \\ \label{eqn:vXYZallform}
    & = & \sum_{\{n_X,n_Y,n_Z\, : \, n_X + n_Y + n_Z = 2k\}} v(n_X,n_Y,n_Z) \left[\sum_{\pi \in S_{2k}} \pi(X^{\otimes n_X} \otimes Y^{\otimes n_Y} \otimes Z^{\otimes n_Z})\pi^\dagger\right]\,.
\end{eqnarray}
By expanding the tensor $[ \,\cdot\, ]^{\otimes 2k}$ on the right side of~\eqref{eqn:preexpansion2k}, we find permutations of Pauli strings with $n_X$ $X$'s, $n_Y$ $Y$'s, and $n_Z$ $Z$'s such that $n_X+n_Y+n_Z=2k$. Notice that the bracketed term in~\eqref{eqn:vXYZallform} is just a sum of these Pauli strings. And we absorbed all constants into the function $v(n_X,n_Y,n_Z)$:
\begin{equation} \label{eqn:defofvnXnYnZ}
    v(n_X,n_Y,n_Z) = \frac{1}{n_X!\hspace{1mm}n_Y!\hspace{1mm}n_Z!}
    \hspace{2mm}
    \mathbb{E}_{V(\theta, \phi, \psi) \sim SU(2)} \hspace{1mm} (-\sin\theta \cos\phi)^{n_X} (\sin\theta \sin\phi)^{n_Y} (\cos\theta)^{n_Z}.
\end{equation}
The $n_X!\hspace{1mm}n_Y!\hspace{1mm}n_Z!$ contribution comes from canceling out distinct permutations that yield the same Pauli string in~\eqref{eqn:vXYZallform}.
Moreover, it turns out that we can further simplify this expression for $v(n_X,n_Y,n_Z)$ by integrating over $\theta$, $\phi$, and $\psi$ for $V(\theta, \phi, \psi) \in SU(2)$. Performing this integral we obtain
\begin{eqnarray}
    \mathbb{E}_{V(\theta, \phi, \psi) \sim SU(2)} \hspace{1mm} (&-&\sin\theta \cos\phi)^{n_X} (\sin\theta \sin\phi)^{n_Y} (\cos\theta)^{n_Z} \\
    &=& 
    \frac{1}{(4\pi)^2} \int_0^\pi d\theta \int_0^{2\pi} d\phi \int_0^{4\pi} d\psi \sin \phi (-\sin\theta \cos\phi)^{n_X} (\sin\theta \sin\phi)^{n_Y} (\cos\theta)^{n_Z} \\
    &=& \frac{1}{4} (1 + (-1)^{n_X})(1 + (-1)^{n_Y})(1 + (-1)^{n_Z}) \frac{n_X!\hspace{1mm}n_Y!\hspace{1mm}n_Z!}{\frac{n_X}{2}!\hspace{1mm}\frac{n_Y}{2}!\hspace{1mm}\frac{n_Z}{2}!} \frac{(k+1)!}{(2k+2)!}
\end{eqnarray}
Finally, substituting this into our expression for $v(n_X,n_Y,n_Z)$ from~\eqref{eqn:defofvnXnYnZ}, we find
\begin{eqnarray}
    v(n_X,n_Y,n_Z) &=& \frac{1}{4} (1 + (-1)^{n_X})(1 + (-1)^{n_Y})(1 + (-1)^{n_Z}) \frac{1}{\frac{n_X}{2}!\hspace{1mm}\frac{n_Y}{2}!\hspace{1mm}\frac{n_Z}{2}!} \frac{(k+1)!}{(2k+2)!} \\ \label{egn:v}
     &=&  \label{eqn:vnXnYnZCases}
     \begin{cases} 
      2 \hspace{1mm} [\frac{n_X}{2}!\hspace{1mm}\frac{n_Y}{2}!\hspace{1mm}\frac{n_Z}{2}!]^{-1} 
      \frac{(k+1)!}{(2k+2)!}  & n_X,n_Y,n_Z \textnormal{ all even} \\
      0 & \textnormal{otherwise}
     \end{cases}
\end{eqnarray}

Now that we have simplified our expression for $v(n_X,n_Y,n_Z)$, we can return to~\eqref{eqn:vXYZallform} and write the bracketed expression in terms of the basis elements defined in Observation \ref{obs:BiAndBiPERPBasis}.  Let us consider $k$-qubit basis elements, which have the property that $R = \{\}$, i.e.~there are no components equal to the identity. We call these basis elements $B_S^{(k)}$. These $B_S^{(k)}$ operators are otherwise defined exactly as in Observation~\ref{obs:BiAndBiPERPBasis}.
Using these particular operators,~\eqref{eqn:vXYZallform}'s bracketed term with $n_X$ $X$'s, $n_Y$ $Y$'s, and $n_Z$ $Z$'s can be rewritten a sum over terms like $B_{S_1}^{(k)} \otimes B_{S_2}^{(k)}$ where $(n_{X,1}+n_{X,2}) = n_X$, $(n_{Y,1}+n_{Y,2}) = n_Y$, and $ (n_{Z,1}+n_{Z,2})= n_Z$. 
After including the necessary normalization factors (see Observation \ref{obs:BiAndBiPERPBasis}) for the operators, we can rewrite~\eqref{eqn:vXYZallform} as 
\begin{align} \label{eqn:sumofXYZs}
    &\mathbb{E}_{V \sim SU(2)} [V^{\dagger}ZV]^{\otimes 2k} = \sum_{\Big\{S_1, S_2 \, : \,\substack{n_{X,1} + n_{X,2} = n_X \\ n_{Y,1} + n_{Y,2} = n_Y \\ n_{Z,1} + n_{Z,2} = n_Z} \Big\}}  v(n_{X,1}+n_{X,2},n_{Y,1}+n_{Y,2},n_{Z,1}+n_{Z,2})\\
    & \qquad \qquad \qquad \qquad \qquad \qquad \qquad \qquad \qquad \qquad \qquad \times
    \frac{2^n \hspace{1mm} k! \hspace{1mm} (n_{X,1}+n_{X,2})!(n_{Y,1}+n_{Y,2})!(n_{Z,1}+n_{Z,2})! }{\sqrt{n_{X,1}!\hspace{1mm}n_{Y,1}!\hspace{1mm}n_{Z,1}! \hspace{1mm} n_{X,2}!\hspace{1mm}n_{Y,2}!\hspace{1mm}n_{Z,2}!}}  
    B_{S_1}^{(k)} \otimes B_{S_2}^{(k)}\,. \nonumber
\end{align}
At this point, we have derived an entirely new expression for  $\mathbb{E}_{V \sim SU(2)} [V^{\dagger}ZV]^{\otimes 2k}$. 
Plugging this back into our equation for $O_{AB}^{\text{Global }SU(2)}$ in~\eqref{eqn:intermediateOAB}, we find that the permutations $\pi(\mathds{1}^{\otimes n-k}\otimes[V^{\dagger}ZV]^{\otimes k}) \pi^\dagger$ intersperse identities among the Pauli strings in our $B_{S}^{(k)}$ operators. These permutations yield $B_S$ operators on $n$ qubits with $|R| = n-k$.  Therefore, we can write
\begin{equation}
    O_{AB}  
    = \sum_{S_1,S_2} c_{S_1,S_2} \hspace{1mm} B_{S_1} \otimes B_{S_2},
\end{equation}
where the coefficients $c_{S_1,S_2}$ are defined as follows:
    \begin{equation}
       c_{S_1,S_2} = \begin{cases} 
      2 \frac{1}{\sqrt{n_{X,1}!\hspace{1mm}n_{Y,1}!\hspace{1mm}n_{Z,1}! \hspace{1mm} n_{X,2}!\hspace{1mm}n_{Y,2}!\hspace{1mm}n_{Z,2}!}} 
        \hspace{1mm} \frac{(n_{X,1}+n_{X,2})!(n_{Y,1}+n_{Y,2})!(n_{Z,1}+n_{Z,2})!}{(\frac{n_{X,1}+n_{X,2}}{2})!(\frac{n_{Y,1}+n_{Y,2}}{2})!(\frac{n_{Z,1}+n_{Z,2}}{2})!}
        \hspace{1mm} \frac{K!(K+1)!}{(2K+2)!}  
        & \textnormal{conditions (a), (b)}
         \\
      0 & \textnormal{otherwise}
    \end{cases}
    \end{equation}
    In this expression $K := \frac{1}{2} (n_{X,1} + n_{Y,1} + n_{Z,1} + n_{X,2} + n_{Y,2} + n_{Z,2})$, and we have the conditions: (a) $n_{\alpha,1} + n_{\alpha,2} \text{ is even for all }\alpha \in \{X,Y,Z\}$, and (b) $R_1 = R_2$.
Here condition (a) comes from Eq.~\eqref{eqn:vnXnYnZCases} for $v$. Condition (b)  requires that $B_{S_1}$ and $B_{S_2}$ have $\mathds{1}$ Paulis on the same sites, and so $R_1 = R_2$.

Since we now have an expression for $O_{AB}^{\textnormal{Global }SU(2)}$, we can substitute our newly-derived expression for $O_{AB}^{\textnormal{Global }SU(2)}$ into our definition of the $\mathcal{M}$ channel from Lemma \ref{lem:NewFormMchannel} (see~\eqref{eqn:Mchannel3}) and recover the expression stated in our proposition above.
\end{proof}

\subsection{\label{sec:ProofOfVisInvisBasis} Proof of visible and invisible space basis}

\noindent Here we prove Proposition~\ref{prop:basisGlobalSu2VisInvis}, which was first stated in Appendix~\ref{sec:su2Visible}.

\begin{repproposition}{prop:basisGlobalSu2VisInvis} [\textbf{A Basis for the Global $SU(2)$ Visible Space}]
    For $\V = \{U^{\otimes n}\}_{U\in SU(2)}$, an orthonormal basis for $\text{\rm\textsf{VisibleSpace}}(\V)$ is
\begin{align}
\bigcup_S \{B_S\}
\end{align}
and an orthonormal basis for $\text{\rm\textsf{InvisibleSpace}}(\V)$ is
\begin{align}
    \bigcup_{S} \,\{B_{S,1}^\perp, B_{S,2}^\perp,...,B_{S,|S|-1}^\perp\}\,.
\end{align}
\end{repproposition}
\begin{proof}
In Observation \ref{obs:BiAndBiPERPBasis}, we found that the $B$'s and $B^\perp$'s together form an orthonormal basis for the operator space $\mathcal{B}(\mathcal{H})$. 
Therefore, to show that these two sets of operators constitute orthonormal bases for the invisible and visible spaces, respectively, we will show that all the $B$'s are entirely in the visible space and all the $B^{\perp}$'s are entirely in the invisible space.

\vspace{2mm}
\noindent\textit{(Invisible Space Basis)} Consider any $V \in SU(2)$ and computational basis state $|b\rangle$. We will show that for all invisible basis elements $B_{S,k}^\perp$, we have $\bra{b} V^{\otimes n} B_{S,k}^\perp V^{\dagger \otimes n} \ket{b} = 0$, and thus $B_{S,k}^\perp$ must be in $\text{\textsf{InvisibleSpace}}(\text{Global }SU(2))$. We will start by substituting in our proposed expression for $B_{S,k}^\perp$. 
\begin{eqnarray} \label{eqn:bkperpbb}
    \bra{b} V^{\otimes n} B_{S,k}^\perp V^{\dagger \otimes n} \ket{b} 
    & = & \bra{b} V^{\otimes n} \cdot\frac{1}{\sqrt{2^n k(k+1)}} \left(\sum_{j = 1}^k P_j - k P_{k+1}\right) \cdot V^{\dagger \otimes n} \ket{b} \\ 
    & = & \frac{1}{\sqrt{2^n k(k+1)}}
    \left(\sum_{j = 1}^k 
    \bra{b} V^{\otimes n} P_j V^{\dagger \otimes n} \ket{b} - k \bra{b} V^{\otimes n} P_{k+1} V^{\dagger \otimes n} \ket{b} \right)\,.
\end{eqnarray}
Recall that $B_{S,k}^\perp$ is constructed using some fixed-$\mathds{1}$ permutation-invariant set $S$ containing Pauli strings $P$ with (i) $\mathds{1}$s on the same sites and (ii) $n_X$ $X$ Paulis, $n_Y$ $Y$ Paulis, and $n_Z$ $Z$ Paulis on any of the remaining sites. Using the Euler parameterization of $SU(2)$, we can evaluate $\bra{b} V^{\otimes n} P V^{\dagger \otimes n} \ket{b}$ for all Pauli strings $P \in S$.  We start by writing $\bra{b} V^{\otimes n} P V^{\dagger \otimes n} \ket{b}$ as
\begin{eqnarray} \label{eqn:eulerrep1}
    \bra{b} V^{\otimes n} P V^{\dagger \otimes n} \ket{b}
    & = & \bra{b} V(\theta, \phi, \psi)^{\otimes n} P V(\theta, \phi, \psi)^{\dagger \otimes n} \ket{b} \\ \label{eqn:eulerrep2}
    & = & (\sin\theta \hspace{1mm} \cos\phi)^{n_X}(\sin\theta \hspace{1mm} \sin\phi)^{n_Y}(\cos\theta)^{n_Z} \bra{\Tilde{b}} Z^{\otimes (n_X + n_Y + n_Z)} \ket{\Tilde{b}} \,.
\end{eqnarray}
In the expression above, the unitary transformation $V(\theta, \phi, \psi)^{\otimes n}$ rotates $P$ to a linear combination of other Pauli strings. These other Pauli strings will all have $\mathds{1}$ on the same sites, and \textit{one} of them will have $Z$ on all remaining sites. This particular Pauli string of only $\mathds{1}$'s and $Z$'s will be the only term that survives the expectation $\bra{b} \,\cdot\, \ket{b}$ on the right-hand side of~\eqref{eqn:eulerrep1}, because all other terms contain $X$ or $Y$ Paulis. As a result,~\eqref{eqn:eulerrep2} only retains this particular Pauli string of $\mathds{1}$'s and $Z$'s, and $\ket{\Tilde{b}} \in (\mathbb{C}^2)^{\otimes(n_X + n_Y + n_Z)}$ represents the projection of $\ket{b}$ onto the sites with the $Z$ Paulis.
Finally, notice that equation~\eqref{eqn:eulerrep2} is the \textit{same} for all Pauli strings $P \in S$, and therefore we find 
\begin{eqnarray}\nonumber
    \bra{b} V^{\otimes n} B_{S,k}^\perp V^{\dagger \otimes n} \ket{b} 
    & = & \frac{1}{\sqrt{2^n k(k+1)}}
    (\sin\theta \hspace{1mm} \cos\phi)^{n_X}(\sin\theta \hspace{1mm} \sin\phi)^{n_Y}(\cos\theta)^{n_Z} \bra{\Tilde{b}} Z^{\otimes (n_X + n_Y + n_Z)} \ket{\Tilde{b}} 
    \left(\sum_{j = 1}^k 
    1 - k\right) \\ \nonumber
     &=& 0\,.
\end{eqnarray}

\vspace{2mm}
\noindent\textit{(Visible Basis)} Following Proposition \ref{prop:globalsu2mmtchannel}, the image of the global $SU(2)$ measurement channel is the span of the $B_S$ basis elements. Therefore, since all $B_S$'s are in the image of the measurement channel, by Observation \ref{obs:mChEquivalenceVIS} all $B_S$'s are are also in the visible space.
\end{proof}

\section{\label{app:globalCL2}Another form of limited control: global $\text{Cl}(2)$ }

Another form of limited control is $\V = \{U^{\otimes n}\}_{U \in \textnormal{Cl}(2)}$. For this form of control, which we call ``global $\text{Cl}(2)$ control,'' we choose either the $X$, $Y$, or $Z$ basis and then measure \textit{every} qubit in that basis.  Here we provide an explicit orthonormal basis of operators for the global $\text{Cl}(2)$ visible space, derive the associated measurement channel $\mathcal{M}$, and bound the sample complexity of our measurement protocol with global $\text{Cl}(2)$ control. Since our results with this ensemble are straightforward and intuitive, we use the ensemble in main text to develop intuition for Theorem~\ref{thm:shadowtomtheorem}.

\subsection{\label{subsec:cl2mmtchannel} Global  $\text{Cl}(2)$ measurement channel}

First, we will derive the measurement channel $\mathcal{M}$ for global $\text{Cl}(2)$. This channel will allow us to define the global $\text{Cl}(2)$ visible space and then upper bound the sample complexity when using this ensemble with the choice of kernel $\mathcal{K} = \mathcal{K}_\textnormal{CS}$. We will find that the global $\text{Cl}(2)$ visible space contains all Pauli strings made up of the identity and all $X$, all $Y$, or all $Z$ Paulis. 
In other words, the visible space for $n$ qubits contains all strings in the set $\bigcup_{\sigma \in \{X,Y,Z\}} \text{Pauli}_{\mathds{1},\sigma}(n)$, where $\text{Pauli}_{\mathds{1},\sigma}(n)$ is the set of $n$-qubit Pauli strings constructed from only $\mathds{1}$'s and $\sigma$'s.

\begin{proposition}[\textbf{\textit{Global $\text{\rm Cl}(2)$ Measurement Channel}}]\label{prop:globalCl2mmtchannel}
    Assume that $p(V)$ is uniform over $\V$. The global \text{\rm Cl}$(2)$ measurement channel for $n$ qubits can be written as
    \begin{equation}
        \label{eqn:CL2Mch}
        \mathcal{M}_{\text{\rm Cl}(2)}(A) =  \frac{1}{3 \cdot 2^n} \sum_{\sigma \in \{X,Y,Z\}} \sum_{P \in \text{\rm Pauli}_{\mathds{1},\sigma}(n)}  \text{\rm tr}(A P) P\,,
    \end{equation}
    where $\text{\rm Pauli}_{\mathds{1},\sigma}(n)$ is the set of $n$-qubit Pauli strings constructed from only $\mathds{1}$'s and $\sigma$'s. Note that the normalization of a Pauli string $P$ in this set is $\text{\rm tr}(P P^\dagger) = 2^n$.
\end{proposition}
\begin{proof}
Evaluating the measurement channel for the global $\text{Cl}(2)$ ensemble has two parts: (i) evaluating $O_{AB}^{\V} = O_{AB}^{\textnormal{Global } \textnormal{Cl}(2)}$ and (ii) substituting the newly-evaluated $O_{AB}^{\textnormal{Global } \textnormal{Cl}(2)}$ into our expression for the measurement channel from Lemma \ref{lem:NewFormMchannel}. First we will evaluate $O_{AB}^{\textnormal{Global } \textnormal{Cl}(2)}$. We will start from~\eqref{eqn:OAB} in Lemma \ref{lem:NewFormMchannel}:
\begin{equation}
    O_{AB}^{\textnormal{Global Cl}(2)} = \int_{V \in \V} dV \, p(V)  \hspace{2mm} \frac{1}{2^n} \sum_{k=0}^{n} \frac{1}{k! (n-k)!} \sum_{\pi \in S_n} V^{\dagger}\pi(\mathds{1}^{\otimes n-k}Z^{\otimes k}) \pi^\dagger V\otimes V^{\dagger} \pi(\mathds{1}^{\otimes n-k}Z^{\otimes k}) \pi^\dagger V\,.
\end{equation}
Setting $V = U^{\otimes n}$ where $U \in \text{Cl}(2)$, we can simplify this expression by evaluating $\mathbb{E}_{U \sim \textnormal{Cl}(2)} [U^{\dagger}ZU]^{\otimes k}$ as
\begin{equation}
    \mathbb{E}_{U \sim \textnormal{Cl}(2)} [U^{\dagger}ZU]^{\otimes k} = \frac{1}{3} \big( X^{\otimes k} + Y^{\otimes k} + Z^{\otimes k}\big)\,.
\end{equation}
We have probability $\frac{1}{3}$ each of measuring in the $X$, $Y$, or $Z$ bases. Therefore, the operator $O_{AB}^{\textnormal{Global Cl}(2)}$ becomes 
\begin{equation}
    O_{AB}^{\textnormal{Global Cl}(2)} = \frac{1}{3 \cdot 2^n} \sum_{\sigma \in \{X,Y,Z\}} \sum_{k=0}^{n} \frac{1}{k! (n-k)!} \sum_{\pi \in S_n} \left[ \pi(\mathds{1}^{\otimes n-k}\sigma^{\otimes k}) \pi^\dagger \right]^{\otimes 2}.
\end{equation}
Next, notice that the summand of the $\sigma$-indexed sum is simply a sum of Pauli strings (tensored twice) made up of only $\mathds{1}$ and $\sigma$ single-site Paulis.
Using the set $\text{Pauli}_{\mathds{1},\sigma}(n)$, which contains all $n$-qubit Pauli strings constructed from only $\mathds{1}$'s and $\sigma$'s, we can express $O_{AB}^{\textnormal{Global Cl}(2)}$ as
\begin{equation}
    O_{AB}^{\textnormal{Global Cl}(2)} = \frac{1}{3 \cdot 2^n} \sum_{\sigma \in \{X,Y,Z\}} \sum_{P \in \text{Pauli}_{\mathds{1},\sigma}(n)}  P^{\otimes 2}.
\end{equation}
Plugging this back into~\eqref{eqn:Mchannel3} from Lemma \ref{lem:NewFormMchannel}, we find the desired expression  
\begin{equation}
    \mathcal{M}_{\text{\rm Cl}(2)}(A) =  \frac{1}{3 \cdot 2^n} \sum_{\sigma \in \{X,Y,Z\}} \sum_{P \in \text{Pauli}_{\mathds{1},\sigma}(n)}  \tr(A P) P.
\end{equation}
\end{proof}

\begin{corollary} [\textbf{The Global} $\text{Cl}(2)$ \textbf{Visible Space}]
    For $\V = \{U^{\otimes n}\}_{U\in \textnormal{Cl}(2)}$, a basis for $\text{\rm\textsf{VisibleSpace}}(\V)$ is all Pauli strings in the set
\begin{align}
\bigcup_{\sigma \in \{X,Y,Z\}} \text{\rm Pauli}_{\mathds{1},\sigma}(n)\,,
\end{align}
where $\text{Pauli}_{\mathds{1},\sigma}(n)$ is the set of $n$-qubit Pauli strings constructed from only $\mathds{1}$'s and $\sigma$'s.
\end{corollary}
\begin{proof}
By Proposition \ref{prop:globalCl2mmtchannel}, the image of the global $\text{Cl}(2)$ measurement channel is the span of all Pauli strings in $\text{Pauli}_{\mathds{1},X}(n)$, $\text{Pauli}_{\mathds{1},Y}(n)$, or $\text{Pauli}_{\mathds{1},Z}(n)$. Since by Observation  \ref{obs:mChEquivalenceVIS} the visible space is the image of the measurement channel, the visible space is this same span of all Pauli strings in $\text{Pauli}_{\mathds{1},X}(n)$, $\text{Pauli}_{\mathds{1},Y}(n)$, or $\text{Pauli}_{\mathds{1},Z}(n)$ 
And note that since Pauli strings are orthogonal, this basis for the global $\text{Cl}(2)$ visible space is not overcomplete. 
\end{proof}

\subsection{\label{subsec:cl2mmtsamplecomplexity} Global $\text{Cl}(2)$ sample complexity}

Now that we have an analytic expression for the global $\text{Cl}(2)$ measurement channel, we can provide rigorous guarantees on the sample complexity when using this ensemble and the estimator $\mathcal{K} = \mathcal{K}_\textnormal{CS}$. In particular, we can upper bound Var$_\text{max} \mathcal{K}^O_\textnormal{CS}$ in Theorem \ref{thm:shadowtomtheorem} with the shadow norm \cite{huang2020predicting}. This canonical object in classical shadow tomography takes the form ~\cite{huang2020predicting}
\begin{equation}
    \|O\|^2_{\textnormal{shadow}} = \max_{\sigma \,:\, \textnormal{state}} \tr \big( \sigma \Lambda(\mathcal{M}^{-1}(O)) \big) \,,
\end{equation}
where $\mathcal{M}$ is the global $\text{Cl}(2)$ measurement channel, for which we just derived an analytic expression (Eq.~\eqref{eqn:CL2Mch}). 
We define the nonlinear ``shadow map'' $\Lambda: \mathcal{B}(\mathcal{H}) \rightarrow \mathcal{B}(\mathcal{H})$ below. 

\begin{definition}[\textbf{Shadow Map}] \label{def:shadowmap} Assume that $p(V)$ is uniform over $\V$. The shadow map over $n$ qubits is
\begin{equation}\label{eqn:shadowmapdef}
    \Lambda (O) = \Expect_{V \sim \V} \sum_{b \in \{0,1\}^n} V^\dagger \ketbra{b}{b} V \bra{b} V O V^\dagger \ket{b}^2.
\end{equation}
\end{definition}

\begin{lemma}\label{lem:NewFormShadowMapchannel}
    Consider a probability density function $p(V)$ over $\V$. Then the associated shadow map can be written as follows:
    \begin{equation} \label{eqn:fullshadowmaplemma3}
        \Lambda (O)  = \text{\rm tr}_{AB}[O_{ABC}^{\V}(O \otimes O \otimes \mathds{1})]\,
    \end{equation}
    where $O_{ABC}^{\V} \in \mathcal{B}(\mathcal{H}_A \otimes \mathcal{H}_B \otimes \mathcal{H}_C)$ for $\mathcal{H}_A \simeq \mathcal{H}_B \simeq \mathcal{H}_C \simeq \mathbb{C}^{2^n}$ is
    \begin{equation}\label{eqn:OABC}
        O_{ABC}^{\V} := \frac{1}{4^n} \Expect_{V \sim \V}    \sum_{P_1, P_2 \in \text{\rm Pauli}_{\mathds{1},Z}(n)}  V^{\dagger}P_1V \otimes V^{\dagger}P_2 V \otimes V^{\dagger}P_1 P_2 V
    \end{equation}
    and $\text{Pauli}_{\mathds{1},Z}(n)$ is the set of $n$-qubit Pauli strings constructed from only $\mathds{1}$'s and $Z$'s.
\end{lemma}
\begin{proof}
The expression for the shadow map (Eq.~\eqref{eqn:shadowmapdef}) can be rewritten as
\begin{eqnarray}
    \Lambda(O) 
    &=& \tr_{AB}\left[\Expect_{V \sim \V} \sum_b (V^\dagger \ketbra{b}{b} V)^{\otimes 3} (  O \otimes O \otimes\mathds{1}) \right] \\
    &=& \tr_{AB}\left[Q (  O \otimes O \otimes\mathds{1}) \right] 
\end{eqnarray}
where $Q$ is defined as
\begin{equation}
    Q
    = \Expect_{V \sim \V} \sum_b (V^\dagger \ketbra{b}{b} V)^{\otimes 3}.
\end{equation}
The inside of the trace lives in the Hilbert space $\mathcal{H}_A \otimes \mathcal{H}_B \otimes \mathcal{H}_C$, and subsystems $A$, $B$, $C$ run from left to right. At this point, it suffices to show that  $Q$ takes the form of $O_{ABC}^{\V}$ in~\eqref{eqn:OABC}. Rearranging our expression for $Q$, we can consider the sum over all $\ketbra{bbb}{bbb}$ as 
\begin{equation}\label{eqn:rewrite3xVbbv}
    \Expect_{V \sim \V} \sum_b (V^\dagger \ketbra{b}{b} V)^{\otimes 3} = \Expect_{V \sim \V} V^{\dagger \otimes 3} \sum_b \ketbra{bbb}{bbb}  V^{\otimes 3}\,.
\end{equation}
For a single qubit corresponding to $n=1$, the sum $\sum_b \ketbra{bbb}{bbb}$ takes the form 
\begin{eqnarray}
    \sum_{b\in\{0,1\}} \ketbra{bbb}{bbb} &=& \frac{1}{4}(\mathds{1}\otimes\mathds{1}\otimes\mathds{1} + \mathds{1}\otimes Z \otimes Z + Z\otimes\mathds{1}\otimes Z + Z \otimes Z \otimes\mathds{1}) \\
    &=& \frac{1}{4}  \sum_{\sigma_A, \sigma_B \in \{\mathds{1}, Z\}}  \sigma_A \otimes \sigma_B \otimes \sigma_A\sigma_B,
\end{eqnarray}
where $\sigma_A$ and $\sigma_B$ are single qubit Pauli $Z$'s or $\mathds{1}$'s.
We can readily generalize this to $n$ qubits and find
\begin{equation}
    \sum_{b} \ketbra{bbb}{bbb}  =  \frac{1}{4^n}  \sum_{P_1, P_2 \in \text{Pauli}_{\mathds{1},Z}(n)}  P_1 \otimes P_2 \otimes P_1 P_2\,,
\end{equation}
where $\text{Pauli}_{\mathds{1},Z}(n)$ is the set of $n$-qubit Pauli strings constructed from only $\mathds{1}$'s and $Z$'s. Plugging this back into equation \eqref{eqn:rewrite3xVbbv}, we find the desired expression  
\begin{equation}
    Q = \frac{1}{4^n} \Expect_{V \sim \V}    \sum_{P_1, P_2 \in \text{Pauli}_{\mathds{1},Z}(n)}  V^{\dagger}P_1V \otimes V^{\dagger}P_2 V \otimes V^{\dagger}P_1 P_2 V
\end{equation}
which matches~\eqref{eqn:OABC}.

\end{proof}

\begin{proposition}[\textbf{\textit{Global} $\text{\rm Cl}(2)$ \textbf{\textit{Shadow Map}}}]\label{prop:globalCl2shadowchannel}
    Assume that $p(V)$ is uniform over $\V$. The global $\text{\rm Cl}(2)$ shadow map for $n$ qubits can be written as
    \begin{equation}
        \label{eqn:CL2shadowmap}
        \Lambda_{\text{\rm Cl}(2)}(O) =\Lambda_{\text{\rm Cl}(2)}(O) = \frac{1}{3 \cdot 4^n} \sum_{\sigma \in \{X,Y,Z\}} \sum_{P_1, P_2 \in \text{Pauli}_{\mathds{1},\sigma}(n)}  \text{\rm tr}(P_1 O) \,\text{\rm tr}( P_2 O) P_1 P_2\,.
    \end{equation}
    where $\text{Pauli}_{\mathds{1},\sigma}(n)$ is the set of $n$-qubit Pauli strings constructed from only $\mathds{1}$'s and $\sigma$'s. Note that the normalization of a Pauli string $P$ in this set is $\text{\rm tr}(P P^\dagger) = 2^n$.
\end{proposition}
\begin{proof}
Similar to our derivation of the measurement channel, evaluating the shadow map for the global $\text{Cl}(2)$ ensemble has two parts: (i) evaluating $O_{ABC}^{\V} = O_{ABC}^{\textnormal{Global } \textnormal{Cl}(2)}$ and (ii) substituting the newly-derived expression for $O_{ABC}^{\textnormal{Global } \textnormal{Cl}(2)}$ into our expression for the shadow map. First, we will evaluate $O_{ABC}^{\textnormal{Global } \textnormal{Cl}(2)}$. We will start from~\eqref{eqn:OABC} in Lemma \ref{lem:NewFormShadowMapchannel}:
\begin{equation} 
    O_{ABC}^{\textnormal{Global Cl}(2)}  = \frac{1}{4^n} \Expect_{V \sim \V}    \sum_{P_1, P_2 \in \text{Pauli}_{\mathds{1},Z}(n)}  V^{\dagger}P_1V \otimes V^{\dagger}P_2 V \otimes V^{\dagger}P_1 P_2 V\,.
\end{equation}
Each term within this sum has some number, $k$, of $Z$ Paulis across the $3n$ qubit system. Setting $V = U^{\otimes n}$ where $U \in \text{Cl}(2)$, we can simplify the expression by evaluating $\mathbb{E}_{U \sim \textnormal{Cl}(2)} [U^{\dagger}ZU]^{\otimes k}$ as
\begin{equation}
    \mathbb{E}_{U \sim \textnormal{Cl}(2)} [U^{\dagger}ZU]^{\otimes k} = \frac{1}{3} \big( X^{\otimes k} + Y^{\otimes k} + Z^{\otimes k}\big)\,.
\end{equation}
We have a probability of $\frac{1}{3}$ of measuring in the $X$, $Y$, or $Z$ bases. Therefore, the operator $O_{ABC}^{\textnormal{Global Cl}(2)}$ becomes 
\begin{equation}
    O_{ABC}^{\textnormal{Global Cl}(2)} = \frac{1}{3 \cdot 4^n} \sum_{\sigma \in \{X,Y,Z\}} \sum_{P_1, P_2 \in \text{Pauli}_{\mathds{1},\sigma}(n)}  P_1 \otimes P_2 \otimes P_1 P_2\,,
\end{equation}
where $\text{Pauli}_{\mathds{1},\sigma}(n)$ is the set of $n$-qubit Pauli strings constructed from only $\mathds{1}$'s and $\sigma$'s. Plugging this back into equation \eqref{eqn:fullshadowmaplemma3}, we find the desired expression  
\begin{equation}
    \Lambda_{\text{\rm Cl}(2)}(O) = \frac{1}{3 \cdot 4^n} \sum_{\sigma \in \{X,Y,Z\}} \sum_{P_1, P_2 \in \text{Pauli}_{\mathds{1},\sigma}(n)}  \tr(P_1 O) \tr( P_2 O) P_1 P_2\,.
\end{equation}


\end{proof}

Now that we have an analytic expression for the global $\text{Cl}(2)$ measurement channel and shadow map, we can evaluate the shadow norm of any operator $O$. For example, consider the operator $O = X^{\otimes k}$. This operator is an eigenvector of $\mathcal{M}_{\textnormal{Cl}(2)}$ with eigenvalue $\frac{1}{3}$. One can easily check this by evaluating $\mathcal{M}_{\textnormal{Cl}(2)}(X^{\otimes k})$ using~\eqref{eqn:CL2Mch}. Moreover, we also find that $\Lambda_{\textnormal{Cl}(2)}(X^{\otimes k}) = \frac{1}{3} \mathds{1}$, which one can check using equation \eqref{eqn:CL2shadowmap}. Therefore, we find that the shadow norm of $X^{\otimes k}$ is
\begin{eqnarray}
    \|X^{\otimes k}\|^2_{\textnormal{shadow}} &=& \max_{\sigma \,:\, \textnormal{state}} \tr \big( \sigma \Lambda(\mathcal{M}^{-1}(X^{\otimes k})) \big) \\
    &=& 3^2 \max_{\sigma \,:\, \textnormal{state}} \tr \big( \sigma \Lambda(X^{\otimes k}) \big) \\
    &=& 3 \max_{\sigma \,:\, \textnormal{state}} \tr \big( \sigma \mathds{1} \big) \\
    &=& 3\,.
\end{eqnarray}
Accordingly, the shadow norm of $X^{\otimes k}$ is $\mathcal{O}(1)$ independent of $k$. This makes sense: with our set of implementable unitaries $\V$, we measure all qubits in the $X$ basis, all qubits in the $Y$ basis, or all qubits in the $Z$ basis. Therefore, whenever we measure all qubits in the $X$ basis, we can reconstruct the expectation value of $X$ just as easily as for $X^{\otimes n}$.


\section{\label{app:OptimizingBiasVar}Minimizing on the bias-variance tradeoff}

Here we discuss how to estimate $\tr(\rho O)$ with the smallest error given a fixed number of measurements $N$ and a fixed $p(V)$ with support $\V$.  Using Theorem~\ref{thm:shadowtomtheorem}, instead of estimating $\tr(\rho O)$ directly, we can estimate $\tr(\rho \widetilde{O})$ for some $\widetilde{O} \in \text{\textsf{VisibleSpace}}(\V)$ close to $O$. Our $N$-measurement estimate for $\tr(\rho \widetilde{O})$ takes the form 
\begin{equation}
    \tr(\rho \widetilde{O}) \approx \frac{1}{N} \sum_{i=1}^N \mathcal{K}^{\widetilde{O}}(V_i,b_i)
\end{equation}
(see Appendix~\ref{app:VisibleSpace} for details), and via Theorem \ref{thm:shadowtomtheorem} the corresponding estimation error is upper bounded by
\begin{equation}\label{eqn:fullestimationerror}
    \left|\tr(\rho O) - \frac{1}{N}\sum_{i=1}^N \mathcal{K}^{\widetilde{O}}(V_i,b_i) \right| \leq \|O-\widetilde{O}\|_{\infty} + \epsilon_{\widetilde{O}}\,.
\end{equation}
The left side of inequality~\ref{eqn:fullestimationerror}, namely the estimation error, is upper bounded by two terms. The first term, the spectral norm $\|O-\widetilde{O}\|_{\infty}$, bounds the error of estimating the \textit{biased} operator $\widetilde{O}$ instead of $O$.  The $\epsilon_{\widetilde{O}}$ term arises from using a finite number of measurements to estimate $\tr(\rho\widetilde{O})$ and is related to the \textit{variance} of $\mathcal{K}^{\widetilde{O}}$. 

Therefore, to reduce the estimation error on $\tr(\rho O)$, we want to find the $\widetilde{O}$ that minimizes the right hand side of~\eqref{eqn:fullestimationerror}.
In other words, we want to perturb $O \rightarrow \widetilde{O}$ such that the variance term $\epsilon_{\widetilde{O}}$ substantially decreases, but we do not want to perturb it so much that the bias term $\|O-\widetilde{O}\|_{\infty}$ explodes. Accordingly, we refer to finding the optimal $\widetilde{O}$ as ``minimizing on the bias-variance tradeoff,'' and we mathematically define $\widetilde{O}$ as the operator that minimizes the cost function
\begin{equation}\label{eqn:thecostfunctionappx}
    \textnormal{Cost}(\widetilde{O}) 
    := 
    \|O-\widetilde{O}\|_{\infty} 
    +
    \sqrt{\frac{2}{N}\, \textnormal{Var}_{V,b\sim P_{\mathds{1}/2^n}}[\mathcal{K}^{\widetilde{O}}(V, b)]  
    \hspace{1mm} \log\!\left(\frac{M}{2\delta}\right)}\,.
\end{equation}
Here $M$ is the number of observables we wish to estimate using the randomized measurement data, and the probability of obtaining the outcome $(V,b)$ in a randomized measurement is $P_\rho(V,b) = p(V)\bra{b}V\rho V^\dagger\ket{b}$.  Observe that in the variance term we have opted to sample with respect to $P_{\mathds{1}/2^n}$.  This choice 
makes our cost function convex -- notice that the cost function would not be convex if we used $\textnormal{Var}_\textnormal{max}$.  
Our choice should be regarded as a heuristic since we are interested in computing expectation values $\text{tr}(\rho O)$ where $\rho$ need not be maximally mixed.  When we consider numerical examples later on, we will find that our heuristic leads to favorable bias-variance tradeoffs in practice when $\rho$ is not maximally mixed.

The cost function above can be motivated by solving for $\epsilon_{\widetilde{O}}$ in~\eqref{eqn:fullestimationerror} using Theorem~\ref{thm:shadowtomtheorem} and then  discarding the $\frac{1}{3}Q\epsilon$ term.  This discarding is heuristically justified if we are in a regime where it is very small compared to the variance term. 

\vspace{3mm}
\begin{adjustwidth}{1cm}{}
\noindent\textit{Method for sampling operators on the bias-variance tradeoff.}  One can always tune the cost function to allow more or less bias by introducing a parameter $\alpha$ in front of the bias term:
\begin{equation}\label{eqn:alphaparameterizedcostfcn}
    \textnormal{Cost}_\alpha(\widetilde{O}) 
    := 
    \alpha \,\|O-\widetilde{O}\|_{\infty} 
    +
    \sqrt{\frac{2}{N}\, \textnormal{Var}_{V,b\sim P_{\mathds{1}/2^n}}[\mathcal{K}^{\widetilde{O}}(V, b)]  
    \hspace{1mm} \log\!\left(\frac{M}{2\delta}\right)}.
\end{equation}
Introducing $\alpha$ allows one to \textit{choose} where the estimator sits along bias-variance tradeoff. Choosing $\alpha \ll 1$ will yield a highly-biased estimator, while choosing  $\alpha \gg 1$ forces the estimator's bias to be very small. To find 
a $\widetilde{O}$ with favorable estimation properties, consider the following steps:
\begin{enumerate}
    \item Scan over many $\alpha$. 
    \item For each $\alpha$, find the operator $\widetilde{O}_\alpha$ that minimizes the $\alpha$-parameterized cost function~\eqref{eqn:alphaparameterizedcostfcn} and calculate the corresponding $\textnormal{Var}_\textnormal{max}$ estimation error (e.g. the right side of ~\eqref{eqn:fullestimationerror} where the variance in $\epsilon_{\widetilde{O}}$ is upper bounded with $\textnormal{Var}_\textnormal{max}$)
    \item Choose the $\widetilde{O}_\alpha \in \{\widetilde{O}_\alpha\}_\alpha$ that yields the smallest $\textnormal{Var}_\textnormal{max}$ estimation error.
\end{enumerate}

\end{adjustwidth}

\vspace{6mm}
The rest of this section provides methods for minimizing the cost function with respect to the operator $\widetilde{O}$. 
Since we parameterize operators $\widetilde{O}$ in the visible space via
\begin{equation}
    \widetilde{O} = \int_{V \in \V} dV \, p(V) \sum_b \mathcal{K}^{\widetilde{O}}(V,b) \hspace{1mm} V^{\dagger} \ketbra{b}{b} V,
\end{equation}
finding the optimal $\widetilde{O}$ corresponds to finding the $\mathcal{K}^{\widetilde{O}}(V,b)$ that minimizes the estimation error cost function. In general, it is impossible to completely optimize for the best $\mathcal{K}^{\widetilde{O}}(V,b)$; since in the most general setting we would have to define $\mathcal{K}^{\widetilde{O}}(V,b)$ for all $V \in \V$ and computational basis states $|b\rangle$, the number of parameters over which we need to optimize grows exponentially with the system size. However, for certain cases we can employ tricks to find a $\mathcal{K}^{\widetilde{O}}(V,b)$ that minimizes our estimation error.
We will discuss two optimization methods in the subsections which follow. The first and simpler method works for small system sizes. The second method derives a new parameterization for $\widetilde{O}$ in order to be able to efficiently optimize on the bias variance tradeoff for local $O$ on a large system so long as each $V \in \V$ can be written as a low-depth quantum circuit.

\subsection{Optimization for small system sizes}\label{subsec:appBVSmallSystems}
For small system sizes, we can optimize on the bias variance tradeoff to find the $\mathcal{K}^{\widetilde{O}}(V,b)$ such that the operator $\widetilde{O}$ minimizes the estimation error ``cost function''~\eqref{eqn:thecostfunctionappx}. If $\V$ is a continuous set, it is useful to replace it with a finite set, which can be achieved by subsampling.  For the moment, let us assume that $\V$ is a discrete and finite set. We parameterize $\widetilde{O}$ as
\begin{align}
\widetilde{O}(\mathcal{K}) = \sum_{V \in \V} p(V) \sum_b \mathcal{K}(V,b) \, V^\dagger |b\rangle \langle b| V\,,
\end{align}
where here $\mathcal{K}(V,b)$ is a $|\V| \times 2^n$ vector.
To find the optimal $\widetilde{O}$ which minimizes our cost function, we can treat the cost function as being a function of $\mathcal{K}$,
\begin{equation}\label{eqn:piecewisekcostfunctionk}
    \textnormal{Cost}(\mathcal{K}) = \|O-\widetilde{O}(\mathcal{K})\|_{\infty}
    +
    \sqrt{\frac{2}{N}\, \textnormal{Var}_{V,b\sim P_{\mathds{1}/2^n}}\left[\mathcal{K}(V, b)\right]  
    \hspace{1mm} \log\left(\frac{M}{2\delta}\right)}\,,
\end{equation} 
and then minimize over $\mathcal{K}$.  Since the cost function is convex in the entries of $\mathcal{K}$,  we simply need to employ a convex optimization algorithm (e.g.~gradient descent). Solving this convex optimization problem is feasible when $|\V|$ and the number of qubits $n$ are sufficiently small.



\subsection{Optimization for local operators and low-depth implementable unitaries}\label{subsec:appBVLargeSystems}

If we have a large system size, we can still efficiently optimize on the bias-variance tradeoff as long as our operator $O$ acts on a small, local subsystem and all $V \in \V$ can be implemented via low-depth circuits. 
Just as in the previous section, we optimize on the bias variance tradeoff by finding the $\widetilde{O}$ that minimizes the cost function~\eqref{eqn:thecostfunctionappx}
given some fixed number of measurements. We can employ a convex optimization algorithm (e.g.~gradient descent) to find the minimum of this cost function. 

The trick for doing this optimization efficiently on a large system will arise from how we parameterize $\widetilde{O}$.  As before, let us suppose that $\V$ is a finite set, and if it is not, we can render a finite set via sub-sampling.
Then we will parameterize $\widetilde{O}$ as follows:
\begin{equation}\label{eqn:parameterizeOIZstrings}
    \widetilde{O} = \sum_{V \in \V} p(V) \sum_{P \in \text{Pauli}_{\mathds{1},Z}(n)} \mathcal{J}^{\widetilde{O} }(V,P) \hspace{1mm}  V^{\dagger}P V\, ,
\end{equation}
 where $\text{Pauli}_{\mathds{1},Z}(n)$ is the set of $n$-qubit Paulis constructed from only $\mathds{1}$'s and $Z$'s.
 $\mathcal{J}^{\widetilde{O}}(V,P)$ provides a  parameterization of $\widetilde{O}$ related to that of $\mathcal{K}^{\widetilde{O}}(V,b)$. Below we will show how the two parameterizations are related, but first we will develop some convenient notation in order to show why this parameterization is so useful.
 
Supposing that we have $n$ qubits, let $L$ denote a subset of the sites $\{1,...,n\}$.  We say that an operator is $L$-local if it has identity elements on all sites except those corresponding to $L$.  A $k$-local operator in the usual sense (i.e.~$k$ is an integer, and not a set) has a set $L$ with $|L| = k$. Let $L_O$ be the subset of sites on which the operator $O$ is supported.  Define the potentially larger set of sites
\begin{equation}
L_{O,\V} = \bigcup_{V \in \V} L_{V O V^\dagger} \, 
\end{equation}
that $O$ can be spread to by $\V$.
When we derive our expression for $\mathcal{J}^{ \widetilde{O} }(V,P)$, we will see that
$\mathcal{J}^{ \widetilde{O} }(V,P) = 0$ when $P$ is not 
$L_{\widetilde{O},\V}\text{-local}\,$.
As such, we have
\begin{equation}
    \widetilde{O} = \sum_{V \in \V} p(V) \sum_{\substack{\{P \in \text{Pauli}_{\mathds{1},Z}(n) \, : \\\, P\text{ is }L_{\widetilde{O},\V}\text{-local}\}}}\mathcal{J}^{ \widetilde{O} }(V,P) \hspace{1mm}  V^{\dagger}P V\,.
\end{equation} 
This new expression for $\widetilde{O}$ reduces our parameterization from $\sim 2^n$ to $\sim 2^{L_{\widetilde{O},\V}}$ parameters, making it it feasible to estimate sufficiently-local operators on large system sizes. 



Finally we will discuss how to relate the $\mathcal{K}^{\widetilde{O}}(V,b)$ parameterization to the $\mathcal{J}^{\widetilde{O}}(V,P)$ parameterization of $\widetilde{O}$ over the all strings in $\text{Pauli}_{\mathds{1},Z}(n)$.  In establishing this connection, we show that the two parameterizations are equivalent. Starting with the $\mathcal{K}^{\widetilde{O}}(V,b)$ expression, we have 
\begin{equation}
    \widetilde{O} = \int_{V \in \V} dV \, p(V) \sum_{b} \mathcal{K}^{\widetilde{O}}(V,b) V^{\dagger}\ketbra{b}{b}V\,.
\end{equation}
The $n$-qubit computational basis state $\ketbra{b}{b}$ can be expanded in terms of strings of $\mathds{1}$ and $Z$ Paulis. For a single qubit, $\ketbra{0}{0} = \frac{1}{2} (\mathds{1} +Z)$ and $\ketbra{1}{1} = \frac{1}{2} (\mathds{1} -Z)$. Expanding this for an $n$-qubit $\ketbra{b}{b}$, we can write
\begin{equation}
    \ketbra{b}{b} = \frac{1}{2^n} \sum_{P \in \text{Pauli}_{\mathds{1},Z}(n)} (-1)^{f(b,P)}P,
\end{equation}
where $f(b,P)$ counts the number of sites that have \emph{both} a $\ket{1}$ in the computational basis state $|b\rangle$ \textit{and} a $Z$ in the Pauli string $P$. Then our equation for $\widetilde{O}$ becomes
\begin{eqnarray}
    \widetilde{O} &=& \int_{V \in \V} dV \, p(V) \sum_{P \in \text{Pauli}_{\mathds{1},Z}(n)} \sum_{b} \mathcal{K}^{\widetilde{O}}(V,b) \frac{1}{2^n} (-1)^{f(b,P)} V^{\dagger}PV \\
    &=& \int_{V \in \V} dV \, p(V) \sum_{P \in \text{Pauli}_{\mathds{1},Z}(n)} \mathcal{J}^{\widetilde{O}}(V,P) V^{\dagger}PV\,,
\end{eqnarray}
where
\begin{align}
\mathcal{J}^{\widetilde{O}}(V,P) = \frac{1}{2^n}\sum_{b} \mathcal{K}(V,b) \,(-1)^{f(b,P)}\,.
\end{align}

\section{\label{app:Adaptivity} Adaptive probability density functions}
Here we discuss how to construct an optimal probability density function $p(V)$ with support $\V$ for estimating $\tr(\rho O)$. We answer the following question: for fixed $\V$ and $O$, what is a good choice of $p(V)$?
We refer to the \textit{optimal} probability density function as the one which minimizes the number of measurements required to estimate $\tr(\rho O)$. At a high level, this corresponds to assigning a probability $p(V)$ to each $V \in \V$ where unitaries that are more ``helpful'' in estimating $\tr(\rho O)$ are assigned a higher probability.
For example, consider some $V_*$ where $\mathcal{K}^{O}(V_*,b)$ is zero for all $b$. The unitary $V_*$ gives us no information about $\tr(\rho O)$, and therefore, measuring with it is wasteful. 
Therefore, we want $V_*$ to have a small probability of being sampled in our randomized measurements. 
We can imagine finding a probability density function over the unitaries $\V$ that most efficiently estimates $\tr(\rho O)$.
Spending less time measuring with wasteful unitiaries, we are more efficient with our measurements and, therefore, estimate $\tr(\rho O)$ with fewer total measurements.

Let us notate our optimal probability density function by $q(V)$.  The first section below explains how to update $\mathcal{K}^O\rightarrow\mathcal{K}_q^O$, where $\mathcal{K}^O$ is defined for some distribution $p(V)$ and $\mathcal{K}_q^O$ for the optimal probability density function $q(V)$. 
Notice that finding the optimal probability density function for estimating $\tr(\rho O)$ depends on both the state $\rho$ and the operator $O$. Since we may not have any a priori information about $\rho$, we will find the probability density function corresponding to the fewest number of measurements for the worst case $\rho$.
The second section below derives and gives intuition for the optimal probability density function when the variance used in Theorem~\ref{thm:shadowtomtheorem} is the max variance over states. 
The third section extends these results and derives the optimal probability density function when estimating multiple operators. For completeness, we also provide a the fourth section, where we derives the optimal probability density function in the special case that $\rho$ is the maximally mixed state.

For ease of notation, let us suppose that $\V$ is a finite set.  However, we note that all of our results generalize to continuous sets $\V$.

\vspace{3mm}
\subsection{How to update $\mathcal{K}^O$ for a new probability density function} \label{sec:UpdateKforPDF}

Consider some 
$\mathcal{K}^O(V,b)$ defined for the operator $O$ and the probability density function $p(V)$. We have our usual formula
\begin{equation}
    O = \int_{V \in \V} dV \, p(V) \sum_b \mathcal{K}^O(V,b) V^\dagger \ketbra{b}{b} V.
 \end{equation}
We will define a new, optimal probability density function $q(V)$ in the following section, and we want to be able to use this probability density function to construct our operator $O$ as we did in the expression above. Therefore, we will need to define a new and suitable $\mathcal{K}$, which we call $\mathcal{K}_q$.  Starting with our usual formula we can derive $\mathcal{K}_q$,
\begin{eqnarray}
    O &=& \sum_{V \in \V} q(V) \frac{p(V)}{q(V)} \sum_b \mathcal{K}^O(V,b) V^\dagger \ketbra{b}{b} V \\
    &=&  \sum_{V \in \V} q(V) \sum_b \mathcal{K}_q(V,b) V^\dagger \ketbra{b}{b} V \\
     &=&  \mathbb{E}_{V \sim q} \sum_b \mathcal{K}_q(V,b) V^\dagger \ketbra{b}{b} V\,, \\
 \end{eqnarray}
where
\begin{equation}
\mathcal{K}_q(V,b) = \mathcal{K}(V,b)\, \frac{p(V)}{q(V)}\,.
\end{equation}
As such, if we sample $V$ with probability $q(V)$, we can estimate $\tr(\rho O)$ with the empirical average
\begin{equation}\label{eqn:appNewEstimator}
    \tr(\rho O) \approx
    \frac{1}{N} \sum_{i=1}^N \mathcal{K}_q (V,b)\,.
\end{equation}

\subsection{Optimal probability density function} \label{sec:PDFmaxvar}

Next we will solve for the $q(V)$ that gives the smallest $\textnormal{Var}_{\textnormal{max}} [\mathcal{K}_{q}(V,b)]$.
Writing out this variance as
\begin{align}\label{eqn:varToMinimize}
    \textnormal{Var}_{\textnormal{max}}[\mathcal{K}_{q}(V,b)] &= \sum_{V \in \V} q(V) \max_b \hspace{1mm} \mathcal{K}_q(V,b)^2 \nonumber \\
    &= \sum_{V \in \V} q(V) \max_b \frac{p(V)^2}{q(V)^2}\,\mathcal{K}(V,b)^2\,,
\end{align}
we see that we would like to minimize
\begin{equation}
    \sum_{V \in \V} q(V) \max_b \frac{p(V)^2}{q(V)^2}\mathcal{K}(V,b)^2 \hspace{1mm} - \lambda\left(1-\sum_{V \in \V} q(V)\right),
\end{equation}
where $\lambda$ is a Lagrange multiplier. 
This expression is simply $\textnormal{Var}_{\textnormal{max}}$ with an additional term enforcing the normalization of $q(V)$. 
Taking the derivative of the above expression with respect to $q(V)$ and setting the derivative to zero, we find that
\begin{equation}\label{eqn:appxOptimalGammarho}
    q(V) = \frac{1}{\sqrt{\lambda}} \, p(V) \max_b  |\mathcal{K}(V,b)|\,.
\end{equation}
In this expression our Lagrange multiplier provides a normalization factor, which we find to be
\begin{equation}\label{eqn:appxLagrangeMultrho}
    \sqrt{\lambda} = \sum_{\{V\}} p(V) \max_b  |\mathcal{K}(V,b)|\,.
\end{equation}
These are the results quoted in the main text, where the normalization factor $\mathcal{N}$ is simply $\mathcal{N} = \sqrt{\lambda}$. 

\subsection{Optimal probability density function for many operators} \label{sec:PDFmanyop}

Our formulas generalize to the setting of multiple operators $\{O_1,O_2,...,O_M\}$. 
In this context we would like to find the $q(V)$ that minimizes
\begin{equation}
    \max_i \textnormal{Var}_{\textnormal{max}}[\mathcal{K}^{O_i}_{q}(V,b)] = \sum_{\{V\}} q(V) \max_{i,b} \hspace{1mm} \mathcal{K}^{O_i}_q(V,b)^2\,. \hspace{1mm}
\end{equation}
An essentially identical derivation as above gives
\begin{equation}
    q(V) = \frac{1}{\mathcal{N}}  \,p(V) \max_{i,b}  |\mathcal{K}^{O_i}(V,b)|\,,
\end{equation}
where
\begin{equation}
    \mathcal{N} = \sum_{V \in \V} p(V) \max_{i,b}  |\mathcal{K}^{O_i}(V,b)|\,.
\end{equation}

\subsection{Probability density function for the infinite temperature state}\label{sec:PDFmaxmixed}

Instead of using the maximum variance over all states $\rho$, if we happen to know information about our state we can tailor our $q$-optimization accordingly.  As a toy example, suppose that $\rho$ is the maximally mixed state.
Then we have
\begin{eqnarray}
    \textnormal{Var}_{V,b} \mathcal{K}_{q}(V,b)
    &=&\sum_{V \in \V} q(V) \sum_b \bra{b}V\rho V^\dagger \ket{b} \frac{p(V)^2}{q(V)^2} \mathcal{K}(V,b)^2 \hspace{1mm} - \tr(\rho O)^2 \\ 
    &=& \frac{1}{2^n} \sum_{V \in \V} q(V) \sum_b \frac{p(V)^2}{q(V)^2} \mathcal{K}(V,b)^2 \hspace{1mm} - \tr\!\left(\frac{1}{2^n} O\right)^2\,.
\end{eqnarray}
In the second line, we have substituted $\rho = \frac{1}{2^n} \mathds{1}$ into the variance $\textnormal{Var}_{V,b\,\sim\,P_\rho}[\mathcal{K}^{O}_{q}(V, b)]$ defined in Appendix~\ref{app:VisibleSpace}. 
We again minimize the variance using a Lagrange multiplier to enforce the normalization and find the optimal $q(V)$ to be 
\begin{equation}\label{eqn:appxOptimalGammaMIXED} 
    q(V) = \frac{1}{\mathcal{N}} \,p(V) \left(\sum_b \mathcal{K}(V,b)^2\right)^{\frac{1}{2}}.
\end{equation}
with
\begin{equation}
    \mathcal{N} = \sum_{V \in \V} p(V) \left(\sum_b \mathcal{K}(V,b)^2\right)^{\frac{1}{2}}.
\end{equation}

\section{\label{app:NumericsLGT} Estimating the energy density of a $U(1)$ lattice gauge theory}

This Appendix provides details on the lattice gauge theory numerics discussed in the Applications section of the main text. 
We consider 2+1 $U(1)$ lattice gauge theory, with the bosonic degrees of freedom truncated to finite dimensions.  The lattice lives in two spatial
dimensions, tiled by spatial triangles with periodic boundary conditions. For the moment, let us begin with the case of one spatial triangle, and we will eventually generalize to a full two-dimensional spatial triangular lattice.
In this setting, our one spatial triangle has three points, and the corresponding Hamiltonian 
\begin{equation}
H_{3\,\text{points}} = \sum_{s} \left(H_\Delta^{(s)} + \sum_{j=1}^3 H_{\leftrightarrow}^{(s,j)}\right)
\end{equation}
has two types of terms: (i) individual triangular plaquette terms $H_{\triangle}$ that represent the magnetic part of the energy,
 \begin{equation}\label{eqn:TriangleH}
    H_{\Delta}^{(s)} :=  - \frac{1}{24g^2}(X_{1,s}X_{2,s}X_{3,s} - Y_{1,s}Y_{2,s}X_{3,s}  - Y_{1,s}X_{2,s}Y_{3,s} - X_{1,s}Y_{2,s}Y_{3,s})\,,
\end{equation}
and (ii) link terms $H_{\leftrightarrow}$ that represent the electric part of the energy and couple adjacent plaquettes,
\begin{equation}\label{eqn:LinkH}
    H_{\leftrightarrow}^{(s,j)} :=  \frac{g^2}{3}Z_{j,s}Z_{j,s+1}  +  \frac{\alpha}{12g^2}(X_{j,s}X_{j,s+1} + Y_{j,s}Y_{j,s+1})\,.
\end{equation}
Here $s$ labels the bosonic modes. We can visualize the one spatial triangle system as a stack of triangular plaquettes each labelled by $s$, and adjacent triangles are coupled to one another by edge degrees of freedom indexed by $j$.  More specifically, the index $j$ enumerates over the pairs of adjacent vertices between two triangular plaquettes. See Supplementary Figure~\ref{fig:LGTfiglattice} for a depiction.  If $s_{\text{max}}$ is the maximum number of bosonic modes, then $H_{3\,\text{points}}$ acts on the Hilbert space $\mathcal{H} \simeq (\mathbb{C}^2)^{\otimes (3 s_{\text{max}})}$. Moreover, we assume periodic boundary conditions between the top plaquette $s=s_\textnormal{max}$ and the bottom plaquette $s=s_1$.

Now suppose that we have a triangular lattice $\mathscr{T} = \{(a_i,b_i,c_i)\}_i$, where $(a_i,b_i,c_i)$ label vertices of a single spatial triangle and the index $i$ enumerates over the spatial triangles.  To each spatial triangle, we associate a stack of triangular plaquettes of height $s_{\text{max}}$, and a corresponding Hilbert space $\mathcal{H}_{(a_i, b_i, c_i)} \simeq (\mathbb{C}^2)^{\otimes (3 s_{\text{max}})}$.  Then $H_{3\,\text{points}}^{(a_i,b_i,c_i)}$ is the same as $H_{3\,\text{points}}$, where the vertices in $H_{3\,\text{points}}$ are ${(a_i,b_i,c_i)}$. The total Hamiltonian for the triangular lattice $\mathscr{T}$ is
\begin{equation}
H = \sum_{(a_i, b_i, c_i) \in \mathscr{T}} H_{3\,\text{points}}^{(a_i,b_i,c_i)}\,.
\end{equation}
Since we require this spatial triangular lattice to have periodic boundary conditions, it lives on a torus, and thus for $|\mathscr{T}|$ spatial triangles and $s_\textnormal{max}$ bosonic modes, our system has $\frac{3}{2}|\mathscr{T}|s_\textnormal{max}$ total qubits. Note that in order to enforce periodic boundary conditions, we assume an even number of spatial triangles $|\mathscr{T}|$.

\begin{figure*}[ht]
\centering
\includegraphics[scale = 0.33]{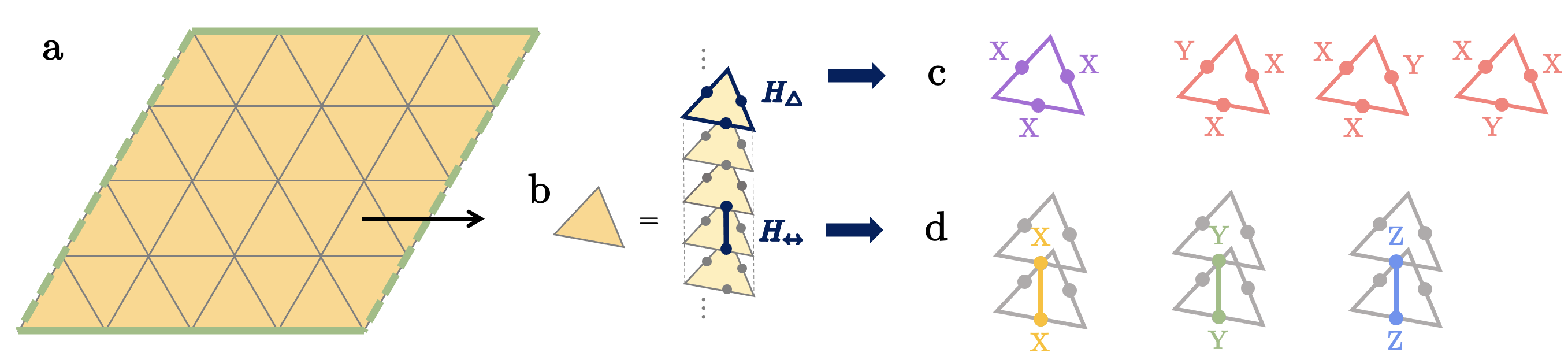}
\caption{\emph{Triangular lattice representing a $2+1$ $U(1)$ lattice gauge theory.} (a) The spatial triangular lattice lives in two spatial dimensions, and we assume periodic boundary conditions (connect the solid and dashed green lines). (b) A single spatial triangle is made up of a stack of triangular plaquettes, which also have periodic boundary conditions. The Hamiltonian on this stacked triangular lattice has two types of terms, $H_\triangle$ and $H_\leftrightarrow$, each of which live in the global $SU(2)$ visible space. We can express both of these terms using the basis defined in Proposition \ref{prop:basisGlobalSu2VisInvis}, and here we depict (b) $H_\triangle$  and (c) $H_\leftrightarrow$ terms' decomposition into visible basis elements. Each visible basis element is represented by a different color. }
\label{fig:LGTfiglattice}
\end{figure*}

\begin{observation}
    The Hamiltonian $H$ lives entirely in the global $SU(2)$ visible space.
\end{observation}
\begin{proof}
The Hamiltonian $H$ is made up of link terms $H_{\leftrightarrow}$ and plaquette terms $H_{\Delta}$ acting on various sets of qubits. We will show that both types of terms are in $\text{\textsf{VisibleSpace}}(\text{Global }SU(2))$, and so the total Hamiltonian is as well.
The link terms $H_{\leftrightarrow}$ are in $\text{\textsf{VisibleSpace}}(\text{Global }SU(2))$ because they are a linear combination of $2$-local $B_S$ basis elements (see~\eqref{eqn:defineBi} in Appendix~\ref{app:globalSU2}) defined by fixed-$\mathds{1}$ permutation-invariant sets with $(n_X,n_Y,n_Z) = (2,0,0)$, $(n_X,n_Y,n_Z) = (0,2,0)$, and $(n_X,n_Y,n_Z) = (0,0,2)$. Similarly, the plaquette terms $H_{\Delta}^{(s)}$ live in $\text{\textsf{VisibleSpace}}(\text{Global }SU(2))$ because they are a linear combination of two $3$-local $B_i$ elements corresponding to the fixed-$\mathds{1}$ permutation-invariant sets $(n_X,n_Y,n_Z) = (3,0,0)$ and $(n_X,n_Y,n_Z) = (1,2,0)$.
\end{proof}

Since the Hamiltonian $H$ lives entirely in the global $SU(2)$ visible space, we can estimate its energy density by utilizing our $\mathcal{K}(V,b)$ formalism with \textit{only} global $SU(2)$ control. In order to obtain the energy density, we estimate the energy of all local plaquette and link terms, $H_{\triangle}$ and $H_{\leftrightarrow}$.

For this problem of estimating the energy density, we wish to compare our probability density function adaptivity and bias-variance optimizing methods to ordinary classical shadow tomography, and so we will set
$\mathcal{K}^{\leftrightarrow}=\mathcal{K}^{\leftrightarrow}_\textnormal{CS}$ and $\mathcal{K}^{\triangle}=\mathcal{K}^{\triangle}_\textnormal{CS}$ for our link and plaquette terms, respectively. These are the ``classical shadow versions''  of $\mathcal{K}$ (Definition \ref{def:shadowversionofK} in Appendix~\ref{app:VisibleSpace}), and using these $\mathcal{K}$-representations of our operators, we will optimize on the bias-variance tradeoff and adapt the probability density function. Rather than working with the full $SU(2)$ group in our simulations, we sample a subset of $25$ unitaries $\{V'\} = \V'$ from $SU(2)$. It is on this set $\V'$ that we define $\mathcal{K}^{\leftrightarrow}_\textnormal{CS}(V',b)$ and $\mathcal{K}^{\triangle}_\textnormal{CS}(V',b)$.

Note that for an operator $O$, the classical shadow kernel $\mathcal{K}_{\textnormal{CS}}^O$ assumes that the visible space contains all unitaries in $SU(2)$, whereas in our simulations we utilize a discrete subset $\V' = \{V'\}$ of $25$ unitaries sampled from $SU(2)$. As a result, we cannot construct $\mathcal{K}_\textnormal{CS}^O$ exactly as formulated in~\eqref{eqn:shadowVersionKVb}, and so we come up with a modified expression for $\mathcal{K}_\textnormal{CS}^O$. In particular, define the kernel $K$ by
\begin{equation}\label{eqn:LGTapproxshadowK}
    \Vec{K} = M^{-1} \Vec{O}
    \hspace{5mm}\textnormal{ where }
    \hspace{3mm}
    M=
    \begin{bmatrix} 
	| &   &   \\
	V'^\dagger\ketbra{b}{b}V' & \cdots & \\
	| &   &    \\
	\end{bmatrix}\,.
\end{equation}
$\Vec{K}$ is a length $25 \times 2^n$ vector, where each entry $K_{(V',b)}$ corresponds to some $(V',b)$ pair such that $V' \in \V'$ and $b \in \{0,1\}^n$. 
$\Vec{O}$ is the flattened version of the $n$-qubit operator $O$, meaning its $2^n \times 2^n$ entries are now represented in a length-$4^n$ vector.
Similarly, the columns of the matrix $M$ are the flattened operators $V'^\dagger \ketbra{b}{b} V'$, where the ordering of $V'^\dagger \ketbra{b}{b} V'$s is the same as the indexing of $(V',b)$ pairs in $\Vec{k}$. 
As a result, since each entry $K_{(V',b)}$ of $\Vec{K}$ corresponds to some pairing of $V'$ and $b$, we set 
\begin{equation}\label{eqn:E4LGTappendix_ktoK}
    \mathcal{K}^O(V',b) = K_{(V',b)}\,.
\end{equation}
This expression has an inverse channel similar to the inverse measurement channel $\mathcal{M}^{-1}_{SU(2)}$ and retains a similar structure to the kernel $\mathcal{K}_\textnormal{CS}^O$ defined over all of $SU(2)$.
Using the version of $\mathcal{K}^O$ in~\eqref{eqn:E4LGTappendix_ktoK} we will showcase our techniques of adapting the probability density function and optimizing on the bias-variance tradeoff.

In the two Appendix subsections below, we discuss the numerics for Figures 3(b) and 3(c).
In each we consider the lattice gauge theory described above with our constants $g = \alpha = 1$ and work with quantities that are a function of variance $\textnormal{Var}_{V',b\,\sim\,P_\rho}[\mathcal{K}^O(V', b)]$. 
We calculate the following quantities in 3(b) and 3(c) respectively: estimation error~\eqref{eqn:thecostfunctionappx} and the number of measurements (Theorem \ref{thm:shadowtomtheorem}).
Recall that the variance depends on $\rho$ because $P_\rho(V',b) = p(V') \bra{b}V'\rho V'^\dagger\ket{b} $, and since we do not know the state $\rho$, we upper bound variance with
\begin{equation}\label{eqn:varupperboundLGT}
    \textnormal{Var}_{V',b\,\sim\,P_\rho} \mathcal{K}^O(V',b) \leq \sum_{V' \in \V'} p(V') \max_b \mathcal{K}^O(V',b)^2\,.
\end{equation}
Therefore, the data represented in both plots represents an upper bound on the estimation error (Figure 3(b)) and number of measurements (Figure 3(c)). 

\vspace{3mm}

\subsection{Showcasing the bias-variance tradeoff}\label{appx:LGT_fig3b}

In order to construct the energy density profile of our lattice gauge theory, we have to estimate $\expval{H_\leftrightarrow}$ on many pairs of qubits. 
Figure 3(b) showcases the utility of bias-variance tradeoffs for this task: estimating $\expval{\widetilde{H}_\leftrightarrow}$ instead of $\expval{H_\leftrightarrow}$ will reduce our estimation error when $\widetilde{H}_\leftrightarrow$ is close to $H_\leftrightarrow$ but has much smaller variance.
As discussed in Appendix~\ref{app:OptimizingBiasVar}, the estimation error has a bias term and a variance term. Using $\widetilde{H}_\leftrightarrow$ rather than $H_\leftrightarrow$ allows us to significantly decrease the variance term by slightly increasing the bias term. 
However, at some point the operator $\widetilde{H}_\leftrightarrow$ becomes so biased that the bias outweighs the gains made by reducing the variance.
Hence, we have a bias-variance `tradeoff'. 
This behavior can be observed in Figure 3(b)'s bowl shape. 

In the present subsection, we will discuss how we find the biased operators $\widetilde{H}_\leftrightarrow$ represented as points in Figure 3(b). 
Finding a biased operator $\widetilde{H}_\leftrightarrow$ corresponds to finding a kernel that represents $\widetilde{H}_\leftrightarrow$. 
Consider~\eqref{eqn:E4LGTappendix_ktoK} above, where we set $O = H_\leftrightarrow$. 
To obtain this kernel,
one solves for the \textit{unbiased} $\Vec{K}^\leftrightarrow$ that satisfies $M \Vec{K}^\leftrightarrow = \Vec{H}_\leftrightarrow$, where $\Vec{H}_\leftrightarrow$ is the flattened version of $H_\leftrightarrow$, and sets $\mathcal{K}^\leftrightarrow(V',b) = K^\leftrightarrow_{(V',b)}$. 
Here, we will solve for a \textit{biased} vector $\Vec{K}^\leftrightarrow (\lambda)$, where $\lambda$ parameterizes the bias, and set $\mathcal{K}^{\leftrightarrow}_\lambda(V',b)$ to be the $(V',b)$ entry of $\Vec{K}^\leftrightarrow (\lambda)$.
We obtain $\Vec{K}^\leftrightarrow (\lambda)$ by solving the equation
\begin{equation}\label{eqn:biasedksolveLS}
    M^\dagger M \cdot \Vec{K}^\leftrightarrow + \lambda \mathds{1}= M^\dagger \cdot \Vec{H}_\leftrightarrow\,.
\end{equation}
First, notice that the \textit{unbiased} $\Vec{K}^\leftrightarrow$ that is a solution to $M \Vec{K}^\leftrightarrow = \Vec{H}_\leftrightarrow$ will also be a solution to~\eqref{eqn:biasedksolveLS} when $\lambda = 0$. Therefore, $\Vec{K}^\leftrightarrow (\lambda=0)$ is unbiased and defined as 
\begin{equation}\label{eqn:zerolambdakvec}
    \Vec{K}^\leftrightarrow (0) = M^{-1}\Vec{H}_\leftrightarrow\,.
\end{equation}
Next, notice that as $\lambda$ grows, our solution $\Vec{K}^\leftrightarrow (\lambda)$ corresponds to a more biased $\widetilde{H}_\leftrightarrow$.
We generated the points in Figure 3(b) by solving for $\Vec{K}^\leftrightarrow(\lambda)$ for many different $\lambda$'s; each point corresponds to a different $\lambda$, with increasing $\lambda$ corresponding to increasing bias.

Sine our goal is to optimize the bias-variance tradeoff,  
we should bias our operator $H_\leftrightarrow \rightarrow \widetilde{H}_\leftrightarrow$ in a way that reduces the variance. Therefore, given the set of all vectors that solve~\eqref{eqn:biasedksolveLS}, we choose the vector with the smallest variance. 
That is, we choose the $\Vec{K}^\leftrightarrow(\lambda)$  with the smallest variance within this degenerate solution space using Numpy's least squares regression \cite{harris2020array}. The least squares function (numpy.linalg.lstsq) returns the solution to~\eqref{eqn:biasedksolveLS} that has the smallest 2-norm, and
assuming we have the maximally-mixed state, the 2-norm $\|\Vec{K}^\leftrightarrow(\lambda)\|$ is an upper bound on the variance. See Appendix~\ref{app:OptimizingBiasVar} for comments on why we use the heuristic of a maximally-mixed state. Our implementation enables us to choose the biased operator with the smallest variance and, therefore, to optimize the bias-variance tradeoff.
Before optimizing the tradoeff, we make sure our $25$ unitary subset of $SU(2)$ allows us to construct $H_\leftrightarrow$ exactly. In other words, we guarantee that there exists a vector $\Vec{K}^\leftrightarrow$ that satisfies~\eqref{eqn:biasedksolveLS} for $\lambda = 0$. As a result, there will exist a nontrivial solution space for every $\lambda$.

Once we have the $\Vec{K}^\leftrightarrow(\lambda)$ vectors, whose $(V',b)$ entries correspond to $\mathcal{K}^{\leftrightarrow}_\lambda(V', b)$, we plot their estimation error in Figure 3(b). As mentioned earlier in this Appendix, the estimation error we plot is an upper bound. For $N$ measurements we plot 
\begin{equation}
    \textnormal{Error}(\mathcal{K}_\lambda^\leftrightarrow) = \|H_\leftrightarrow-\widetilde{H}(\mathcal{K}_\lambda^\leftrightarrow)\|_{\infty}
    +
    \sqrt{\frac{2}{N}\, \log\left(\frac{1}{2\delta}\right)
    \sum_{V'} p(V') \max_b \mathcal{K}_\lambda^{\leftrightarrow}(V', b)^2}\,,
\end{equation}
where $\widetilde{H}(\mathcal{K}_\lambda^\leftrightarrow)$ represents the biased operator $\widetilde{H}^\leftrightarrow$ constructed with $\mathcal{K}_\lambda^\leftrightarrow$ via equation \ref{eqn:OvisKVb2}.
Here $\textnormal{Error}(\mathcal{K}_\lambda^\leftrightarrow)$ is the estimation error from Appendix~\ref{app:OptimizingBiasVar} (see~\eqref{eqn:thecostfunctionappx}) with the variance upper bounded by~\eqref{eqn:varupperboundLGT}. 
We can use the expression from Appendix~\ref{app:OptimizingBiasVar} because the $\frac{1}{3}Q_\leftrightarrow \epsilon$ term in~\eqref{eqn:thm1NumMmts} is small relative to the variance in this case.

\subsection{Reducing the required number of measurements}\label{appx:LGT_fig3c}

Our goal is to estimate the energy density of our $U(1)$ lattice gauge theory by estimating the expectation value of each $H_\triangle$ and $H_\leftrightarrow$ term throughout our lattice. Using~\eqref{eqn:thm1NumMmts} in Theorem 1 (Appendix~\ref{app:VisibleSpace}), if we want to estimate $M$ terms to precision $\epsilon$ with probability at least $1-\delta$, we require a number of measurements
\begin{equation}\label{eqn:LGTnummmtsExpression}
    N = 2  \log\!\left(\frac{M}{2\delta}\right) \max_{i\in \{\triangle, \leftrightarrow\}} \frac{\textnormal{Var}_{\textnormal{max}}[\mathcal{K}^{i}] + \frac{1}{3}Q_i(\epsilon - \|H_i - \widetilde{H}(\mathcal{K}^i)\|_\infty)}{(\epsilon - \|H_i - \widetilde{H}(\mathcal{K}^i)\|_\infty)^2}\,.
\end{equation}
In our numerics, we calculate  $\mathcal{K}^\leftrightarrow$ and $\mathcal{K}^\triangle$ for four different cases: (1) ordinary classical shadow tomography, (2) solely optimizing the bias variance tradeoff, (3) solely adapting the probability density function, and (4) optimizing the bias variance tradeoff \textit{and} adapting the probability density function. For each of these cases, the $\max$ in~\eqref{eqn:LGTnummmtsExpression} was dominated by $H_\leftrightarrow$. Therefore, the link term sets the total number of measurements $N$ needed to estimate the energy density:
\begin{equation}\label{eqn:nummmtsweplotLGT}
    N \leq  \frac{2  \log \!\left(\frac{M}{2\delta}\right)}{(\epsilon - \|H_\leftrightarrow - \widetilde{H}(\mathcal{K}^\leftrightarrow)\|_\infty)^2} \left[    \sum_{V' \in \V'} p(V') \max_b \mathcal{K}^{\leftrightarrow}(V', b)^2    + \frac{1}{3}Q_\leftrightarrow(\epsilon - \|H_\leftrightarrow - \widetilde{H}(\mathcal{K}^\leftrightarrow)\|_\infty)\right]\,.
\end{equation}
In this equation, $\textnormal{Var}_{\textnormal{max}}$ is upper bounded using the inequality~\eqref{eqn:varupperboundLGT}, and we set $\epsilon = \delta = 0.1$.  
Following Theorem \ref{thm:shadowtomtheorem}, we also set $Q_\leftrightarrow$ in~\eqref{eqn:nummmtsweplotLGT} to be 
\begin{equation}
    Q_\leftrightarrow = \max_{V',b} |\mathcal{K}^{\leftrightarrow}(V', b)|\,.
\end{equation}
The number of terms $M$ we need to estimate is the total number of plaquette terms $H_{\triangle}$ and link terms $H_{\leftrightarrow}$ in our Hamiltonian. This depends on our system size: increasing system size corresponds to adding more \textit{spatial} triangles. We assume the number of stacked plaquettes $s_\textnormal{max}$  (i.e. the range of bosonic modes~\cite{brower2019lattice}) is fixed, and we choose $s_\textnormal{max}=5$. 
A lattice with $|\mathscr{T}|$ total spatial triangles has $ \frac{5}{2} |\mathscr{T}| s_{\text{max}}$ total terms of type $H_{\triangle}$ and $H_{\leftrightarrow}$. In our numerics we let the number of spatial triangles run from $|\mathscr{T}| = 2$ to $|\mathscr{T}| = 30$, and recall that since we enforce periodic boundary conditions, we only have an even number of spatial triangles $|\mathscr{T}|$.

In Figure 3(c), we plot the number of measurements (i.e.~\eqref{eqn:nummmtsweplotLGT}) for the four cases just described, and for each we obtain $\mathcal{K}^{\leftrightarrow}$ by solving for $\Vec{K}_\leftrightarrow$.  Recall that $\mathcal{K}^{\leftrightarrow}(V',b)$ is the entry in $\Vec{K}_\leftrightarrow$ corresponding to $(V',b)$.
Now, let us discuss how we obtain $\Vec{K}_\leftrightarrow$ for each case. 

\begin{enumerate}
    \item The bare shadow $\Vec{K}_\leftrightarrow$ is simply the zero-bias $\Vec{K}^\leftrightarrow(\lambda=0)$ vector we found in~\eqref{eqn:zerolambdakvec} of this Appendix.
    \item The bias-variance optimized  $\Vec{K}_\leftrightarrow$ is the vector from Figure 3(b) which yielded the smallest estimation error. 
    \item The probability density function-adapted  $\Vec{K}_\leftrightarrow$ is obtained by performing the adaptivity procedure on the bare shadow $\Vec{K}_\leftrightarrow(\lambda=0)$. The adaptivity procedure we use is given in~\eqref{eqn:appxOptimalGammarho} in Appendix~\ref{app:Adaptivity}. 
    \item Finally, the $\Vec{K}_\leftrightarrow$ that is both bias-variance optimized and probability density function adapted is created by performing the adaptivity procedure on the bias-variance optimized $\Vec{K}_\leftrightarrow$.
\end{enumerate}


\section{\label{app:NumericsPhases} Classifying topological phases with global $SU(2)$ control}

In the Applications section of the main text, we performed our learning protocol with global $SU(2)$ control and, equipped with kernel principal component analysis (PCA), were able to distinguish the toric code topological phase from the trivial phase in a 200-qubit system. See Figure~\ref{fig:Figure3} in the main text. This Appendix first describes our results, their significance, and how they fit into the broader literature (Subsection \ref{appx:phases_significance}). Then, in the following subsections, we describe in detail the numerical simulations. Subsection \ref{appx:phases_globalsu2data} describes how we construct the states in our two phases and predict their expectation values in the global $SU(2)$ visible space, and Subsection \ref{appx:phases_kernelconstructionPCA} describes how to use these expectation values with kernel PCA to classify topological phases. 
All code is available at \url{https://github.com/katherinevankirk/hardware-efficient-learning}.

\subsection{Results and their significance}\label{appx:phases_significance}

Given a gapped topological phase and a state in the phase, we can obtain other states in the same phase using low-depth geometrically-local quantum circuits \cite{chen2010local, wen2017colloquium, zeng2019quantum}.  In order to generate random states belonging to these our two desired topological phases, we utilize low-depth geometrically-local random Clifford unitaries \cite{huang2022provably}, which are efficient to simulate classically.  We apply these circuits to a product state (to generate a random state in the trivial phase) and to Kitaev’s toric code state with code distance 10 \cite{kitaev2003fault} (to generate a random state in the topological phase).

Once we have a collection of states for each of the two topological phases, we want to obtain information about them and then, using this newly obtained information, learn to classify their phases. Ref.~\cite{huang2022provably} performed tomographically-complete classical shadow tomography to construct a classical shadow $\hat{\rho}$ for each state $\rho$. The classical shadows enable access to (approximations of) nonlinear functions of $\rho$, which can act as order parameters for the topological phases and thus allow one to distinguish between the phases.

We explore how global $SU(2)$ control can be used to distinguish between topological phases.  This is particularly novel since global $SU(2)$ is far from being tomographically complete.
Our numerics show that global $SU(2)$ control collects enough information about the state $\rho$ to distinguish between the trivial phase and the toric code phase. In other words, there exists a nonlinear function on the visible space of global $SU(2)$ for classifying trivial and topologically-ordered phases.  

These results are of interest because they indicate that restricted controls can still access non-trivial order parameters. Furthermore, even for devices with tomographically complete control, our results indicate that one \textit{does not} need to utilize the entire operator space. One can estimate a smaller number $M$ of operators and still distinguish between the toric code topological phase and the trivial phase. Since our protocol's number of measurements scales as $O(\log M)$ with the number of quantities $M$ we wants to estimate, our results suggest one needs fewer measurements.

\subsection{Extracting information about states using global $SU(2)$}\label{appx:phases_globalsu2data}

We start with a square lattice of size $L \times L$ with periodic boundary conditions, where we choose $L=10$. Each edge of the lattice is assigned one qubit.  As such this lattice supports a $n=200$ qubit system, which we prepare in one of two states. We can prepare the $200$-qubit system in a product state; this state lives in the trivial phase. Or, we can prepare the system in Kitaev’s toric code state with code distance 10 \cite{kitaev2003fault}, which is in the toric code topological phase. Since we are consider two states in distinct gapped phases, we can generate other states in the same, respective phases by applying a low-depth, geometrically-local quantum circuits~\cite{chen2010local, wen2017colloquium, zeng2019quantum}. Therefore, to create $50$ random states in each phase, we create $50$ copies of each type of state (i.e.~$100$ states total) and apply a different random Clifford circuit of depth $d$ to each. 
\vspace{3mm}

\begin{figure*}[ht]
\centering
\includegraphics[scale = 0.33]{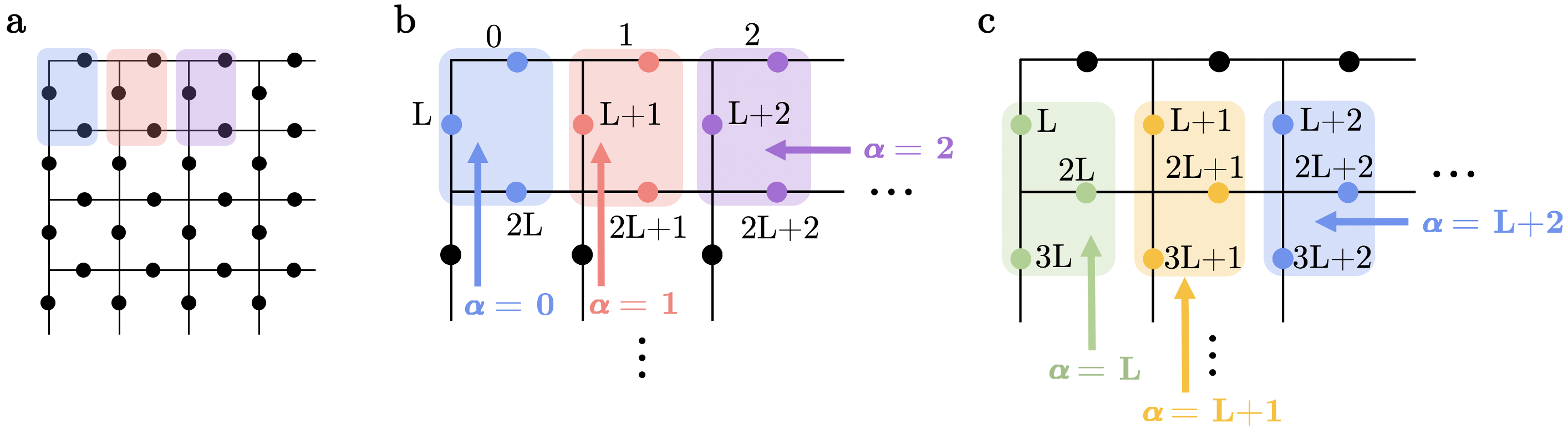}
\caption{\emph{Enumerating over all 3-qubit patches.} (a) In a 200 qubit system on an $L \times L$ square lattice ($L=10$), we consider local, $3$-qubit patches with the index $\alpha \in [0,2L^2 -2L)$, where $\alpha$ corresponds to the index of the top-most qubit in the patch. The blue, red, and purple zones represent the first three patches. (b) We label the qubits in the first three patches. (c) Next, we consider the first three patches in the second row. These patches take a slightly different form, but still the patch $\alpha$ contains qubits $\alpha$, $\alpha + L$, and $\alpha + 2L$. The third row of patches will again look like the first row.}
\label{fig:FigurePhases1}
\end{figure*}

For each state $s$ in this set of $100$ states, we want to estimate observables in the global $SU(2)$ visible space. 
Rather than estimating every global $SU(2)$ visible space observable on the $200$-qubit system, we will estimate all visible space observables on local, $3$-qubit patches. This procedure is more efficient as the (worst case) $3$-local visible space operators have lower sample complexity than the (worst case) $n$-local ones.
We enumerate over the local, $3$-qubit patches with the index $\alpha \in [0,2L^2 -2L)$, where $\alpha$ corresponds to the index of the topmost qubit in the patch. See Supplementary Figure 2. The patch $\alpha = x$ contains qubits $x$, $x + L$, and $x + 2L$. 
For each state $s$, we construct local reduced density matrices on these patches by performing random Pauli classical shadow tomography \cite{huang2020predicting} with $N_{\text{RP}}=10,000$ measurements; we repurposed code from \cite{huang2022provably}. As a result, we obtain $3$-site, high-fidelity reduced density matrices $\hat{\rho}^{(s)}_\alpha \approx \rho_\alpha^{(s)}$ for all $\alpha \in [0,2L^2 -2L)$. See Supplementary Figure 3 for a flow chart of this process.

\begin{figure*}[ht]
\centering
\includegraphics[scale = 0.33]{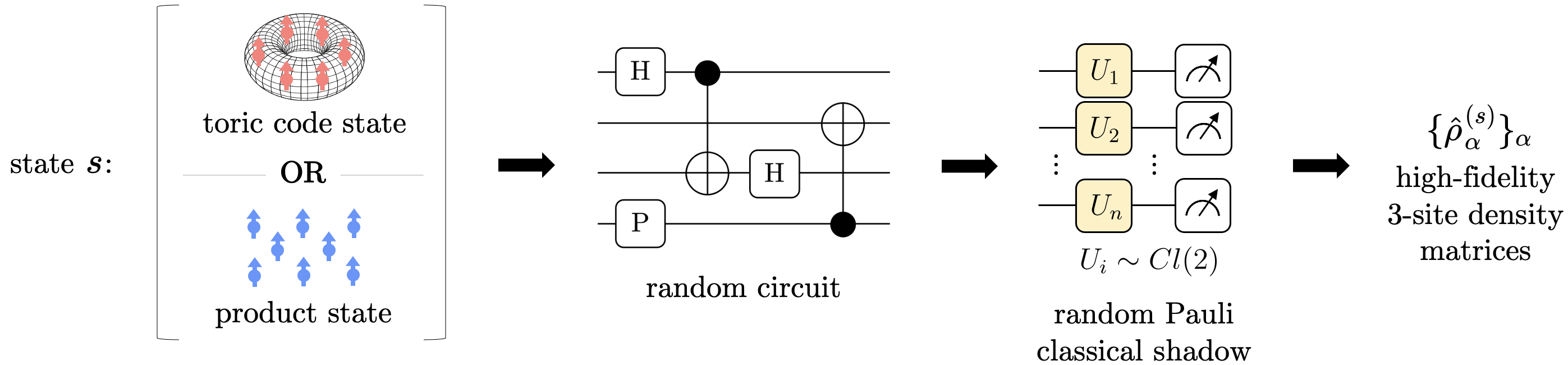}
\caption{\emph{Flow chart for creating $3$-site, high-fidelity reduced density matrices on all patches $\alpha$.} Each state $s$ is prepared in either Kitaev's toric code state \cite{kitaev2003fault} or a product state. Then, to create another state in the topological phase or trivial phase, respectively, we apply a low-depth
geometrically-local quantum circuit \cite{chen2010local, wen2017colloquium, zeng2019quantum}. By performing random Pauli classical shadow tomography \cite{huang2020predicting} on this state, we can reconstruct high-fidelity reduced density matrices $\hat{\rho}^{(s)}_\alpha$ for all patches $\alpha$.}
\label{fig:FigurePhases2}
\end{figure*}

On these newly obtained, high-fidelity reduced density matrices $\hat{\rho}^{(s)}_\alpha$, we perform $N_{SU(2)}=1000$ randomized global $SU(2)$ measurements. We will use this data to estimate our global $SU(2)$ visible space observables. Our resulting measurement data takes the form $\{(V_m,b_m)\}_{m=1}^{N_{SU(2)}}$. For each patch $\alpha$ on state $s$, we estimate all visible space expectation values $o_{\alpha,S}^{(s)} \approx \tr(\hat{\rho}^{(s)}_\alpha B_S)$ with the estimator
\begin{equation} \label{eqn:phaseestimatorvisspace}
    o_{\alpha,S}^{(s)} = \frac{1}{N_{SU(2)}} \sum_{m=1}^{N_{SU(2)}} \mathcal{K}_{\text{CS}}^{B_S}(V_m,b_m)\,.
\end{equation}
The global $SU(2)$ visible space basis $\{B_S\}_S$ was defined in Appendix \ref{app:globalSU2}, and $\mathcal{K}_{\text{CS}}$ was defined in Appendix \ref{sec:KVbAndClassicalShadows}. This $\mathcal{K}_{\text{CS}}$ requires the inverse measurement channel $\mathcal{M}^{-1}_{SU(2)}$. Since the $\mathcal{M}_{SU(2)}$ superoperator matrix (defined in Proposition \ref{prop:globalsu2mmtchannel}) is block diagonal, we obtain $\mathcal{M}^{-1}_{SU(2)}$ by inverting each block.

\subsection{Kernel PCA}\label{appx:phases_kernelconstructionPCA}

In the previous subsection, we described how we estimate the global $SU(2)$ visible space expectation values of local patches $\alpha$ in the state $\rho^{(s)}$. Recall we are working with $100$ states $\rho^{(s)}$, half of which are in the trivial phase and half in the toric code topological phase, and we have $2L^2-2L$ $3$-qubit patches $\alpha$. In this subsection we describe how we use this global $SU(2)$ visible space information to classify the phase of each state. 

For each patch $\alpha$ of the state $\rho^{(s)}$, we can construct a vector $\Vec{o}^{\hspace{1mm}(s)}_\alpha$ of the estimated visible space expectation values,  
\begin{equation}
    \Vec{o}^{\hspace{1mm}(s)}_\alpha = \left[\hspace{1mm} o_{\alpha,S_1}^{(s)}, \hspace{1mm} o_{\alpha,S_2}^{(s)}, \hspace{1mm} o_{\alpha,S_3}^{(s)}, \hspace{1mm} \cdots \hspace{1mm} \right].
\end{equation}
This vector will be used to construct a kernel $K$, a matrix where each entry $K_{s_1,s_2}$ represents the correlation between states $s_1$ and $s_2$. For our setup of $100$ states across the two phases, we will construct a $100 \times 100$ matrix, where each entry $K_{s_1,s_2}$ will be the average (exponentiated) inner product between all the patches $\alpha \in [0,2L^2 -2L)$.  In particular,
\begin{equation} \label{eqn:kernelfunctions1s2}
    K_{s_1,s_2} = \frac{1}{2L^2 - 2L} \sum_{\alpha} e^{\lambda (\Vec{o}^{\hspace{1mm}(s_1)}_\alpha \cdot \Vec{o}^{\hspace{1mm}(s_2)}_\alpha)}
\end{equation}
There is one free parameter $\lambda$, which is the same across all kernel matrix entries. Choosing the `best' $\lambda$ is a form of hyperparameter tuning. In practice, we test a variety of $\lambda$'s and in Figure~\ref{fig:Figure3} we used the following value which gave the best results: 
\begin{equation}
    \frac{1}{\lambda} = \frac{3}{100(2L^2 - 2L)} \sum_{s,\alpha} \Vec{o}^{\hspace{1mm}(s_1)}_\alpha \cdot \Vec{o}^{\hspace{1mm}(s_2)}_\alpha.
\end{equation}
Ultimately, this kernel matrix (or a projected version of it) is be fed into PCA, which is what is referred to as kernel PCA. To highlight the significance of this methodology and the format of~\eqref{eqn:kernelfunctions1s2}, we will compare what we are doing to `normal' PCA. In normal PCA, one would feed in a covariance matrix \cite{mohri2012foundations} where the entry $K_{s_1,s_2}$ would be the inner product of the data vectors of states $s_1$ and $s_2$.  However, since our kernel exponentiates the patch inner products, we are feeding nonlinear functions of our global $SU(2)$ data into PCA. This allows us to extract nonlinear features in the global $SU(2)$ data.  

Once we use our global $SU(2)$ randomized measurement data to construct our kernel, we renormalize the entries. This renormalization can also be thought of as a projection onto the unit sphere in the high dimensional feature space. Our new kernel entries take the form $K'_{s_1,s_2}$ where
\begin{equation}
    K'_{s_1,s_2} = \frac{K_{s_1,s_2}}{\sqrt{K_{s_1,s_1} K_{s_2,s_2}}}\,.
\end{equation}
While this projection yields the best results for our purposes, in general this renormalization is not required. This projection simply helps rearrange the data in the high-dimensional feature space, putting it in a form that PCA has an easier time clustering.

The renormalized kernel $K'$ is then fed into PCA; for Figure~\ref{fig:Figure3}(b) we used the scikit-learn implementation of PCA. PCA finds the axes (the ``principal components'') in the high dimensional feature space that maximize the variance of the data \cite{mohri2012foundations}. The $1$st principal component is the axis along which the data has largest variance.  For phase classification we only care about the first principal component and project the data along this axis. This one-dimensional projection of the feature space, plotted in Figure~\ref{fig:Figure3}(b), creates a low dimensional representation of each quantum state. And in this low-dimensional representation, we find that the quantum states are clustered according to their topological phase.

Moreover, we notice that as the depth of the applied random circuit increases, the low dimensional projections of the states are less clustered. While the low-depth random circuits are meant to generate additional states in the same phase~\cite{zeng2019quantum}, at a certain point, the depth becomes large enough that a nontrivial portion of the total $200$-qubit system becomes correlated. We begin to observe this for states generated via depth $d = 4$ random circuits. Since individual qubits become correlated with a substantial portion of the full system, the separation in Figure~\ref{fig:Figure3}(b) is not as strong.

\end{appendix}

\end{document}